\documentclass[preprint,12pt]{elsarticle}

\usepackage{amssymb}
\usepackage{amsmath}

\usepackage{mathtools}
\usepackage{amsfonts}
\usepackage{amsthm}
\usepackage{algorithm}
\usepackage{graphicx}
\usepackage{fancybox}
\usepackage{algorithmic}
\usepackage{xcolor}
\usepackage{booktabs}
\usepackage{multirow}
\usepackage{url}

\newcommand{\revision}[1]{{#1}}
\DeclareMathOperator*{\argmax}{arg\,max}

\DeclarePairedDelimiter\abs{\lvert}{\rvert}
\newtheorem{Theorem}{Theorem}
\newtheorem{Lemma}{Lemma}
\newtheorem{Corollary}{Corollary}
\newtheorem{Definition}{Definition}
\newtheorem{Example}{Example}
\newcommand{\EE}{\mathbb{E}}
\newcommand{\NN}{\mathbb{N}}
\newcommand{\ZZ}{\mathbb{Z}}
\newcommand{\calC}{\mathcal{C}}
\newcommand{\calP}{\mathcal{P}}
\newcommand{\calS}{\mathcal{S}}
\newcommand{\calT}{\mathcal{T}}
\newcommand{\BR}{\mathsf{BR}}
\newcommand{\LP}{\mathbf{\textsf{LP}}}
\newcommand{\xx}{\mathbf{x}}
\newcommand{\cc}{\mathbf{c}}

\journal{Artificial Intelligence}

\begin{document}

\begin{frontmatter}



\title{Defending a City from Multi-Drone Attacks:\\A Sequential Stackelberg Security Games Approach}

\author[biu]{Dolev Mutzari\corref{cor}}
\ead{dolevmu@gmail.com}
\cortext[cor]{Corresponding author}
\author[neiu]{Tonmoay Deb}
\ead{tonmoay.deb@northwestern.edu}
\author[unical]{Cristian Molinaro}
\ead{cmolinaro@dimes.unical.it}
\author[unical]{Andrea Pugliese}
\ead{andrea.pugliese@unical.it}
\author[neiu]{V. S. Subrahmanian}
\ead{vss@northwestern.edu}
\author[biu]{Sarit Kraus}
\ead{sarit@cs.biu.ac.il}


\affiliation[biu]{organization={Department of Computer Science, Bar Ilan University},
              country={Israel}}

\affiliation[neiu]{organization={Department of Computer Science, Northwestern University},
              state={IL},
              country={USA}}

\affiliation[unical]{organization={DIMES Department, University of Calabria},
              country={Italy}}

\begin{abstract}
To counter an imminent multi-drone attack on a city, defenders have deployed drones across the city. These drones must intercept/eliminate the threat, thus reducing potential damage from the attack. We model this as a Sequential Stackelberg Security Game, where the defender first commits to a mixed sequential defense strategy, and the attacker then best responds. We develop an efficient algorithm called S2D2, which outputs a defense strategy.  We demonstrate the efficacy of S2D2 in extensive experiments on data from 80 real cities, improving the performance of the defender in comparison to greedy heuristics based on prior works. We prove that under some reasonable assumptions about the city structure, S2D2 outputs an approximate Strong Stackelberg Equilibrium (SSE) with a convenient structure.
\end{abstract}




\end{frontmatter}

\section{Introduction}
There has been a lot of recent concern about multi-drone attacks~\cite{brust2017defending,he2020effective,jurn2021anti,guitton2021fighting,brust2021swarm,chen2022countering,castrillo2022review,li2023optimization}, especially in highly populated urban areas where not all countermeasures can be used~\cite{castrillo2022review}.
Drones can target centers of government and severely damage critical infrastructure (e.g., utilities). It has been proposed~\cite{brust2017defending,brust2021swarm,guitton2021fighting,castrillo2022review} that the city can be defended with drones to counter the attacks and reduce damage to life and property. \revision{As drones are cheap, accessible, and can maneuver above city buildings, effective defense should be equally affordable, and free from ground-based constraints.}

Therefore, in this work we focus on defending against multi-drone attacks on large-scale cities, using defense drones. It is clear that certain locations in the city are more attractive to attack for the attacker and hence more critical for the defender to protect. The goal is therefore to minimize damage rather than to catch the attacker drones as fast as possible. Finally, while aerial drones can be relatively easy to purchase, they are subject to battery and payload constraints.

Stackelberg security games (SSGs) offer a framework to optimize the allocation of defense resources against strategic adversaries. \cite{kar2017trends,sinha2018stackelberg} provide an extensive overview of SSG applications successfully deployed to date. An SSG consists of a defender with some defense resources protecting multiple targets against a strategic attacker. The defender commits to a mixed allocation strategy, and the attacker best responds by attacking the target that maximizes her utility.

Many extensions of the original SSG model~\cite{paruchuri2006security} exist today, e.g.\ bounded rationality attackers~\cite{pita2012robust,pita2015bounded,mutzari2022robust}, partial information~\cite{kiekintveld2010robust}, defense schedules \cite{korzhyk2010complexity}, heterogeneous resources~\cite{korzhyk2010complexity}, multiple defenders~\cite{gan2018stackelberg,mutzari2021coalition} and attackers~\cite{solis2015solving}, attackers with multiple resources~\cite{korzhyk2011security}, and repeated SSGs~\cite{kar2015learning}.
Nevertheless, most research is on Stackelberg equilibria in normal-form games: the defender commits to a mixed strategy, the attacker best responds, and the expected utilities are then directly determined. In particular, the attacker has a single opportunity to attack.

To defend a city from multi-drone attacks\footnote{Our framework also applies to land-based attacks by a coordinated set of attackers, targeting a city with simultaneous or sequential attacks by traversing its roads.}, we use \emph{sequential} SSGs, in which the targets are nodes in a graph, which both players' drones traverse. In particular, we model this as an extensive-form game. The attacker's drones are subject to payload and battery capacity constraints.

\subsection{Related Work}
\label{sec:rw}

Defending against \emph{swarm} aerial drone attacks has been studied extensively --- \cite{guitton2021fighting,castrillo2022review} provide a recent overview. In a drone swarm, each drone acts in real-time based on its local observation of the environment, including neighboring drones. Modeling attacker drones as a swarm is limiting since an attacker with sufficient computational and technological resources can conduct coordinated attacks to increase its utility. For similar reasons, while defense using a drone swarm is more scalable with the number of drones, both computationally and from practical perspectives, it is less effective than a coordinated multi-drone defense mechanism.

Past research on drone swarm attacks can be roughly split into three domains: (i) detection mechanisms focusing on identifying an incoming attack, tracking and classifying air-drones~\cite{jurn2021anti,chiper2022drone}, (ii)
\revision{quickly assessing whether a tracked drone is threatening or not~\cite{deb2025drone}}, and (iii) defense mechanisms that seek to counter and protect against \revision{threatening drone} attacks~\cite{brust2017defending,he2020effective,guitton2021fighting,chen2022countering}. The growing body of work on detection mechanisms is complementary to this work, justifying the assumption that attack drones can be monitored.\footnote{In fact, we make the weak assumption that the location of a drone is known only after its first strike takes place.}

Next, we briefly cover the gaps and limitations of defensive mechanisms other than using defense drones. GPS jamming / spoofing (used e.g.\ in~\cite{he2020effective}) cannot tackle drones that use other navigation methods (visual, radar, etc.), and RF jamming is not effective against autonomous malicious drones. Furthermore, anti-jamming/spoofing techniques may undermine their effectiveness. In addition, these methods may jam civilian applications (e.g., mobile phone communications). We refer to~\cite{castrillo2022review} for further discussion and focus on the defensive drone swarm literature.

\cite{brust2017defending,brust2021swarm} and ~\cite{han2023cooperative} study defense using a drone swarm. These works mostly focus on coordinating defensive drones, and the attacker model is limited. First, only a single attacker drone is considered. Second, it is assumed that the attacker drone is nearby, and was detected before causing any damage. This might work for protecting a facility of interest, but spreading them would enable covering much more ground. Third, once it is detected, the defensive drone swarm assumes the attacker drone follows a straight projectile\footnote{\cite{han2023cooperative} adds a brief discussion on other strategies the attacker might choose.} to predict its future location and catch it rapidly. Obstacles might hinder such movement of the attacker, and more importantly, the attacker is interested not only in evading the defensive swarm but also in striking targets, otherwise it would not take off to begin with. ~\cite{de2021decentralized,manoharan2023multi} alleviated the assumption of straight-line movement by learning from simulations using Deep Reinforcement Learning.

The above works fall under multi-pursuer multi-evader differential games \cite{pontryagin1966theory}, where each player decides on a continuous function over time, called \emph{control} that must admit certain constraints. \cite{chen2016multiplayer} pairs the pursuers and evaders thereby reducing the problem into a single pursuer single evader game, and we follow a similar approach. Differential games (DGs) can be roughly divided into two forms: \emph{open-loop} DGs where the controls depend only on time and initial game state and there is no dependence on the current game state, and \emph{closed-loop} DGs where controls may be a function of the continuously evolving state.

In our setting, we want the defender to be closed-loop and utilize recent work on detecting and monitoring attack drones, whereas the attacker should be open-loop as it does not know the defense drone locations. Another well-studied family of evasion games are cops and robbers~\cite{bonato2011game}, traversing a graph. The locations of each cop and robber are typically visible. There are works on invisible robbers~\cite{kehagias2013cops,kehagias2014role,dereniowski2015zero}, but not on invisible cops. \cite{kehagias2013cops} also considers a drunk robber, which effectively does not take the cops' locations into account, but instead takes a random walk, and we are interested in a rational attacker. Moreover, the goal in evasion games (both on graphs and differential games) is to catch the evaders as fast as possible. In particular, they do not take into account rewards and penalties from successful attacks.

Finally, there has been some work on sequential security games (which is the approach taken in this paper) to model the problem at hand. This should not be confused with repeated SSGs, which are one-shot games, played multiple times to enable players to gain information. For instance, \cite{nguyen2020tackling} studies repeated SSGs with unknown attacker type to handle deception, and~\cite{kar2015learning} studies repeated SSGs where the attacker does not know the defense mixed strategy initially.

In sequential SSGs~\cite{nguyen2019tackling},  the defender and attacker simultaneously traverse a graph. The attacker can attack multiple targets on her path. As in classical SSGs, the defender commits to a mixed strategy, and the attacker best responds. Unlike traditional SSGs, the strategy space is huge. 
\cite{nguyen2019tackling} assumes: (i) the attacker has one drone, (ii) drones carry unbounded payload, (iii) a solution is offered only against two sequential strikes, (iv) solutions assume that either defense movement is unrestricted or is prohibited completely. \cite{tomavsek2020using} extended~\cite{nguyen2019tackling} by alleviating (iii), but assumes a zero-sum finite game, where SSE and NE are equivalent~\cite{korzhyk2011stackelberg}.

General sequential SGs were first considered in~\cite{breton1988sequential}. Exact methods~\cite{bosansky2015sequence,cermak2016using} do not scale to our setting as they are at best linear in the game tree. Heuristic algorithms (e.g.,~\cite{vcerny2018incremental}), being generic, perform poorly in our setting. They do not exploit the graph structure of the problem and lack basic tools (e.g., shortest path and TSP solvers).
\cite{karwowski2015new} considers sequential SGs and develops an MCTS-based heuristic algorithm. Nevertheless, this method is not suitable for finding a Strong Stackelberg Equilibrium (SSE). \cite{vasal2022sequential} considered a discrete-time stochastic Stackelberg game where the attacker has a private type that evolves as a controlled Markov process. They compute a Stackelberg equilibrium by solving lower dimensional fixed-point equations for each time $t$. 
Their technique assumes the state to be small.

\subsection{Contributions}
The main contributions we make are summarized below.
\begin{enumerate}
    \item We extend sequential SSGs to handle multiple attack/defense drones with payload/battery constraints.
    \item We propose \emph{Sequential Stackelberg Drone Defense (S2D2)}, an efficient algorithm to output a defense strategy.
    \item We identify conditions for the underlying graph, under which S2D2 outputs an approximate Strong Stackelberg Equilibrium (SSE), along with an upper bound on the error. We also develop an algorithm to check if a given graph admits such a structure.
    \item Though our theoretical results make assumptions to guarantee the existence of approximate SSEs, not all real-world situations satisfy these conditions. Thus:
    \begin{itemize}
        \item We ran extensive experiments on a dataset of 80 famous world cities (1000s to $\sim$250K nodes) using two distributions (Zipf and log-normal) to assign utilities to neighborhoods of the city.
        \item We conducted a detailed case study of 6 cities (one small and two big US cities, a large and a small city in the Middle East, a megacity in Asia) using utilities provided by experts, rather than random assignment. Our experiments compare S2D2 to a heuristic algorithm based on prior works that trades off runtime and defender utility.
        \item We studied the robustness of the computed approximate SSEs by perturbing the utilities and looking at performance variations. Our results showed that slightly perturbing game parameters (e.g., penalties and rewards) led to proportional changes in defender utility.
    \end{itemize}
    We conclude that even when theoretical assumptions do not hold, S2D2 still yields good results.
\end{enumerate}

\revision{
Section~\ref{sec:discussion} contains
a deeper discussion of the rationale behind our model design, including justifications for key choices, alternative approaches with their trade-offs, and other relevant questions. This section also presents non-trivial arguments that further support our modeling decisions.}

\subsection{Organization}
\revision{Section~\ref{sec:birdseye} provides a high-level, birdseye view of the overall S2D2 architecture and decision. In particular, it explains how the different parts of this paper fit together.}
Section~\ref{sec:seq-ssg} presents the problem of interest, modeled as a sequential SSG. \revision{A deeper discussion of the rationale behind our model design, and comparison with alternative approaches is presented in  Section~\ref{sec:discussion}}.
Section~\ref{sec:s2d2-algorithm} then describes our \emph{S2D2} algorithm, which has three steps. First, a ``coarsening'' algorithm (cf.\ Section~\ref{subsec:coarsening}) partitions an input city graph into clusters (``\emph{neighborhoods}'' 
--- clusters of vertices).
Then, an approximate solution is computed (cf.\ Section~\ref{subsec:single-drone-seq-ssg}), assuming both the attacker and the defender have one drone and play in one neighborhood. This algorithm is an extension of~\cite{korzhyk2010complexity}'s method to sequential games, where the attacker strategy space becomes overwhelmingly large.
We then discuss how to use the solution for the single drone game to find an approximate solution for the multi-drone game (cf.\ Section~\ref{subsec:seq-ssg-approx-solution}). This is achieved by generalizing~\cite{korzhyk2011security}'s work on multi-resource attacker SSGs to support non-linear utilities. S2D2 uses this method to decide the allocation of defense drones into neighborhoods.
Section~\ref{sec:theoretical-results} proves that under a set of conditions on a coarsened graph, S2D2 is sure to output an approximate SSE, and Section~\ref{sec:experiments} presents experimental results.
Finally, Section~\ref{sec:conclusions} outlines our conclusions.

\revision{
\section{Birdseye View of S2D2}\label{sec:birdseye}
In this section, we present a birdseye view of the S2D2 system and describe its architecture (cf. Figure~\ref{fig:architecture}). S2D2 contains the following components.

\revision{
\begin{figure}[h!t]
\centering
{ 
  \color{white}
  \setlength{\fboxsep}{4pt} 
  \setlength{\fboxrule}{0.5pt} 
  \fbox{%
    \includegraphics[width=\linewidth]{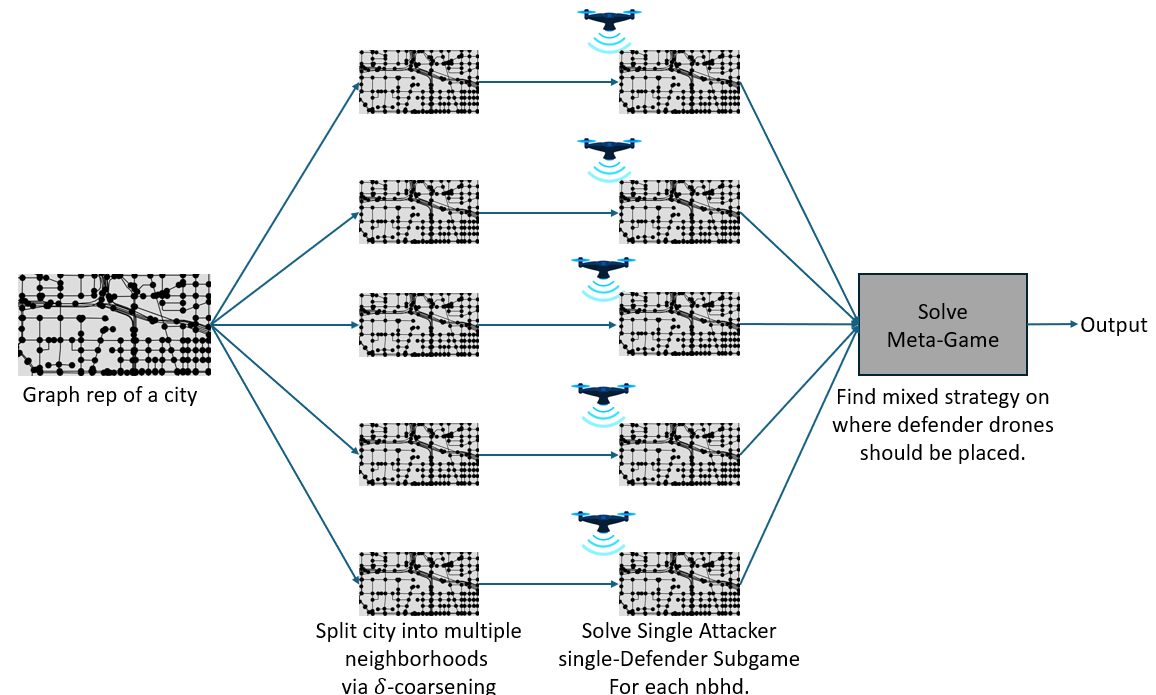}
  }
}
\caption{\revision{S2D2 Architecture.}}
\label{fig:architecture}
\end{figure}
}

\begin{itemize}
\item\textbf{Cities represented as graphs.} We represent cities being protected as a graph. Each node in the graph represents a region on the ground. Adjacent nodes in the graph represent adjacent regions on the ground, i.e., regions that share a common border.
\item\textbf{Coarsening a graph for scalability.} Because cities can be huge (the number of vertices in our
dataset vary from 2.2K to 277K and the number of edges vary from
3.4K to 405K), game-theoretic models will not scale. Because of this, we \emph{coarsen} a graph into \emph{neighborhoods}. A neighborhood consists of a connected set of nodes in the city graph. We will require coarsenings to satisfy some desired properties (discussed further below). An algorithm to find a good coarsening is described in Algorithm~\ref{alg:coarsening} in Section
~\ref{subsec:coarsening}.
\item\textbf{Single-Defender, Single-Attacker Game Per Neighborhood.} Next, we look and ask the question: if a single  defender and a single attacker drone are in a given neighborhood, what strategy would maximize their respective utilities? We solve this problem by building on top of the results of \cite{korzhyk2010complexity}. However, fixing the coarsening first and then solving a single attacker single defender problem could lead to suboptimal solutions. The attacker is not formally restricted to place each drone in a single neighborhood throughout the game, and it may also be suboptimal for the defender to do so. The coarsening algorithm is therefore responsible to correctly capture the attacker and defender incentives, and provide a corresponding coarsening of the graph. We propose the concept of a $\delta$-coarsening that ensures several desirable properties of the coarsening. We then design an algorithm to find a $\delta$-coarsening (Algorithm~\ref{alg:coarsening}).
\item\textbf{Solving the Meta Game.} Once we understand the utilities of the single attacker, single-defender game, one in each neighborhood, we need to determine where the defender must place his/her defender drones.  The third part of the S2D2 algorithm addresses this problem (Algorithm~\ref{alg:meta-game}) using a mixed strategy. This will be discussed further in Section~\ref{sec:multi-drone-milp}.
\end{itemize}
}

\section{Sequential SSGs}
\label{sec:seq-ssg}
\revision{We start by briefly overviewing sequential SSGs in the context of our problem.}
In sequential SSGs, the defender may re-distribute its defense drones after a successful attack. While doing so, the defender knows the attacker drones' location and which targets were destroyed. Meanwhile, the attacker may select and start moving toward other potential targets. The game continues until all attacker drones are either caught, out of battery, or out of payload. The attacker may only attack targets close to her current position. Formally, the game consists of: 
\begin{enumerate}
    \item An undirected graph
    $G = (V, E)$, where:
    \begin{itemize}
        \item $V=\{1,\ldots,m\}$ is a set of $m$ target nodes.
        \item $E$ is a set of undirected edges between targets.
    \end{itemize}
    \item $R^a: V \rightarrow \NN$ and  $P^d: V \rightarrow \ZZ_{<0}$ 
    map each target to the attacker reward and defender penalty,  respectively\footnote{Unlike traditional SSGs, we set attacker penalties and defender rewards to zero ($P^a=R^d=0$) since the attacker is already penalized when caught, as it cannot attack any more targets. Similarly, the defenders are rewarded when they catch the attacker as doing so prevents future strikes.}, from an attack on a given node $v\in V$.
    \item $A,D\in \NN$ are the number of attacker and defender drones, respectively.
    \item The payload $P\in \NN$ each attacker drone is able to carry. This equals the maximal number of attacks each drone can pull-off (if not caught or run out of battery).
    \item The battery capacity $B \in \NN$ each attacker drone has. This equals the maximal total distance it can travel (if not caught). We assume traversing an edge $e\in E$ takes one unit of battery (adding $0$-rewarded/penalized nodes along a long edge if necessary), as well as staying (or loitering) over a node.
\end{enumerate}

\paragraph{Assumptions}
We assume the defender knows  $(A,P,B)$ and the current location of each attacker drone at all times after the first strike by that drone. Defense drones also have a battery capacity $B$. Hence, without loss of generality, the game ends after $B$ steps. The attacker only knows the number of defense drones $D$ at the beginning of the game. Attacker drones do not know the locations of defense drones unless they meet at a node --- this is reasonable as a defender can deploy sensor and other assets in her city. When this occurs, the attacker drone is destroyed. Attacker drones are not informed when other attacker drones are eliminated.

\subsection{Defender and Attacker Strategies} 
The defender knows the location of some attacker drones and can leverage this information. Formally, a pure defender strategy $s_d \in \calS^d$ is a $B$-tuple of functions $(s_1^d,\ldots,s_B^d)$, specifying its strategy at each time-step. The first strategy $s_1^d \in V^D $ specifies the start position of each defense drone. At any step $1<t\le B$, the function $s_t^d$ determines the next step of each drone given the current state of the game, which includes:
\begin{itemize}
    \item Last location of each defense drone $(v_1^d,\ldots,v_D^d) \in V^D$.
    \item Last location of each observed attack drone $(v_1^a,\ldots,v_A^a) \in (V\cup \{ \bot,\dagger\})^A$. We use the special symbol $\bot$ for unknown location (no strike yet), and $\dagger$ for eliminated.
    \item Subset of destroyed targets $I_{t-1} \subseteq V$ (where $I_0=\emptyset$).
\end{itemize}

In a single step, a drone at location $v \in V$  can only reach neighboring locations in graph $G$, i.e.\ $N[v] := \{ v' \in V: \{v,v'\}\in E\} \cup \{v\}$.
The function $s_t^d$ outputs the new location of each defense drone $(\tilde{v}_1^d,\ldots,\tilde{v}_D^d)$ where $\tilde{v}_i^d \in N[v_i^d]$ for each $1\le i \le D$.\footnote{
Note that the game is Markovian: the history of how drones ended up in their last observed location, or the order in which targets have been destroyed, cannot be utilized against a rational attacker.} 
\revision{
Figure~\ref{fig:strategies}(a) provides a visualization (from our S2D2 system) of the defender's strategy overlaid over a map of a city. The locations of defender drones (blue) and attacker drones (red) as well as the destroyed parts of the city are shown as icons. The defender's strategy specifies a function that answers the following question: given a picture like the one depicted, where should the blue drones move to next?}

\revision{
\begin{figure}[h!t]
\centering

{ \color{white}
  \fbox{%
    \setlength{\fboxsep}{0pt}   
    \setlength{\fboxrule}{0.3pt}  
    \begin{tabular}{@{}c@{}c@{}}
        \includegraphics[width=0.48\textwidth]{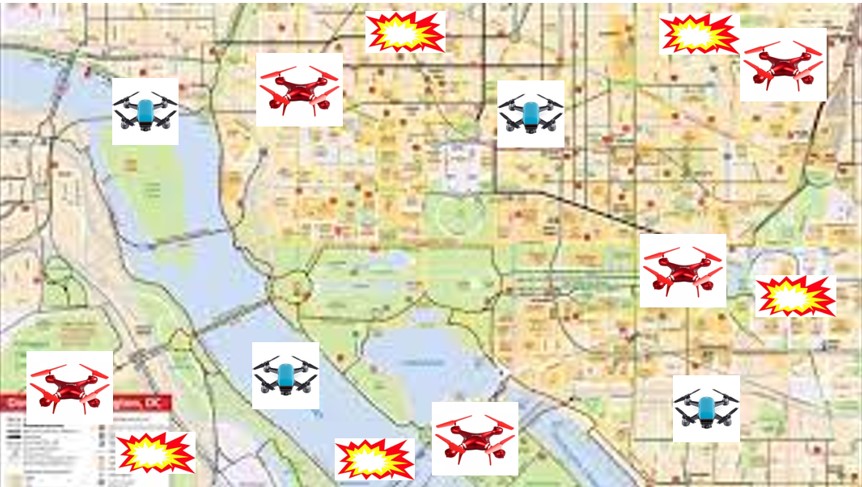} &
        \includegraphics[width=0.48\textwidth]{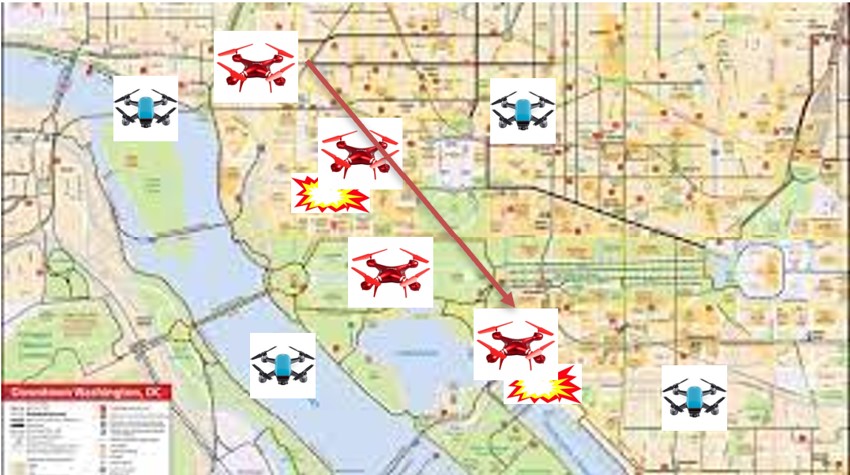} \\
        (a) & (b)
    \end{tabular}
  }
}
\caption{\revision{(a) Visualization of Defender Strategy. (b) Visualization of Attacker
Strategy. }}
\label{fig:strategies}
\end{figure}
}

\revision{
Figure~\ref{fig:strategies}(b) shows the attacker strategy. For each attacker drone (shown in red), a flight path is specified (shown for one red drone in Figure~\ref{fig:strategies}(b) as a red arrow). In addition, the strategy specifies where each attacker drone will actually target with one unit of payload. In Figure~\ref{fig:strategies}(b), we see two locations where payload is used by this attacker, marked by an explosion icon. To keep the figure simple, we do not show these flight paths and payload utilization for the other attacker drones depicted.}
The pure strategies for the attacker are related to $B$-length paths in the graph. We use $\calP_B := \{(v_1,\ldots,v_B) \in V^B \mid \forall 1\le t < B: \{v_t,v_{t+1}\} \in E \revision{~\lor~ v_t=v_{t+1}} \}$ to denote the set of all paths of length $B$ in $G$, and let $\calP_0=\{\emptyset\}$. Recall that traversing each edge requires one battery unit, \revision{as well as hovering over a node ($v_{t+1}=v_t$)}.\footnote{The sequential SSG has a few natural extensions which we may consider in future work. These include: (i) \textit{Heterogeneous drones}: The attacker may have drones of different types, $(B_1,P_1),\ldots,(B_A,P_A)$. (ii) \textit{Distances}: The edges may be weighted as well, by the distance between its endpoints. Adding $d-1$ vertices along an edge with distance $d$ will not yield a reduction. Indeed, one has to define a reward over these new vertices, say $0$. Still, the defender will know the attacker's position in the first step along the split edge. (iii) \textit{Velocities}: Different drones may fly with different velocities. The velocity may also depend on the percentage of loaded payload. (iv) \textit{Defense schedules}: Allocating a defense drone to some target $v$ may also protect its neighbors $N(v)$.}
Furthermore, each attacker drone must decide which targets to attack. Let $\calT_{P,B}=\{I\subseteq \{1,\ldots,B\} \mid \abs*{I} \le P\}$ denote the set containing sets of at most $P$ indices along the path of length $B$ to be attacked. The set of pure strategies of the attacker is therefore $\calS^a = (\calT_{P,B} \times \calP_{B})^A$.

\paragraph{Utility} 
Given an attacker (resp. defender) strategy $s_a \in \calS^a$ (resp. $s_d \in \calS^d$), we can recursively compute utilities at time $t$. Initially, $u_0^a=u_0^d=0$. At time $t>0$, we  compute the position of all surviving drones from the specified strategies and the previous drone locations. We update the utilities $u_t^a = u_{t-1}^a+r_t^a$ and $u_t^d = u_{t-1}^d+p_t^d$ where $r_t^a$ ($p_t^d$) is the sum of rewards (resp. penalties) from successful attacks at step $t$ for the attacker (defender). We then nullify the rewards for targets that were successfully attacked at time step $t$, and eliminate any attacker that is either caught or out of payload. Finally, we set $u^a(s_d,s_a)=u_B^a, u^d(s_d,s_a)=u_B^d$.

\subsection{Mixed Strategies} 
The defender may use a \emph{mixed} strategy. In other words, it may sample its strategy from a distribution $\xx_d \in \Delta(\calS^d)$, where $\Delta(\calS^d)$ is the set of all probability distributions over $\calS^d$. For the special case where $B=1$ (the non-sequential SSG), we can use a compact representation $\calC_D := \{\xx \in [0,1]^m: \sum_{v\in V} x_v \le D \}$ of the set of defender mixed strategies. A vector $\xx\in\calC_D$ is called a \emph{coverage vector}, and it denotes the probability that each node $v\in V$ is covered by some defense drone. Coverage vectors can provably be implemented by a distribution over deterministic allocation strategies, each using at most $D$ resources. This distribution can also be found efficiently, see~\cite{korzhyk2010complexity}, Theorem~1. 

In Stackelberg games, the attacker can conduct surveillance on the defender's (mixed) strategy $\xx_d$ beforehand and best respond to it. Assume now the defender and the attacker play mixed strategies over $\calS^a,\calS^d$, respectively. Given mixed strategies $\xx_d,\xx_a$, the utility of the attacker (and similarly the defender) is given by 
\begin{eqnarray}
\label{eq:u-d}
u^a(\xx_d,\xx_a)&:=&\EE_{(s_d,s_a) \sim \xx_d \times \xx_a}[u^a(s_d,s_a)]\\
&=& \sum_{(s_d,s_a)\in \calS^d \times \calS^a}{\xx_d(s_d)\xx_a(s_a)\cdot u^a(s_d,s_a)} \nonumber 
\end{eqnarray}

\begin{Example}[Sequential SSG: Toy Example]\label{example:seq-ssg-toy}
Consider a toy graph $G=(V,E)$ with $m=41$ vertices and edges depicted in Figure~\ref{fig:seq-ssg-neighborhoods}. 
\begin{figure}[h!t]
\centering
{ \color{white}
\fbox{
\includegraphics[width=0.8\textwidth]{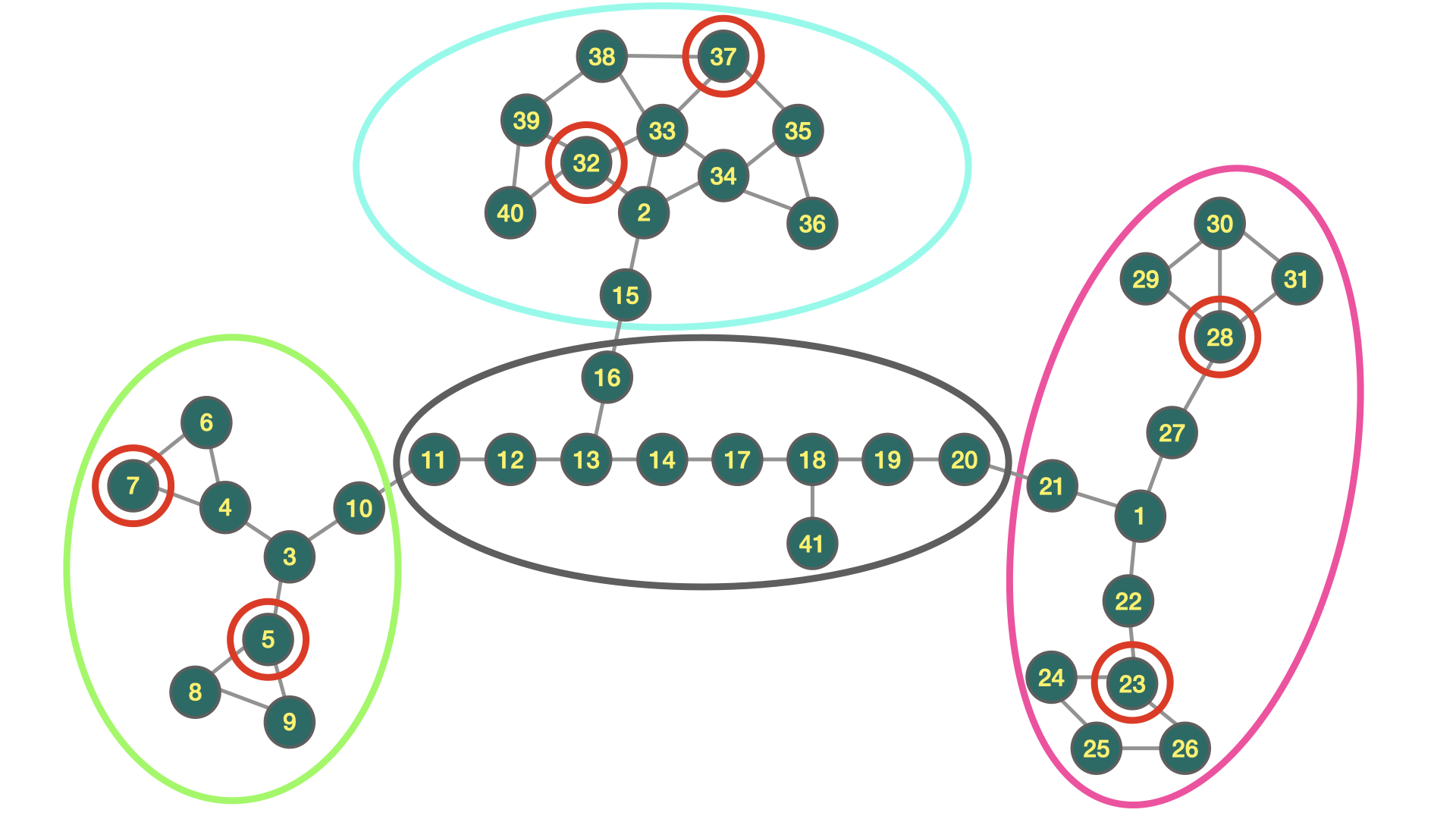}
}}

\caption{\revision{A graph and its coarsening into neighborhoods.}}
\label{fig:seq-ssg-neighborhoods}
\end{figure}
Suppose we set $P^d\equiv -R^a$ in our example, and
the attacker rewards are set to one for targets $v_5,v_7,v_{23},v_{28},v_{32},v_{37}$, and zero for all the rest.
Suppose the defender and the attacker both have $A=D=2$ drones, and that $B=4,P=2$ for attacker drones. 

A defender pure strategy may first place the defense drones on $v_3,v_1$ respectively. Then, given the attacker position, the strategy would let each defense drone follow the closest path towards the closest attacker drone. Denote this strategy by $s_d^1$.
Suppose the attacker plays strategy $s_a$ where her drones are at $v_{37},v_{28}$. The first drone follows path $v_{37} \rightarrow v_{38} \rightarrow v_{39} \rightarrow v_{32}$, and attacks $v_{37}$ and $v_{32}$. The second drone follows $v_{28} \rightarrow v_{27} \rightarrow v_1 \rightarrow v_{22} \rightarrow v_{23}$ and attacks $v_{28}, v_{23}$.
In this case, the defense drone starting at $v_3$ will not do much, but the defense drone starting at $v_1$ will catch the drone that started at $v_{28}$ before $v_{23}$ is attacked. We can verify that when facing pure strategies, the attacker may always successfully attack two meaningful targets using one of her drones. Instead, the defender may use a mixture $\xx_d$ of $3$ strategies, each for instance with probability $1/3$. Suppose $\xx_d(s_d^1)=1/3$, $\xx_d(s_d^2)=1/3$ and $\xx_d(s_d^3)=1/3$.  
In $s_d^2$, the defense drones start from $v_2,v_3$, and in $s_d^3$, they start from $v_1,v_2$. By doing so, there is always a probability ($2/3$ in this case) that a defense drone is ``in the hood''.
\end{Example}

In SSGs, the attacker knows the defender's mixed strategy $\xx_d \in \Delta(\calS^d)$, and then best responds to it with $s_a \in \BR^a(\xx_d)$. Since the utility of the attacker from a mixed strategy is the weighted average of the utilities from each pure strategy, she may always choose a pure strategy that yields the maximal utility. Therefore, w.l.o.g., the attacker's best response set consists of pure strategies only: 

\[\textstyle \BR^a(\xx_d) := \arg\max_{s_a\in \calS^a} u^a(\xx_d, s_a)\] 


When there are multiple targets in $\BR^a(\xx)$, we take the standard approach~\cite{leitmann1978generalized}
and assume that the attacker breaks ties in favor of the defender. The reason is that by reducing the coverage of the desired target by an arbitrarily small amount, the attacker will attack the desired target and the defender will suffer an arbitrarily small utility loss. We therefore define 
\[\textstyle \BR^d(\xx_d) = \arg\max_{s_a \in \BR^a(\xx_d)} u^d(\xx_d,s_a).\]
The set of strategies in $\BR^a(\xx)$ are the ones that are best for the defender. We may then define $u^d(\xx_d):=u^d(\xx_d,s_a), u^a(\xx_d):=u^a(\xx_d,s_a)$ for $s_a\in\BR^d(\xx_d)$. This is well-defined as the value is independent of the choice of $s_a$.

The typical solution concept for SSGs is Strong Stackelberg Equilibrium (SSE).

\begin{Definition}[Strong Stackelberg Equilibrium]
A strategy profile $(\xx_d,s_a)\in\Delta(\calS^d) \times \calS^a$ is a \emph{Strong Stackelberg Equilibrium} iff $$\xx_d\in\argmax_{\xx_d'\in \Delta(\calS^d)}{u^d(\xx_d')} ~\text{and}~ s_a\in\BR^d(\xx_d).$$
\end{Definition}

\noindent Approximate SSE's are defined analogously.

\begin{Definition}[$\epsilon$-approximate SSE]
A strategy profile $(\xx_d,s_a)\!\in\!\Delta(\calS^d) \times \calS^a$ is an \emph{$\epsilon$-approximate SSE ($\epsilon$-SSE)} iff 
\begin{eqnarray}
u^a(\xx_d,s_a) + \epsilon &\ge& \max_{s_a'\in \calS^a}{u^a(\xx_d, s_a')} \label{eq:br-eps-seq} 
~~\text{and}~~ \\
u^d(\xx_d,s_a)+\epsilon &\ge& \max_{\xx_d'\in\Delta(\calS^d)} {u^d(\xx_d',\BR_\epsilon^d(\xx_d'))}, \nonumber
\end{eqnarray}
where $\BR_\epsilon^a(\xx_d')$ consists of all strategies $s_a$ satisfying~(\ref{eq:br-eps-seq}), and $\BR^d_\epsilon(\xx_d')$ $\subseteq \BR_\epsilon^a(\xx_d')$ consists of all strategies  in $\BR_\epsilon^a(\xx_d')$ that maximize defender utility (breaking ties optimistically).
\end{Definition}

\paragraph{Finding SSE Efficiently by Solving Linear Programs}
Equation~(\ref{eq:u-d}) suggests that the defender's utility is linear with respect to the coverage vector $\xx_d$. Furthermore, the defender's strategy space $\Delta(\calS^d)$ is a polytope. This suggests using linear programming. We extend the approach in~\cite{korzhyk2010complexity} for $B=1$ to general sequential games as detailed below. We want to compute: $$\xx_d\in\argmax_{\xx_d'\in\Delta(\calS^d)}{u^d(\xx_d')} = \argmax_{\xx_d'\in\Delta(\calS^d)}{u^d(\xx_d',\BR^d(\xx_d'))}.$$ The only problem is that the $\BR^d(\xx_d)$ is not linear. Our idea is to solve, for each potential $s_a^*$ candidate for $\BR^d(\xx_d)$, the LP (linear program):
\begin{itemize}
    \item Maximize $u^d(\xx_d,s_a^*)$, subject to:
    \begin{enumerate}
        \item $\xx_d \in \calC_D$.
        \item $\forall s_a\in \calS^a$, $u^a(\xx_d,s_a) \le u^a(\xx_d,s_a^*)$.
    \end{enumerate}
\end{itemize}
That is, we add $|\calS^a|$ linear constraints to ensure that $s_a^* \in \BR^a(\xx_d)$, and enumerate over $s_a^*$. At the end, we pick the solution that gives the defender the greatest utility.

\paragraph{Multiple Attack Resources}
In the sequential SSG, we consider multiple attacker drones, that is, multiple attacker resources. In this case~\cite{korzhyk2011security} showed that finding SSE is NP-hard. This also implies that the problem of finding sequential SSGs is NP-hard via a reduction from finding SSE in SSGs with multiple attacker resources. Simply let each attacker drone have a single unit of battery, to make the game effectively a non-sequential SSG. Nevertheless, NP-hard problems like MILPs (Mixed Integer Linear Programs) are well-studied and practical solutions have been developed previously. Indeed, S2D2 involves a reduction to a MILP.

Table~\ref{tab:symbols} summarizes the symbols used in this paper. \revision{A comprehensive discussion regarding our proposed model is provided in Section~\ref{sec:discussion}.}

\begin{table}[h!t]
\centering
\begin{scriptsize}
\begin{tabular}{r|l}
\toprule
\emph{Symbol} & \emph{Meaning}\\
\midrule
& {\bf SSSG Model}: \\
$G = (V, E)$ & City graph, where $V=\{1,\ldots,m\}$ is the set of nodes\\
& and $E$ is the set of undirected edges \\
$(R^a,P^d): V \rightarrow \NN \times \ZZ_{<0}$ & Attacker reward and defender penalty functions\\
$A,D\in \NN$ & Number of attacker and defender drones\\
$P,B\in \NN$ & Payload and battery capacity of attacker drones\\
$(v_1^d,\ldots,v_D^d) \in V^D$ & Locations of defender drones \\
\multirow{2}{*}{$(v_1^a,\ldots,v_A^a) \in (V\cup \{ \bot,\dagger\})^A$} & Locations of attacker drones \\
& ($\bot$ means unknown location, $\dagger$ means eliminated)\\
$s_d = (s_1^d,\ldots,s_B^d) \in \calS^d$ & Pure defender strategy, where $s_t^d$ is the policy at time $t$\\
$s_a\in \calS^a$ & Pure attacker strategy\\
$\xx_d\in\Delta(\calS^d)$ & Mixed defender strategy\\
$u^a(s_d,s_a),u^d(s_d,s_a)$ & Attacker and defender utilities under $s_d$ and $s_a$\\
$u^a(\xx_d,s_a),u^d(\xx_d,s_a)$ & Attacker and defender utility under $\xx_d\in \Delta(\calS^d)$ and $s_a\in\calS^a$\\
$\BR^a(\xx_d) \subseteq \calS^a$ & Set of best attacker responses to defender's mixed strategy $\xx_d\in\Delta(\calS^d)$\\
\midrule
& {\bf Coarsening}: \\
\revision{$\delta$} & \revision{Scale parameter, rewards smaller than $\delta$ are neglected}\\
$\hat{V}=\{\hat{v}_1,\ldots,\hat{v}_k\}$ & Coarsening of $G$, a set of disjoint neighborhoods $\hat{v}_i \subseteq V$\\
$\calS^a_{\hat{V}},\calS^d_{\hat{V}}$ & Set of pure attacker and defender strategies that \\
& respect the coarsening $\hat{V}$ \\
\midrule
& {\bf Single Drone Parameterized Sub-game}: \\
$\lambda$ & \revision{A parameter}, fixing the probability that a defender is present in a neighborhood\\
$u^d(\xx_d,s_a,\lambda),u^a(\xx_d,s_a,\lambda)$ & Attacker and defender utilities in a single drone game at a given neighborhood,\\
& under single drone strategies $\xx_d,s_a$, and defender presence probability $\lambda$\\
\midrule
& {\bf Multi-Drone Meta Game}: \\
$\hat{f}_a: \{1,\ldots,A\} \mapsto \hat{V}$ & Mapping from attacker drones to attacked neighborhoods \\
$\hat{p}_d(\hat{v})$ & Probability that a defender is present in $\hat{v}$\\
$(\hat{f}_a,\hat{s}_a),(\hat{p}_d,\hat{x}_d)$ & Attacker and defender multi-drone strategies \\
\bottomrule
\end{tabular}
\end{scriptsize}
\caption{\revision{Symbols used in the paper.}}
\label{tab:symbols}
\end{table}

\section{The S2D2 Algorithm}\label{sec:s2d2-algorithm}
The S2D2 algorithm generates a mixed defense strategy through three steps:
\begin{enumerate}
    \item \textbf{Coarsening} the graph, which \revision{involves partitioning it into artificial \emph{neighborhoods}}. The goal is to output a partition such that both the attacker and the defender are incentivized to spread their drones across different neighborhoods and stay there throughout the game. The defense (and attack) strategies can then be decomposed into the following two components.
    \item \textbf{Single-Attacker Single-Defender Game per Neighborhood.} 
    For each neighborhood, we solve a Single-Attacker Single-Defender sub-game and compute an approximate SSE. In reality, there is a probability \revision{$\hat{p}_d(\hat{v}) \in [0,1]$} that a defender is present in a neighborhood. \revision{Since this probability is unknown a-priori, it is treated as an unknown variable $\lambda$, provided as an additional input parameter.} S2D2 discretizes the interval $[0,1]$ into evenly spaced intervals and solves the problem for each $\lambda_i \in [0,1]$.
    \item \revision{\textbf{Solving the Meta-Game.} Once we know the defender utilities for each neighborhood, we can solve the problem of assigning a defender drone to each neighborhood.} Basically, each neighborhood is considered as one ``meta''-target. \revision{In this step,} S2D2 \textbf{determines a mixed strategy for allocating defense drones to neighborhoods} via a reduction to a non-sequential SSG between a multi-resource defender and a multi-resource attacker. The utility functions for both the attacker and defender are approximated by piece-wise linear functions, derived from solving the single-defender single-attacker sub-game within each neighborhood, as a function of the defender presence probability $\lambda$.
\end{enumerate}
\revision{The high-level pseudocode of the S2D2 algorithm is provided in 
Algorithm~\ref{alg:s2d2}.

\begin{algorithm}[t!]
\caption{\textsf{S2D2}}\label{alg:s2d2}
\begin{small}
\begin{algorithmic}[1]
\revision{
\REQUIRE
An undirected graph $G = (V, E)$; \\
\quad\ \  numbers of attacker and defender drones $A,D\in \NN$; \\
\quad\ \  attacker drone's  payload $P \in \NN$; \\
\quad\ \ drone's battery capacity $B \in \NN$; \\
\quad\ \ attacker rewards $R^a \in \NN^{|V|}$; \\
\quad\ \ defender penalties $P^d \in \ZZ_{<0}^{|V|}$. \\
\quad\ \ Discretization parameters $\#\lambda,\lambda_c$ for the piece-wise linear approximation of the single-drone utility sub-games within each neighborhood.
\ENSURE  An $\varepsilon$-SSE defense strategy $\xx^d=(\hat{V},\hat{p}_d,\hat{x}_d,\varepsilon)$ and $\varepsilon$, or $\xx^d$ and $\bot$, where: \\
\quad\ \ $\hat{V}$ is a coarsening of $G$ and $(\hat{p}_d,\hat{x}_d) \in \Delta(\calS^d_{\hat{V}})$; \\
\quad\ \ $\hat{p}_d \in \calC_D^{\hat{V}}$ is the allocation strategy of $D$ drones into neighborhoods of $\hat{V}$; \\
\quad\ \ $\hat{x}_d$ is the single-drone defense strategy within each neighborhood of $\hat{V}$;
\STATE Compute a coarsening $(\delta, \hat{V}) \leftarrow \textsf{Coarsening}(G,\ldots)$;
\FOR{\textsf{each} neighborhood $\hat{v}\in\hat{V}$}
    \STATE Compute piece-wise linear approximations of $u^a_{\hat{v}}(\lambda),u^d_{\hat{v}}(\lambda)$, the attacker and defender utilities for the single-drone game in neighborhood $\hat{v}$, where the defender is present with probability $\lambda$:
    \FOR{$\lambda = \frac{0}{\#\lambda},\frac{1}{\#\lambda},\ldots,\lambda_c$}
        \STATE Set $S^a=\textsf{ScanAttackStrategies}(\hat{v},\lambda,\ldots)$ (reduced attack strategy space).
        \STATE Set $S^d=\textsf{ScanDefenseStrategies}(\hat{v},S^a,\ldots)$ (reduced def. strategy space).
        \STATE Compute $u^d_{\hat{v}}(\lambda),u^a_{\hat{v}}(\lambda)$ as in~\cite{korzhyk2010complexity} and corresponding mixed strategy $\hat{x}^d(\hat{v})$, when restricting the attacker and defender strategy space to $S^a,S^d$.
    \ENDFOR
\ENDFOR
\STATE Invoke $\langle \hat{p}^d, \hat{f}^a \rangle \gets \textsf{SolveMetaGame}(\hat{V}, A, D, \{u_{\hat{v}}^d(\lambda), u_{\hat{v}}^a(\lambda)\}_{\hat{v}\in\hat{V}}$, to get the mixed allocation strategy $\hat{p}^d$ of $D$ drones into the neighborhoods of $\hat{V}$, by solving the static, multi-resource SSG with respect to the approximate utility functions $u^d_{\hat{v}}(\lambda),u^a_{\hat{v}}(\lambda)$.
\STATE In case $\delta \neq \bot$, compute $\varepsilon$ as in Theorem~\ref{theorem:approximate-sse}, otherwise set $\varepsilon=\bot$.
\RETURN $\xx^d,\varepsilon$;
}
\end{algorithmic}
\end{small}
\end{algorithm}}

In reality, the attacker may opt to ignore the coarsening found by S2D2. This may happen either since the attacker is not rational, or because a ``good coarsening'' does not exist. In such a case, S2D2 randomly picks, for each defense drone, an attacker drone in its neighborhood, and ignores the rest. In addition, whenever an attacker drone leaves a neighborhood, the defender drone in that neighborhood halts. When the coarsening admits certain properties, we show in Section~\ref{sec:theoretical-results} that this does not result in a major utility loss for the defender.

\revision{\paragraph{Approximations} S2D2 tries to find an $\epsilon$-approximate SSE, balancing the defender's computational resources with the approximation error $\epsilon$. To achieve this, S2D2 introduces a scale parameter $0<\delta<\max_v R^a{v}$, effectively disregarding rewards smaller than $\delta$. As $\delta$ increases, fewer nodes are deemed valuable, allowing S2D2 to focus on smaller subset of nodes to protect. Consequently, while this simplification reduces computational complexity, it also decreases the accuracy of S2D2's view of the game, leading to an expected increase in the approximation error $\epsilon(\delta)$. However, under certain conditions for the underlying graph, $\epsilon(\delta)$ can be bounded, which provides theoretical guarantees for our algorithm. Even when these conditions are not met, empirical results demonstrate that S2D2 performs effectively in practice.}

The next 3 subsections describe the three components listed above.

\subsection{Coarsening the Graph}\label{subsec:coarsening}
A \emph{coarsening} of $G=(V,E)$ is a set $\hat{V}=\{\hat{v}_1,\ldots,\hat{v}_k\}$ such that $\hat{v}_i \subseteq V$ for each $1\le i\le k$ and $\hat{v}_i\cap\hat{v}_j=\emptyset$ for any $i\neq j$. Each subset in $\hat{V}$ is a \emph{neighborhood}. A good coarsening is akin to ``zooming-out'', where nearby nodes are merged into a single neighborhood.

\revision{Ideally, a ``good'' coarsening (Step 1 of the S2D2 algorithm) cannot be found without simultaneously computing the utilities of the defender for that coarsening which is only considered in Step 2 of the S2D2 algorithm. One way to do this is to generate all possible coarsenings, then find the best defender strategy for each coarsening, and then pick the coarsening and defender strategy that yields the best utility for the defender. Unfortunately, this is not practical to compute.
We therefore introduce the concept of a $\delta$-coarsening to ensure that a coarsening is ``good'' and has some desirable properties.}

\revision{The scale parameter $\delta$ controls the granularity of the coarsening. Since S2D2 neglects rewards smaller than $\delta$, increasing $\delta$ reduces the number of nodes the coarsening algorithm considers. A node $v$ is deemed \emph{$\delta$-valuable} if $R^a(v)>\delta$. The coarsening algorithm then clusters these $\delta$-valuable nodes. In each cluster, all the $\delta$-valuable nodes are relatively close, while the clusters themselves remain relatively separated. The resulting coarsening then consists of a set of neighborhoods, each centered around a cluster of $\delta$-valuable nodes (see Figure~\ref{fig:seq-ssg-neighborhoods}).}

S2D2 coarsens via two steps, as depicted in Algorithm~\ref{alg:coarsening}. First, it attempts to detect a ``high-quality'' coarsening, referred to as $\delta$-coarsening. When a $\delta$-coarsening exists, we prove in Section~\ref{sec:theoretical-results} that S2D2 approximates SSE. A $\delta$-coarsening must satisfy four conditions: (i) getting from outside a neighborhood to a $\delta$-valuable node within it takes too much battery; (ii) there are sufficiently many valuable neighborhoods; (iii) a single attacker can collect most $\delta$-valuable rewards in its neighborhood; (iv) the presence of a defender significantly impacts both attacker and defender expected utility.
\revision{When a $\delta$-coarsening exists, the first step aims to minimize $\delta$, and does so efficiently by applying a binary search. Indeed, if any of conditions (i)-(iv) are not met for some $\delta_{\mathit{low}}$, they cannot be met for any $\delta<\delta_{\mathit{low}}$.}

\revision{If $\delta_{\mathit{low}} > \delta_{\mathit{up}}$, a $\delta$-coarsening may not exist at all.} To this end, if S2D2 fails to detect a $\delta$-coarsening in the first step, it proceeds to the second step, where it coarsens the graph using a greedy heuristic. \emph{It is important to note that S2D2 works even when no $\delta$-coarsening exists --- but in this case, the theoretical guarantees do not hold.}
In the following Example~\ref{example:coarsening} we provide an illustrative example of a coarsening.

\begin{Example}[Coarsening]
\label{example:coarsening}
Consider the graph in Example~\ref{example:seq-ssg-toy}, Figure~\ref{fig:seq-ssg-neighborhoods}. The gray neighborhood has no valuable nodes
and so is removed. Next, getting from one neighborhood to a valuable node of another requires going through the grey neighborhood, which takes a prohibitive amount of battery (i). Note that we only consider nodes circled in red when evaluating this condition as other nodes have no reward. Next, note that a single drone can tackle both red nodes within each neighborhood (iii). Unfortunately, the other two conditions (ii) and (iv) are not met with the desired constants required for the theoretical proof to hold. As for (ii), since we present a toy graph as an illustrative example, it only has $3$ neighborhoods (and $4$ are required). Splitting some neighborhoods into two may potentially violate (i). Similarly, for (iv), a defender can always stay put on one red node and block the attacker from successfully attacking both valuable nodes within every neighborhood, yet in this case, it yields a factor of $2$ between the utility from a protected neighborhood and an unprotected one. In more complex games with larger $B,P$ values and larger neighborhoods, the gap could be significantly larger.
\end{Example}

\begin{algorithm}[t!]
\caption{\textsf{Coarsening}}\label{alg:coarsening}
\begin{small}
\begin{algorithmic}[1]
\REQUIRE 
An undirected graph $G = (V, E)$; \\
\quad\ \  numbers of attacker and defender drones $A,D\in \NN$; \\
\quad\ \  attacker drone's  payload $P \in \NN$; \\
\quad\ \ drone's battery capacity $B \in \NN$; \\
\quad\ \ attacker rewards $R^a \in \NN^{|V|}$; \\
\quad\ \ defender penalties $P^d \in \ZZ_{<0}^{|V|}$. \\
\ENSURE 
($\delta$)-Coarsening $\hat{V}$ and $\delta$, or failure.
\STATE $\delta_{\mathit{low}}\leftarrow 1$, $\delta_{\mathit{up}}\leftarrow 1+\max_{v\in V}{R^a(v)}$;
\WHILE{$\delta_{\mathit{low}} < \delta_{\mathit{up}}$}
    \STATE $\delta \leftarrow \lfloor (\delta_{\mathit{low}}+\delta_{\mathit{up}})/2 \rfloor$;
    \STATE $\hat{V} \leftarrow V/\!\approx_\delta$; \label{step:call-coarsen}
    \STATE \textsf{init} table $T$;
    \FOR{\textbf{each} $\hat{v} \in \hat{V}$}
        \STATE $T[\hat{v}] \leftarrow \sum\limits_{v\in \hat{v}.\textsf{top}(P,\textsf{by}=R)}{R^a(v)}$; \revision{\COMMENT{sum of top-$P$ rewards}}
    \ENDFOR
    \STATE $\hat{V} \leftarrow \{\hat{v}\in \hat{V} \mid \frac{4}{3}T[\hat{v}] \ge T.\textsf{max}() \}$; \revision{\COMMENT{Remove poor neighborhoods}} \label{step:remove-poor-neighborhoods} \label{step:set-valuable}
    \IF{$|\hat{V}| < 4\max\{A,D\}$}
        \STATE $\delta_{\mathit{low}} \leftarrow \delta + 1$; \revision{\COMMENT{Not enough neighborhoods}} \label{step:not-enough-neighborhoods}
    \ELSIF{$\exists\hat{v}\in \hat{V}: |\{v\in\hat{v}\mid R^a(v)>\delta\}| > P$}
        \STATE $\delta_{\mathit{low}} \leftarrow \delta + 1$; \revision{\COMMENT{Insufficient attacker payload}} \label{step:attacker-can-collect-payload}
    \ELSIF{$\exists\hat{v}\in \hat{V}: \mathsf{best\mbox{-}path}(\hat{v},\delta)>B$}
        \STATE $\delta_{\mathit{low}} \leftarrow \delta + 1$;
        \revision{\COMMENT{Insufficient attacker battery}} \label{step:attacker-can-collect-battery}
    \ELSIF{$\exists\hat{v}\in\hat{V}: \frac{3}{64|\hat{V}|} u_{1,0}^{\hat{v},a} \le u_{1,1}^{\hat{v},a}$ \textbf{or} $\frac{3}{8|\hat{V}|} |u_{1,0}^{\hat{v},d}| \le |u_{1,1}^{\hat{v},d}| + \delta P$}
        \STATE $\delta_{\mathit{low}} \leftarrow \delta + 1$;
        \revision{\COMMENT{Defender presence is ineffective}} \label{step:defense-presence-meaningful}
    \ELSE
        \STATE $\mathit{sol}\leftarrow(\hat{V},\delta)$;
        \STATE $\delta_{\mathit{up}} \leftarrow \delta$;
    \ENDIF
\ENDWHILE
\IF{$\delta=1+\max_{v\in V}{R^a(v)}$}\label{step:alg-no-coarsening}
    \STATE $\hat{V} \leftarrow \textsf{K-Means}(V,\textsf{num\_clusters} \propto D,\textsf{weights} \propto |P^d|)$;\label{step:k-means}
    \STATE $\mathit{sol} \leftarrow  (\hat{V},\bot)$;
\ENDIF
\RETURN $\mathit{sol}$;
\end{algorithmic}
\end{small}
\end{algorithm}

Consider an SSG $(G,R^a,P^d,A,D,P,B)$ and let $\delta \in \NN$.
Given $v,v'\in V$, we write $v \sim_\delta v'$ iff $R^a(v') > \delta$ and $d(v,v')\leq B$. Intuitively, $v \sim_\delta v'$ means that $v$ and $v'$ must belong to the same neighborhood of a coarsening of $G$ in order to satisfy Condition~\ref{cond:hatg-separated}. Let $\approx_\delta$ denote the reflexive, symmetric, and transitive closure of $\sim_\delta$. Since $\approx$ is an equivalence relation, $V/\!\approx_\delta$ is a partition of $V$ (into equivalence classes). Hence, it is a coarsening that maximizes $|\hat{V}|$ (for Condition~\ref{cond:hatg-many-nodes}) while satisfying Condition~\ref{cond:hatg-separated}. To meet Condition~\ref{cond:hatg-remove-poor-neighborhoods}, we sort the neighborhoods in $V/\!\approx$ by $u_{1,0}^{\hat{v},a}$, and remove poor neighborhoods until Condition~\ref{cond:hatg-remove-poor-neighborhoods} holds.

As $u_{1,0}^{\hat{v},a}$ requires solving an NP-hard problem~\cite{junger1995traveling}, we use TSP (Traveling Salesman Problem)-solvers to get lower bounds, and use $\mathsf{best\mbox{-}path}(\hat{v},\delta)$ to refer to the procedure which looks for a shortest path going through all $\delta$-valuable nodes in $\hat{v}$. Hence, the algorithm may fail to find a $\delta$-coarsening although one exists, and instead return a $\hat{\delta}$-coarsening for some greater $\hat{\delta}$. In turn, the resulting coarsening will only be $\epsilon(\hat{\delta})$-tight. On the other hand, the algorithm is efficient, optimizing on $\delta$ with a simple binary search. Moreover, it returns an upper bound on $\delta$, which translates (by Theorem~\ref{theorem:tight-coarsensing}) to a concrete bound on the loss from respecting the coarsening, instead of playing an SSE defense strategy. Lastly, the algorithm solves the single-attacker single-defender game in each neighborhood, as described in Section~\ref{subsec:single-drone-seq-ssg}, to ensure that defending a neighborhood results with a significant utility change for both players.

Lines~\ref{step:call-coarsen}--\ref{step:set-valuable} of Algorithm~\ref{alg:coarsening} return a partition of $V$ that satisfies (i), i.e.\ incentivizing drones to stay in their starting neighborhoods throughout the game.
Line~\ref{step:remove-poor-neighborhoods} removes ``poor'' neighborhoods that the attacker doesn't care about.
Lines~\ref{step:not-enough-neighborhoods}, \ref{step:attacker-can-collect-payload}, \ref{step:attacker-can-collect-battery}, and~\ref{step:defense-presence-meaningful} check if a $\delta$-coarsening exists by checking the other three conditions (ii), (iii), (iv), respectively. The algorithm performs binary search on $\delta$, to find the smallest one for which a $\delta$-coarsening exists, as the SSE approximation error is linear in $\delta$ (as shown in Section~\ref{sec:theoretical-results}). A formal definition of a $\delta$-coarsening is given in Section~\ref{sec:theoretical-results}.

If the condition in Line~\ref{step:alg-no-coarsening} holds, it means that no $\delta$-coarsening exists. In this case, S2D2 uses weighted K-Means~\cite{kerdprasop2005weighted}, which has three advantages: (i) it is efficient and simple; (ii) it leverages the planar structure of the graph, and the coordinate-based location of each vertex in the graph; (iii) it takes the penalties into account, by setting them as the weights.
The parameter $\delta$ can be viewed as a cut-off, where any node with a smaller reward is considered negligible. Hence, S2D2 heuristically assigns $\delta$ as the $|\hat{V}|P$ most rewarding target, so that each neighborhood has $P$ rewards $>\delta$ on average. The number of neighborhoods $|\hat{V}|$ is set to be proportional to the number of available defense drones $D$. 
We test the performance of this algorithm by conducting experiments on real-world cities in Section~\ref{sec:experiments}. Therefore, in what follows, we will seek defense strategies that \emph{respect} a given coarsening, whether it admits the strict theoretical requirements or not, as defined below.

\begin{Definition}[Strategy Respecting a Coarsening]
A defense (attack) strategy \emph{respects} the coarsening $\hat{V}$ when the following conditions are met:
\begin{enumerate}
    \item Every defense (attack) drone stays within its starting neighborhood throughout the game.
    \item Every neighborhood contains up to a single defense (attack) drone.
\end{enumerate}
$\calS^d_{\hat{V}},\calS^a_{\hat{V}}$ denote the sets of pure strategies that respect the coarsening $\hat{V}$, for the defender and the attacker, respectively.
\end{Definition}

\subsection{Single-Attacker Single-Defender Solution}
\label{subsec:single-drone-seq-ssg}
S2D2 approximates SSE for a single-attacker single-defender game within each neighborhood. \revision{Crucially, in the broader multi-drone, multi-neighborhood setting, the defender's presence in a given neighborhood is probabilistic. In large cities with limited defense resources, it is generally expected that neighborhoods are not protected indefinitely. This probability must be taken into account when considering the single drone game within a given neighborhood, and is therefore introduced as an additional input parameter, denoted by $\lambda$.}

\paragraph{Brute Force Solution}
Since the problem is NP-hard\footnote{The problem is NP-hard even for $\lambda=0$, i.e., when solving the optimization problem for the attacker facing no defender. For example, if $P=B=|\hat{v}|$, deciding whether the attacker has a strategy with utility $u=\sum_{\hat{v}\in\hat{V}}{R(\hat{v})}$ is equivalent to deciding whether a Hamiltonian path exists in graph $(\hat{v},E|_{\hat{v}})$.}, we use smart enumeration as $P,B$ are small.\footnote{This assumption is reasonable as most drone attacks take small amounts of time. \revision{For instance, \cite{deb2025drone} tracked all drone flights over The Hague over 8 months and found the average duration to be 298 seconds and the max duration to be 720 seconds.}
} We begin with a naive approach which linearizes the problem. We compute the matrices $U_\lambda^a$, $U_\lambda^d$ of the attacker and defender utility for each pair of pure strategies. \revision{Note that those values depend on $\lambda$, the defender's presence probability}. We then omit any dominated pure strategies, and find SSE $(\xx_d^*,s_a^*)$ in a similar manner to the single-attacker single-defender SSG (cf.~\cite{korzhyk2010complexity}), i.e., 
we enumerate the set of attacker pure strategies, and for each pure strategy $s_a'$, we then solve the following LP that maximizes the defender utility, under the constraint that $s_a'$ is the best response:
\begin{itemize}
    \item Maximize $u^d(\xx_d,s_a',\lambda)$, subject to:
    \begin{enumerate}
        \item $\xx_d \in \calC_D$ -- Now it is the set of combinations over all non-dominated defense strategies.
        \item For each $s_a \in \calS^a$, $u^a(\xx_d,s_a,\lambda) \le u^a(\xx_d,s_a',\lambda)$.
    \end{enumerate}
\end{itemize}
Note that $u^d(\xx_d,s_a,\lambda)$ is a linear combination of values from  $U_\lambda^d$, according to $\xx_d$, and the same holds for $u^a$ and $U_\lambda^a$.

Finally, we pick $\xx_d^*,s_a^*$ that maximizes the defender utility. The complexity is $|\calS^a|\times \LP(|\calS^d|,|\calS^a|+|\calS^d|)$. Namely, for each attacker strategy, we solve a linear program with $|\calS^d|$ variables and $|\calS^a|+|\calS^d|$ constraints. Next, we improve by reducing the relevant strategy space for both the attacker and the defender.

\paragraph{Reducing the Attacker Strategy Space}
By narrowing down the strategy space, we expect to move away from the optimal solution and trade-off run time vs. solution quality.

When $\lambda$ is small, we know that $s_a^*$ is more greedy, as the $(1-\lambda)$ term  dominates. Hence, $s_a^*$  largely ignores the defender. This may eliminate most of the attacker's possible strategies.  $\lambda$ should anyway be small when there are sufficiently many neighborhoods that are attractive to the attacker. When this is not true, the problem is smaller, and S2D2 takes a random sample of the strategy space, trading-off runtime and quality of the solution. So we may only enumerate a smaller space of possible attacker strategies. 
To some extent, this can be done without damaging performance. Suppose $s_a,s_a' \in \calS^a$ so that $u^a(\bot,s_a') \le (1-\lambda)\cdot u^a(\bot,s_a)$. Then for any strategy $\xx_d\in\Delta\calS^d$, $u^a(\xx_d,s_a',\lambda) \le u^a(\bot,s_a') \le (1-\lambda)\cdot u^a(\bot,s_a) \le u^a(\xx_d,s_a,\lambda)$. Therefore, if the attacker's utility from $s_a$ when facing a defender with probability $\lambda$ is at least the utility from playing $s_a'$ against no defender, we can strike out the strategy $s_a'$, as $s_a$ strictly dominates it.

When there is a small subset of crucial nodes in each neighborhood which are far apart so that an attacker drone must follow an almost optimal path in order to pass through a couple of them, the number of candidate attacker strategies drops significantly. When this is not the case though, S2D2 randomly samples from the large space of possible strategies. This is depicted in Alrogithm~\ref{alg:scan_attack_strategies}.

\begin{algorithm}[h!t]
\caption{\textsf{ScanAttackStrategies}}\label{alg:scan_attack_strategies}
\begin{small}
\begin{algorithmic}[1]
\revision{
\REQUIRE 
A weighted, undirected graph $(\hat{v},E|_{\hat{v}},R|_{\hat{v}})$;\\
\quad\ \ Defender presence probability $\lambda$;\\
\quad\ \   attacker drone battery capacity and payload $B,P \in \NN$;\\
\quad\ \ Threshold $\textsf{th}$ on the number of output attack strategies;
\ENSURE 
attacker drone possible strategies $S^a \subset \calS^a$.
\STATE Compute $u^a_{\max} = \max_{s_a \in \calS^a}{u^a_{\hat{v}}(s_a,\bot)}$, the maximal attacker utility at $\hat{v}$ when facing no defender;
\STATE Set $S^a := \{s_a \in \calS^a | u^a_{\hat{v}}(s_a,\bot) \ge (1-\lambda) u_{\max}^a\}$;
\IF{$|S^a| > \textsf{th}$}
    \RETURN A random sample of size $\textsf{th}$ from $S^a$;
\ENDIF
\RETURN $S^a$;
}
\end{algorithmic}
\end{small}
\end{algorithm}

\paragraph{Narrowing Down Defender Strategy Space}
As the attacker's set of best response pure strategies is now small, the dominating set of defense strategies is also expected to be small. Algorithm~\ref{alg:scan_defense_strategies}'s goal is to output a small subset of dominating defense strategies, as explained below.\footnote{Narrowing down the defender strategy space is complex: as there are multiple possible attack strategies, the defender might want to cover many of them with a single strategy, rather than considering the optimal strategy against every potential attack strategy.}

Suppose the defender and attacker drones' starting positions are $v_d,v_a$, respectively, and $S^a$ is the (narrowed) set of possible attack strategies starting from $v_a$. For each strategy $s_a\in S^a$, up to $P$ nodes are attacked, $v_1(s_a),\ldots,v_P(s_a)$, at times $t_1(s_a),\ldots,t_P(s_a)$. To further reduce runtime, we may only consider targets with a significant (i.e.\ less than $ -\delta$) defender penalty. 

\begin{algorithm}[h!t]
\caption{\revision{\textsf{catch}}}\label{alg:catch}
\begin{small}
\begin{algorithmic}[1]
\revision{
\REQUIRE 
An undirected graph $(\hat{v},E|_{\hat{v}})$;\\
\quad\quad~\ \  Attacker pure strategy $s_a$;\\
\quad\quad~\ \  Defender start position $v_d\in\hat{v}$;
\ENSURE \revision{$1\le i \le P+1$, the index of the first target the defender is able to protect};\\
\quad\quad\ \ \revision{($i=P+1$ indicates the defender is not in time to protect any target)}

\STATE Define $(v_1(s_a),\ldots,v_{P'}(s_a))$ as the ordered list of targeted nodes in $s_a$;
\STATE Remove nodes with an absolute penalty less than $\delta$;
\STATE Re-index the remaining nodes, and update $P'$;
\STATE Define $(t_1(s_a),\ldots,t_{P'}(s_a))$ as the planned time steps for each node to be attacked;
\FOR{$i$ \textbf{from} 1 \textbf{to} $P'$}
    \STATE Find shortest path $\pi_i$ from $v_d$ to $v_i(s_a)$;
    \STATE Denote its length by $t_i^d$;
    \IF{$t_i^d \le t_i(s_a)$}
        \RETURN $i$;
    \ENDIF
\ENDFOR
\RETURN $P+1$;
}
\end{algorithmic}
\end{small}
\end{algorithm}

\begin{algorithm}[h!t]
\caption{\textsf{ScanDefenseStrategies}}\label{alg:scan_defense_strategies}
\begin{small}
\begin{algorithmic}[1]
\REQUIRE 
A weighted, undirected graph $(\hat{v},E|_{\hat{v}},R|_{\hat{v}})$;\\
\quad\ \   attacker drone battery capacity and payload $B,P \in \NN$;\\
\quad\ \  
attacker drone start position $v_a$;\\
\quad\ \  defense drone start position $v_d$;\\
\quad\ \  attacker drone possible strategies $S^a \subset \calS^a$.
\ENSURE 
Defense drone possible strategies $S^d \subset \calS^d$.
\IF{$\abs{S^a}=1$}
\RETURN $\textsf{catch}(v_d,S^a)$; \COMMENT{Compute first strike feasible to prevent (and respective path).}
\ENDIF
\STATE $\textsf{\bf init}$ $T$;
\FOR{\textbf{each} $v_d' \in N(v_d)\cup\{v_d\}$ and $s_a'\in S^a$}
    \STATE $T[v_d',s_a']\leftarrow\textsf{catch}(v_d',s_a')$;
\ENDFOR
\STATE $\text{next\_step} \leftarrow \textsf{prune}(T)$; \COMMENT{Omit dominated neighbors}
\STATE $S^d \leftarrow \emptyset$;
\FOR{\textbf{each} $v_d'\in\text{next\_step}$}
    \STATE $\textsf{\bf init}$ $T_S$;
    \FOR{\textbf{each} $v_a'\in N(v_a)\cup\{v_a\}$} \STATE \COMMENT{DFS visit}
    \STATE $\textsf{update}(S^a)$; \COMMENT{Consider only strategies in $S^a$ that goes from $v_a$ to $v_a'$}
    \STATE $\tilde{S}^d \leftarrow \textsf{ScanDefenseStrategies}
    (v_d',v_a',B-1)$;
    \STATE $T_S[v_a']\leftarrow\tilde{S}^d$
    \ENDFOR
\STATE $S^d \leftarrow S^d \cup \textsf{lift\_strategies}(v_d',T_S)$; \COMMENT{Combine strategies from recursion}
\ENDFOR
\RETURN $S^d$.
\end{algorithmic}
\end{small}
\end{algorithm}

We can then compute for the defender, the minimal time to get to each such node $(t_1^d,\ldots,t_P^d)$, and let $1\le i \le P$ be the first target the defender can protect. This is the output of $\textsf{catch}(v_d,s_a)$ \revision{(Algorithm~\ref{alg:catch})} which corresponds to the best strategy when the attacker's \revision{pure} strategy is known.\footnote{Note that we only find the first node targeted by the attacker that is feasible to protect, not the first node we can catch the attacker at. This is because following a longer path may cover other potential paths the attacker may take, without losing utility from not following the shortest path, when considering the given attacker path.}

It should be observed that, at each time point in Algorithm~\ref{alg:scan_defense_strategies}, it suffices to decide the set of possible next steps for the defender. We can then explore these using DFS, and eventually return all non-dominated pure strategies. The more steps the attacker takes (recursion depth), the narrower its strategy space gets, so the search should converge relatively quickly.

The \textsf{ScanDefenseStrategies} algorithm has $3$ steps:
\begin{enumerate}
    \item For each possible next step $v_d'\in N(v_d)\cup\{v_d\}$, compute $\textsf{catch}(v_d,s_a)$ for each $s_a\in S^a$. Then prune any dominated strategy (where for any strategy of the attacker, it catches the attacker later or at the same targeted node).
    \item For each $v_d'$ that survived, and for each possible next attacker step $v_a'$, recursively call \textsf{ScanDefenseStrategies} and retrieve the set $T_S[v_a']$ of non-dominated pure strategies (with $B-1$, and updated $S^a$).
    \item Lastly, lift pure strategies from $(v_d',\cdot)$ to a strategy from $(v_d,v_a)$ of the form: ``go to $v_d'$, and for each possible attacker next step $v_a'$, pick a pure strategy from $T_S[v_a']$''.
\end{enumerate}
The recursion ends either either when $B=0$ or when the attacker strategy space is a singleton --- we then use \textsf{catch}. To save space, we leverage dynamic programming, and start by solving the problem for $B=0$ and increment the battery capacity by 1 at every step, solving each instance problem once. After this, we get a reduced matrix $U_\lambda$, which only considers a smaller subset of defense and attack strategies.

\subsection{The Meta Game: Multi-Drone Solution}
\label{subsec:seq-ssg-approx-solution}
\revision{The third step in the S2D2 algorithm is to solve the Meta Game, once we know the optimal defender strategy for each neighborhood. The MetaGame looks at the question of which neighborhoods to deploy a defense drone to. This is done via a mixed strategy.  The pseudo-code of the MetaGame is in Algorithm~\ref{alg:meta-game} and can be described at a high level as follows.

\begin{enumerate}
\item We translate the meta-game of allocating defense drones to neighborhoods into a multi-resource attacker defender SSG with nonlinear utilities.
\item An approximation of utilities is given as an input. This is a piece-wise linear approximation derived from solving the single-drone neighborhood game for different $\lambda$ values.
\item Next, we translate the SSG problem into a MIP.
\item We then build on past work~\cite{DeAngelis+1971+503+510} to translate the MIP into a MILP. Their technique allows to replace piecewise linear functions with a linear one by adding a linear number of continuous variables and a logarithmic number of binary variables.
\item Finally, we solve the above MILP and extract the attacker and defender solutions.
\end{enumerate}

\begin{algorithm}[t!]
\caption{\textsf{SolveMetaGame}}\label{alg:meta-game}
\begin{small}
\begin{algorithmic}[1]
\revision{
\REQUIRE 
A set of neighborhoods $\hat{V}$; \\
\quad\ \  numbers of attacker and defender drones $A,D\in \NN$; \\
\quad\ \ (Approximate) attacker and defender utility functions $\{u_{\hat{v}}^a(\lambda),u_{\hat{v}}^d(\lambda)\}_{\hat{v}\in\hat{V}}$; \\
\ENSURE  An SSE $\langle \hat{p}_d,\hat{f}_a \rangle$, where: \\
\quad\ \ $\hat{p}_d \in \calC_D^{\hat{V}}$, a coverage vector of $D$ drones over the neighborhoods of $\hat{V}$; \\
\quad\ \ $\hat{f}_a$ maps each attacker drone to a neighborhood of $\hat{V}$; \\
\STATE Compute piece-wise linear approximations of the attacker and defender utility functions $\tilde{u}_{\hat{v}}^a(\lambda),\tilde{u}_{\hat{v}}^d(\lambda)$;
\STATE Initialize a MIP with the objective of maximizing $\sum_{\hat{v}}{x_a(\hat{v}) \cdot \tilde{u}^d_{\hat{v}}(\xx_d(\hat{v}))}$;
\STATE Add constraints on attacker and defender resources: $\sum_{\hat{v}}{x_a(\hat{v})}=A, \sum_{\hat{v}}{\xx_d(\hat{v})}=D$;
\STATE Require variables $x_a(\hat{v}) \in \{0,1\}$ to be binary and limit continuous variables $0\le \xx_d(\hat{v})\le 1$;
\STATE Add a continuous variable $\theta_a$ for attacker threshold;
\STATE Add the following inequality constraints, forcing attacker best response:
\FOR{\textbf{each} neighborhood $\hat{v} \in \hat{V}$}
    \STATE (\romannumeral 1) $\tilde{u}^a_{\hat{v}}{\xx_d(\hat{v})} \ge x_a(\hat{v}) \cdot \theta_a$;
    \STATE (\romannumeral 2) $\tilde{u}^a_{\hat{v}}{\xx_d(\hat{v})} \le (1-x_a(\hat{v})) \cdot \theta_a + x_a(\hat{v}) \tilde{u}^a_{\hat{v}}(0)$;
\ENDFOR
\STATE Linearize the above MIP (using~\cite{DeAngelis+1971+503+510}).
\STATE Let $\langle \hat{p}_d, \hat{f}_a \rangle$ be MILP solution.
\RETURN $\langle \hat{p}_d, \hat{f}_a \rangle$;
}
\end{algorithmic}
\end{small}
\end{algorithm}

Next, we delve into the technical details of the high-level structure of the MetaGame algorithm described above.}

Recall that we only consider strategies that respect $\hat{V}$, i.e., drones stay within their starting neighborhood, and there is up to one attacker and one defender per neighborhood. Therefore, an attacker pure strategy naturally decomposes into an injection $\hat{f}_a: \{1,\ldots,A\} \mapsto \hat{V}$ mapping each attacker drone to a neighborhood which it will attack, and for each drone $1\le i \le A$, a pure strategy $\calS_i^a=\calT_{P,B} \times \calP_{B}^{\hat{f}_a(i)}$, where $\calP_{B}^{\hat{f}_a(i)}$ considers only paths within the neighborhood $\hat{f}_a(i)$.

Similarly, each defender strategy decomposes into a mapping of each defense drone to a neighborhood, and a strategy within this neighborhood.  Since all defense drones are identical, when considering mixed defense strategies, it suffices to specify (i) within each neighborhood $\hat{v}\in\hat{V}$ a mixed single-drone defense strategy $\hat{x}^d(\hat{v})\in\Delta(\calS_{\hat{v}}^d)$; (ii) for each neighborhood the probability of it being protected, as a \emph{coverage vector }$\hat{p}_d\in\calC_D^{\hat{V}}$, where $\calC_D^{\hat{V}} := \{\xx \in [0,1]^{|\hat{V}|}: \sum_{\hat{v}\in \hat{V}} x_{\hat{v}} \le D \}$. Therefore, the defender mixed strategy space decomposes to $\Delta(\calS^d) = \calC_D^{\hat{V}} \times \prod_{\hat{v}\in\hat{V}}{\Delta(\calS^d_{\hat{v}})}$.

When solving the single-attacker single-defender instance for a neighborhood $\hat{v}$, $\hat{p}_d(\hat{v})$ denotes the probability $\lambda_{\hat{v}}$ that ``a defender is in the hood''.  $\hat{p}_d$ is a coverage vector, representing the probability of presence of a defense drone in each neighborhood. Recall that any vector with entries in $[0,1]$ that sums up to $\le D$ is feasible to implement with some mixed strategy of assigning defense drones to neighborhoods.

Given defender (resp. attacker) strategy $(\hat{p}_d,\hat{x}_d)$ (resp. $(\hat{f}_a,\hat{s}_a)$), where $\hat{s}_a(\hat{f}_a(i))=(T_a^i,\pi_a^i)\in\calT_{P,B}\times\calP_{B}$, the expected utility is the sum of expected utilities from each neighborhood $\hat{f}_a(i)$ attacked, for $1\le i\le A$. The expected utility from neighborhood $\hat{v}$ is the average of the sum of the rewards over the attacker drone set of chosen targets, and the utility when facing a single defender with strategy $\hat{x}_d(\hat{v})$, weighted by $\hat{p}_d(\hat{v})$. That is, for $u\in\{u^a,u^d\}$:

$$u\left(\langle \hat{p}_d, \hat{x}_d \rangle, \langle \hat{f}_a, \hat{s}_a \rangle\right) = \sum\limits_{\substack{1\le i \le A\\ \hat{v}_i = \hat{f}_a(i)}} \biggl[\hat{p}_d(\hat{v}_i)\cdot u\left(\hat{x}_d(\hat{v}_i),\hat{s}_a(\hat{v}_i)\right) +\left(1-\hat{p}_d(\hat{v}_i)\right) 
u(\bot,\hat{s}_a(\hat{v}_i))
\biggr]$$

Thus, as the probability a defender is ``in the hood'' $\hat{p}_d(\hat{v})$ decreases, the attacker drone is better off taking a greedy action. This implies that it is not sufficient to compute SSE for single defender attacker game within each neighborhood to solve the overall multi-drone game. Focusing on a neighborhood, we can extend the utility definition $u(\xx_d,s_a,\lambda):=\lambda u(\xx_d,s_a) + (1-\lambda) u(\bot,s_a)$, to consider the probability
$\lambda$, denoting the probability a defender is in the hood.
This may remind the reader of the SSG model with penalties where, even when the attacker is caught, it gets a penalty $P>0$. We may effectively tune the parameters of the game so that rewards are scaled by $\lambda$, and the penalties are the rewards scaled by $(1-\lambda)$. Section~\ref{subsec:single-drone-seq-ssg} discusses how to approximate SSE in a single-attacker single-defender game with parameter~$\lambda$. 
We next focus on allocation to neighborhoods and assume an oracle returns (an approximation of) optimal $\hat{x}_d^*,\hat{s}_a^*$ strategies within each neighborhood given $\hat{p}_d,\hat{f}_a$. This is possible as Section~\ref{subsec:single-drone-seq-ssg} shows how to implement the oracle, and Lemma~\ref{lemma:sse-continuous} below shows that an approximation suffices. Therefore we get for $u\in\{u^a,u^d\}$:

$$u(\hat{p}_d,\hat{f}_a)=u\left(\langle \hat{p}_d, \hat{x}_d^* \rangle, \langle \hat{f}_a, \hat{s}_a^* \rangle\right)= \sum_{i=1}^{A} u\left(\hat{x}_d^*(\hat{f}_a(i)),\hat{s}_a^*(i),\hat{p}_d(\hat{f}_a(i))\right).$$

We next pick a distribution $\hat{p}_d$ which minimizes the utility above when $\hat{f}_a$ is the best response to $\hat{p}_d$. Hence, we get a typical SSG, with an attacker with multiple ($A$) resources, with one important detail: the utility of each neighborhood $\hat{v}$ is not necessarily linear with the coverage $\hat{p}^d(\hat{v})$, although it is monotonic decreasing.

\subsection{Generalization of Multi-Resource SSGs}
\label{sec:multi-drone-milp}
In this section, we show how to generalise the work of~\cite{korzhyk2011security} to handle a non-linear dependency of the attacker and utility functions on $\xx_d$, the defense probability on each target.

When both utilities are linear with $\xx_a$ \emph{and} $\xx_d$, there is a complete characterization of the Nash equilibrium of the game. Indeed, best-responding simply means attacking (defending) the $D$ ($A$) targets with the highest (marginal) utility for the defender (attacker). Therefore:

\begin{Lemma}
If $u^a,u^d$ are linear with $\xx_d$ and $\xx_a$, let $v^d(t,x_a(t))=a_t (R^d(t)-P^d(t))$ be the defender marginal utility from attacking target $t$. Then $(\xx_d,\xx_a)$ is a Nash equilibrium iff there exist thresholds $\theta_a,\theta_d$ such that:
\begin{itemize}
    \item $\xx_a \in \BR(\xx_d)$. Equivalently:
    \begin{itemize}
        \item $u^a(t,\xx_d(t)) < \theta_a \Rightarrow x_a(t)=0$.
        \item $u^a(t,\xx_d(t)) > \theta_a \Rightarrow x_a(t)=1$.
        \item $\sum_{t}{x_a(t)}=A$.
    \end{itemize}
    \item $\xx_d \in \BR(\xx_a)$. Equivalently:
    \begin{itemize}
        \item $v^d(t,x_a(t)) < \theta_d \Rightarrow \xx_d(t)=0$.
        \item $v^d(t,x_a(t)) > \theta_d \Rightarrow \xx_d(t)=1$.
        \item $\sum_{t}{\xx_d(t)}=D$.
    \end{itemize}
\end{itemize}
\end{Lemma}

When $u^d$ is linear with $\xx_d$, the defender's marginal utility $\frac{\partial u^d}{\partial \xx_d}$ is a constant, and in particular, is independent of $\xx_d$. Therefore, the utility the defender gets from protecting target $t$ with probability ``budget'' $\xx_d(t)$ is $\xx_d(t)\cdot v^d(t,x_a(t))$. Therefore, best responding means first covering the top $D$ targets, and when there are ties for the $D$\textsuperscript{th} place, any randomization over the corresponding targets will result in a valid best response.

However, when $\frac{\partial u^d}{\partial \xx_d}$ is a function of $\xx_d$, this is not the case any longer. Indeed, the above condition would be necessary, suggesting $\xx_d$ to be a \emph{local} maximum of $u^d$, as otherwise (assuming $u^d$ is continuously differentiable) one could make small changes and increase the defender's utility. Nevertheless, it will not ensure a \emph{global} maximum of $u^d$, meaning a best response. If $u^d$ was concave with $\xx_d$, any local maximum would also be global and therefore~\cite{korzhyk2011security}'s algorithm would still work. Unfortunately, we cannot make such an assumption in our game.

Nevertheless, we are not interested in computing a Nash equilibrium, but a SSE. Therefore, we first show that the criterion for the attacker to best respond remains intact:

\begin{Lemma}
\label{lemma:attacker-br-nonlinear-u-x-d}
Assume $u^a$ is linear with $\xx_a$, and that $|V| > A$, and that $P^a(t) < R^a(t)$ for every target $t$. Let $\xx_d$ be a defense mixed strategy. Then $\xx_a \in \BR^a(\xx_d)$ iff there exist a threshold $\theta_a$ such that:
\begin{itemize}
    \item $u^a(t,\xx_d(t)) < \theta_a \Rightarrow x_a(t)=0$.
    \item $u^a(t,\xx_d(t)) > \theta_a \Rightarrow x_a(t)=1$.
    \item $\sum_{t}{x_a(t)}=A$.
\end{itemize}
\end{Lemma}
\begin{proof}
\noindent ($\Leftarrow$) Suppose $\xx_a$ admits the above conditions. Then, the marginal attacker utility from attacking target $t$ is $u^a(t,\xx_d(t))$, therefore, independent of $\xx_a$. Hence, best responding would first protect the targets with the highest attacker utility given $\xx_d$, and any randomization over the $A$\textsuperscript{th} target will result with the same overall attacker utility.
($\Rightarrow$) Assume by way of contradiction that one of the above conditions doesn't hold. If there are two targets $t_1,t_2$ such that $u^a(t_1,\xx_d(t_1)) < u^a(t_2,\xx_d(t_2))$, and $0<x_a(t_1),x_a(t_2)<1$, the attacker's utility will increase by shifting attacker probability mass from $t_1$ to $t_2$ until either $x_a(t_1)$ gets to $0$ or $x_a(t_2)$ gets to $1$. Last, if not all of the attacker resources are utilized, we can increase the attack probability on all targets, and increase the attacker's overall utility as well. Note that this is why we need to assume $\abs{V} > A$ and $P^a(t) < R^a(t)$ on every target $t$.
\end{proof}

Next, we opt to transform the SSE computation into a mixed integer program, which is a well-studied problem. We start from the following optimization problem:

\begin{flalign}
\label{eqn:meta-game-mip}
\text{maximize: }& 
\sum\limits_{t} {x_a(t)\cdot u^d(t,\xx_d(t))} &&  \\
\text{subject to: }& \sum\limits_{t}{x_a(t)} = A, \sum\limits_{t}{\xx_d(t)} = D, && \nonumber\\
& x_a(t) \in \{0, 1\}, 0\le \xx_d(t) \le 1, && \nonumber\\
& u^a(t,\xx_d(t)) \ge x_a(t) \cdot \theta_a, && \nonumber\\
& u^a(t,\xx_d(t)) \le (1-x_a(t))\cdot \theta_a + x_a(t) R^a(t). && \nonumber
\end{flalign}

Evidently, a solution to the above MIP is SSE. Indeed, the objective is to maximize the defender's utility, over all possible coverage vectors $\xx_d$ of the defender. Demanding $\sum_{t}{\xx_d(t)}=D$ is okay because the utilities are monotonically increasing. Finally, in the SSE framework, we can assume that the attacker's strategy is pure, that is, $x_a(t)\in\{0,1\}$ which enables us to write the condition for the attacker to best respond (described in Lemma~\ref{lemma:attacker-br-nonlinear-u-x-d}) with linear inequalities over the variables $\xx_a(t)$.

The only problem is that $u^d(t,\xx_d(t))$ and $u^a(t,\xx_d(t))$ are non-linear w.r.t.\ $\xx_d(t)$ in general. However, this can be handled using standard techniques to approximate the utility functions with piece-wise linear approximations $\tilde{u}^d,\tilde{u}^a$. This is inevitable as we don't have closed form formulas for the utilities --- rather, they are derived from the algorithm for the single attacker/single defender drone problem in Step 2 of the S2D2 algorithm). We refer to~\cite{DeAngelis+1971+503+510} for an overview of the technique. In principle, we can add a logarithmic number of integer variables, and linear number of continuous variables, and replace the utilities with linear expressions using the new variables.

We can bound the error from approximating the utilities by the following lemma:

\begin{Lemma}
\label{lemma:sse-continuous}
Let $G=(\hat{V},A,D,u^a,u^d)$ be a (non-sequential) attacker SSG. Let $\epsilon >0$ and let $\tilde{u}^a,\tilde{u}^d$ be different attacker and defender utility functions, such that $\|(u^a,u^d)-(\tilde{u}^a,\tilde{u}^d)\|_\infty < \epsilon$. That is, on every pair of strategies $(\xx_d,s_a)$, the attacker and defender utility outputs differ by up to $\epsilon$, using the other utility functions. Then if $(\xx_d,s_a)$ is an $\epsilon$-SSE of $G$, it is also a $2\epsilon$-approximate SSE of $\tilde{G}$ where the utilities are replaced with $\tilde{u}^a,\tilde{u}^d$.
\end{Lemma}

\begin{proof}
Indeed, assume that for any pair of strategies, $(\xx_d,\xx_a)$, we have that $|u^a(\xx_d,\xx_a) - \tilde{u}^a(\xx_d,\xx_a)| < \epsilon$ and $|u^d(\xx_d,\xx_a) - \tilde{u}^d(\xx_d,\xx_a)| < \epsilon$.

Let $(\xx_d,s_a)$ be an $\epsilon$-SSE with respect to $(u^a,u^d)$. Then $s_a \in \BR_{\epsilon,u^a,u^d}^d(\xx_d)$, and therefore, $s_a \in \BR_{2\epsilon,\tilde{u}^a}^a(\xx_d)$. Thus:
$$ \tilde{u}^d(\xx_d, \BR_{2\epsilon,\tilde{u}^a,\tilde{u}^d}^d(\xx_d)) \ge \tilde{u}^d(\xx_d,s_a).$$

Next, let $s_a' \in \BR_{2\epsilon,\tilde{u}^a,\tilde{u}^d}^d(\xx_d)$. Then, the above inequality says $\tilde{u}^d(\xx_d,s_a') \ge \tilde{u}^d(\xx_d,s_a)$. Analogously, let $(\tilde{\xx}_d,\tilde{s}_a)$ be an $\epsilon$-SSE with respect to $(\tilde{u}^a,\tilde{u}^d)$, and let $\tilde{s}_a' \in \BR_{2\epsilon, u^a,u^d}^d(\tilde{\xx}_d)$. Then $u^d(\tilde{\xx}_d,\tilde{s}_a') \ge u^d(\tilde{\xx}_d,\tilde{s}_a)$. Thus:

\begin{align*}
\tilde{u}^d(\xx_d,s_a') - \tilde{u}^d(\tilde{\xx}_d,\tilde{s}_a) \ge 
\tilde{u}^d(\xx_d,s_a) - \tilde{u}^d(\tilde{\xx}_d,\tilde{s}_a) \ge \\
u^d(\xx_d,s_a) - u^d(\tilde{\xx}_d,\tilde{s}_a) -2\epsilon \ge 
u^d(\xx_d,s_a) - u^d(\tilde{\xx}_d,\tilde{s}_a') -2\epsilon \ge 0-2\epsilon.
\end{align*}

Finally, since $s_a\in \BR_{\epsilon,u^a,u^d}^d(\xx_d)$, $s_a \in \BR_{2\epsilon,\tilde{u}^a}^a(\xx_d)$, as desired. Therefore, $(\xx_d,s_a)$ is a $2\epsilon$-approximate SSE with respect to $(\tilde{u}^a,\tilde{u}^d)$.
\end{proof}

Finally, we can use standard techniques, such as the one described in~\cite{oral1992linearization}, to linearize the resulted MIP, and solve a MILP.

\section{Theoretical Analysis: SSE Approximation}
\label{sec:theoretical-results}
In this section, we prove that if Algorithm~\ref{alg:coarsening} outputs a $\delta$-coarsening, then it is an $\epsilon(\delta)$-approximate SSE.

First, a formal definition of a $\delta$-coarsening is provided in Definition~\ref{def:delta-coarsening}. While this definition provides a precise framework, it is somewhat restrictive, \revision{and the choice of constants may impose limitations}. We stress that this definition is only needed for the rigorous correctness proof of S2D2 (Theorem~\ref{theorem:tight-coarsensing}). \emph{Nevertheless, it is important to note that S2D2 yields good results in practice on real-world large-scale cities, even if such a $\delta$-coarsening does not exist, as demonstrated via exhaustive experimentation described in Section~\ref{sec:experiments}.}

\begin{Definition}[$\delta$-Coarsening]
\label{def:delta-coarsening}
Let $G_{\hat{v}}=(\hat{v},E|_{\hat{v}})$ be some neighborhood. We denote by $u_{A,D}^{\hat{v},d}$ ($u_{A,D}^{\hat{v},a}$), the (maximal) utility of a defender (an attacker) at SSE in $G_{\hat{v}}$ given $A$ attacker drones and $D$ defense drones.
Let $\delta > 0$. A $\delta$-coarsening $\hat{V}$ is a coarsening that satisfies the following conditions:

\begin{enumerate}
    \item \label{cond:hatg-separated} For each $\hat{v}\in\hat{V}$, $v' \not\in \hat{v}$, and $v\in\hat{v}$ with $R^a(v)>\delta$, it is the case that $d(v,v') > B$, where $d$ is shortest path length.
    
    \item \begin{enumerate}
        \item \label{cond:hatg-many-nodes} Number of neighborhoods $|\hat{V}| > 4\max\{A,D\}$.
        \item \label{cond:hatg-remove-poor-neighborhoods}
    For each $\hat{v},\hat{v}'\in \hat{V}$: $\frac{4}{3}u_{1,0}^{\hat{v},a} >
    u_{1,0}^{\hat{v}',a} - \delta P$.
    \end{enumerate} 
    
    \item \label{cond:hatg-single-attack-drones} For each $\hat{v}\in\hat{V}$: $u_{A,0}^{\hat{v},a} < u_{1,0}^{\hat{v},a} + \delta AP$.
    
    \item \label{cond:hatg-single-defense-drones} For each $\hat{v}\in\hat{V}$: $\frac{3}{64|\hat{V}|} u_{1,0}^{\hat{v},a} > u_{1,1}^{\hat{v},a}$ and $\frac{3}{8|\hat{V}|} |u_{1,0}^{\hat{v},d}| > |u_{1,1}^{\hat{v},d}| + \delta P$.
\end{enumerate}
\end{Definition} 

Conditions~\ref{cond:hatg-separated}-\ref{cond:hatg-single-defense-drones} formalize conditions (\romannumeral 1)-(\romannumeral 4) in Section~\ref{subsec:coarsening} respectively.
Condition~\ref{cond:hatg-separated} suggests that it takes a prohibitive amount of time to move from any node outside a neighborhood into a valuable node within that neighborhood. 
\footnote{We don't require $R^a(v')> \delta$ because the goal of defense drones is to catch the attacker before it causes more damage. Therefore, if there is a node $v\in\hat{v}$ with $R^a(v)=0$, that is close to valuable nodes of multiple different neighborhoods, placing a defense drone at $v$ could be a good strategy. After the attacker places her drones, the defensive drone will decide which neighborhood to go to in order to catch the attacker.}
This condition incentivizes drones to stay within their starting neighborhoods throughout the game. It is also a practical political reality --- city security officials need to be \emph{seen} to be distributing defensive assets in a fair way across the city rather than appearing to give ``preference'' to certain places, even if they are high utility locations. 

Condition~\ref{cond:hatg-many-nodes} suggests that there are not enough defense/attack drones to protect/attack each neighborhood with probability $\ge 1/4$, as security resources are limited. If not, one may consider partitioning the neighborhoods further, though this may violate Condition~\ref{cond:hatg-separated}. This condition incentivizes the attacker to be more greedy, as neighborhoods with no defender with probability $\ge 3/4$ are sure to exist, and Condition~\ref{cond:hatg-separated} ensures defender drones will not reach an unprotected neighborhood in time. In turn, Condition~\ref{cond:hatg-remove-poor-neighborhoods} says that since there are many neighborhoods, the attacker will not go to a low value neighborhood regardless of the defense strategy, and therefore there is no reason for defending it either. As a result, we may ignore this neighborhood altogether, and simplify the graph.\footnote{For this reason, we do not require the coarsening $\hat{V}$ to be a partition of $V$ (i.e., $\bigcup_{i\in [1,k]}\hat{v}_j\neq\hat{V}$).}

As for Condition~\ref{cond:hatg-single-attack-drones}, note that $u_{A,0}^{\hat{v},a} \le A \cdot u_{1,0}^{\hat{v},a}$ always holds. However, when there is variability in the rewards and valuable rewards are sparse, we expect a smaller gap between the two, since one cannot exploit the same target twice. In particular, Condition~\ref{cond:hatg-single-attack-drones} holds if a single attacker can collect all rewards in $\hat{v}$ with $R^a(\cdot) > \delta$. This should be the case when valuable targets are sparse and lie in the interior of neighborhoods rather than near the periphery.
This condition incentivizes the attacker to spread her drones across different neighborhoods to increase the chance of attacking an unprotected neighborhood, as by Condition~\ref{cond:hatg-many-nodes} the chance of a neighborhood being unprotected is not negligible. At the same time, it incentivizes the defender to spread her drones across different neighborhoods to decrease the chance of a successful attack on an unprotected neighborhood. However, this argument holds only if neighborhoods are comparably valuable, which is captured by the following condition.

Condition~\ref{cond:hatg-single-defense-drones} suggests that the presence of a defense drone in a neighborhood makes a significant impact on defender and attacker drone utility. The constraints ensure that the damage done by the attacker facing an undefended neighborhood is significantly larger than the damage done when facing a single defender, where $u_{A,D}^{\hat{v},d}$ is defined analogously to $u_{A,D}^{\hat{v},a}$ (cf. Condition~\ref{cond:hatg-single-attack-drones}) for the defender. The intuition is that a defender can always start at the center of a neighborhood, and thus be able to catch the attacker relatively quickly, whereas by Condition~\ref{cond:hatg-single-attack-drones}, the attacker has enough battery and payload to destroy all crucial spots of a neighborhood when no defender is present.

\paragraph{Sufficiency}
When a $\delta$-coarsening exists, we will show that an $\epsilon$-SSE can be computed efficiently. The reason is that both the attacker and defender are incentivized to spread their drones out across different neighborhoods, which results in a decomposition of the multi-drone game into multiple single-attacker single-defender drone sub-games. To show this, we start with a definition:

\begin{Definition}[$\epsilon$-tight coarsening of a Graph]
We say that $\hat{V}$ is an $\epsilon$-\emph{tight} coarsening if there exist $\xx_d \in \Delta(\calS_{\hat{V}}^d)$ and $s_a \in \calS_{\hat{V}}^a$ such that $(\xx_d,s_a)$ is $\epsilon$-approximate SSE. 
\end{Definition}

We emphasize again that the S2D2 algorithm works even when a tight coarsening does not exist. The above definition is only needed for the formal proof that yields theoretical results on the quality of the output strategy. Specifically, we prove a theoretical bound on the loss of the defender and the attacker caused by restricting their strategies to respect a given coarsening $\hat{V}$, which is $\epsilon$ for an $\epsilon$-tight coarsening. Therefore, we will need to compute $\epsilon(\delta)$ for a $\delta$-coarsening.

Our restrictions on $\delta$-coarsening enable us to prove some nice properties, e.g.\ that a $\delta$-coarsening is always $\epsilon$-tight. To show this, we first analyze the loss of the attacker from respecting a coarsening $\hat{V}$.

\begin{Lemma}
\label{lemma:attacker-respect-coarsening}
Let $\xx_d\in\Delta(\calS^d)$, $s_a\in\calS^a$, $\delta>0$, and $\hat{V}$ be a $\delta$-coarsening. Then, there exists a strategy $s_a' \in\calS^a_{\hat{V}}$ such that $u^a(\xx_d,s_a') \ge u^a(\xx_d,s_a)-\epsilon$, and $u^d(\xx_d,s_a') \ge u^d(\xx_d,s_d)$, for $\epsilon=2\delta AP$.
\end{Lemma}

\begin{proof}
Strategy $s_a'$ is constructed in two steps. First, in $s_a^1$ each drone stays within its starting neighborhoods. Then, in $s_a':=s_a^2$, in addition there is a single attacker drone in each neighborhood. The attacker loss is then bounded by the sum of the losses from the two steps.

First, consider the following strategy $s_a^1$. All attacker drones are placed as in $s_a$, and follow the same paths. Whenever an attacker drone in $s_a$ crosses a neighborhood, the corresponding drone in $s_a^1$ halts. Note that since attacker drones are not coordinated after initial allocation, $s_a^1$ is well defined. Specifically, the strategy of other attacker drones is not affected.
By Condition~\ref{cond:hatg-separated}, when an attacker drone moves across neighborhoods, it can only get negligible rewards. Therefore, following $s_a^1$ may have a utility loss of up to $\delta AP$ for the attacker drone compared to $s_a$. Indeed, for every attack drone and attack payload unit, it could be that in $s_a$ it picked a reward smaller than $\delta$, and in $s_a^1$ it doesn't collect this reward.

Next, assume $A' > 1$ attacker drones were assigned the same neighborhood $\hat{v}'$ in $s_a^1$. For each neighborhood $\hat{v}\in\hat{V}$, let $\lambda_{\hat{v}}$ be the probability that at least one defense drone is allocated to neighborhood $\hat{v}$ at time $t=0$, with respect to $\xx_d$.
Among all neighborhoods that are not occupied with any attacker drone, let $\hat{V}_{A'}$ be the $A'$ least protected neighborhoods with respect to $\xx_d$.
Then at $t=0$, each neighborhood $\hat{v}_{a'}\in \hat{V}_{A'}$ is protected with probability at most $\lambda_{\hat{v}_{a'}}\le \frac{D}{|\hat{V}|-A}$. Indeed, assume for purposes of contradiction that they are protected with probability $> \frac{D}{|\hat{V}|-A}$. Then since those are the least protected, all unoccupied neighborhoods are protected with probability $> \frac{D}{|\hat{V}|-A}$, and there are at least $|\hat{V}|-A$ such neighborhoods. However, even protecting $|\hat{V}|-A$ neighborhoods with probability $\frac{D}{|\hat{V}|-D}$ already requires $D$ defense resources, hence such a defense coverage vector is not feasible $\xx_d \not\in\calC_D$, a contradiction.

Now, by Condition~\ref{cond:hatg-many-nodes}, $\frac{D}{|\hat{V}|-A} \le \frac{D}{3D+A-A} = \frac{1}{3}$. Consider spreading the attacker drones from $\hat{v}$ to $\hat{V}_{A'}$, and play greedily, that is, maximize the attacker utility when facing no defender. Denote this strategy by $s_a^2$.

At worst, the utility of the attacker drones from $\hat{V}_{A'}$ is $\frac{2}{3}\sum_{\hat{v}_{a'}\in\hat{V}_{A'}}{u_{1,0}^{\hat{v}_{a'}}}-\delta PA'$. Indeed, with probability $\ge \frac{2}{3}$, there are no defenders in $\hat{v}_{a'}$ at $t=0$. Assume by way of contradiction that a defender catches an attacker drone in $\hat{v}_{a'}$ at $v_m$, before it reaches some valuable node $v\in\hat{v}_{a'}$ with reward $R^a(v)>\delta$. Then, let $v_a$ be the start node for the attacker and $v_d$ be the start node for the defender. Since they both begin at $t=0$ and meet at $v$, we know that $d(v_d,v_m)=d(v_a,v_m)$. By triangular inequality, $d(v_d,v) \le d(v_d,v_m) + d(v_m,v) = d(v_a,v_m)+d(v_m,v) \le B$, as the attacker moves from $v_a$ to $v_m$ and then to $v$ in less then $B$ units of battery. However, $d(v_d,v) \le B$ and $R^a(v) > \delta$ contradicts Condition~\ref{cond:hatg-separated}. Therefore, the defenders can cause a utility loss for each attacker of up to $P\delta$.

On the other hand, at best, the utility of the $A'$ drones in $s_a^1$ is $u_{A',0}^{\hat{v}'}$. Therefore, the utility loss of the attacker is at most:
$\delta PA' + u_{A',0}^{\hat{v}'} - \frac{2}{3}\sum_{\hat{v}_{a'}\in\hat{V}_{A'}}{ u_{1,0}^{\hat{v}_{a'}}}$. By Condition~\ref{cond:hatg-single-attack-drones}, this is less than $2\delta PA' + u_{1,0}^{\hat{v}',a} - \frac{2}{3}\cdot 2 \min_{\hat{v}_{a'}\in\hat{V}_{A'}}{ u_{1,0}^{\hat{v}_{a'}}}$. Then, by Condition~\ref{cond:hatg-remove-poor-neighborhoods}, the overall utility loss is bounded by $2\delta PA'$. Repeating the above for all neighborhoods in $s_a^1$ that were initially assigned with multiple attacker drones will result with a total attacker utility loss of up to $\epsilon=2\delta PA$.

Thus, a greedy strategy $s_a'=s_a^2\in\calS_{\hat{V}}^a$ of spreading the attacker drones unoccupied neighborhoods and playing greedily, ignoring the defender, results with a negligible loss in utility, regardless of the defense strategy.
\end{proof}

It follows that the attacker respects the coarsening. 

\begin{Corollary}
\label{corollarly:attacker-respects-coarsening}
Let $\hat{V}$ be a $\delta$-coarsening. Then:
\begin{eqnarray*}
\max_{\xx_d\in\Delta(\calS^d_{\hat{V}})}{u^d(\xx_d,\BR^d_{\epsilon}(\xx_d))} \ge \max_{\xx_d\in\Delta(\calS^d_{\hat{V}})}{u^d(\xx_d,\BR^d_{\epsilon,\hat{V}}(\xx_d))} \\
\max_{\xx_d\in\Delta(\calS^d)}{u^d(\xx_d,\BR^d_{\epsilon,\hat{V}}(\xx_d))} \ge \max_{\xx_d\in\Delta(\calS^d)}{u^d(\xx_d,\BR^d(\xx_d))}
\end{eqnarray*}
where $\BR^d_{\epsilon,\hat{V}}(\xx_d):=\argmax\limits_{s_a \in \BR_{\epsilon}(\xx_d) \cap \calS_{\hat{V}}^a}{u^d(\xx_d,s_a)}$ considers only strategies from $\BR_{\epsilon}^a(\xx_d)$ that respect the coarsening, and only then takes the strategy that favors the defender.
\end{Corollary}
\begin{proof}
Let $\xx_d\in\Delta(\calS^d)$ be any strategy. Then by definition, $\BR_{\epsilon,\hat{V}}^a(\xx_d) = \BR_{\epsilon}^a(\xx_d) \cap \calS_{\hat{V}}^a \subseteq \BR_{\epsilon}^a(\xx_d)$. Therefore, $\BR_{\epsilon}^d(\xx_d)$ can only increase defender utility:
\begin{align*}
u^d(\xx_d,\BR^d_{\epsilon}(\xx_d)) = \max_{s_a\in \BR_{\epsilon}^a(\xx_d)}{u^d(\xx_d,s_a)} \ge\\
\max_{s_a\in \BR_{\epsilon,\hat{V}}^a(\xx_d)}{u^d(\xx_d,s_a)}= u^d(\xx_d,\BR^d_{\epsilon,\hat{V}}(\xx_d)).
\end{align*}
Note that $\BR_{\epsilon,\hat{V}}^d(\xx_d) \neq \emptyset$ by Lemma~\ref{lemma:attacker-respect-coarsening}. Since the inequality above holds for every $\xx_d\in\calS^d$, it also holds when maximizing over $\xx_d\in\Delta(\calS_{\hat{V}}^d)$.
Conversely, let $s_a \in \BR^d(\xx_d)$. Then by Lemma~\ref{lemma:attacker-respect-coarsening}, there exist $s_a' \in \BR_{\epsilon,\hat{V}}^a(\xx_d)$ such that $u^d(\xx_d,s_a') \ge u^d(\xx_d,s_a)$. Therefore:
\begin{align*}
u^d(\xx_d,\BR^d_{\epsilon,\hat{V}}(\xx_d)) = \max_{s_a' \in \BR^a_{\epsilon,\hat{V}}(\xx_d)}{u^d(\xx_d,s_a')} \ge\\
\max_{s_a \in \BR^a(\xx_d)}{u^d(\xx_d,s_a)}= \max_{\xx_d\in\Delta(\calS^d)}{u^d(\xx_d,\BR^d(\xx_d))}.    
\end{align*}
Again, since the above holds for every $\xx_d$, it also holds when maximizing over $\xx_d \in \Delta(\calS^d)$.
\end{proof}

We now derive bounds on the utility loss (gain) of the attacker (defender) from increasing the protection of a neighborhood.

\begin{Lemma}
\label{lemma:utility-gain-lower-bound}
Consider any neighborhood $\hat{v}$, protected with some probability $\lambda$. Then, increasing the probability a defender is in the hood $\hat{v}$ by $\eta > 0$, will result in the following lower bounds on the attacker utility loss and defender utility gain:
\begin{enumerate}
    \item $u^{\lambda,a} - u^{\lambda+\eta,a} \ge \eta \cdot u_{1,0}^{a} - u_{1,1}^{a}$.
    \item $u^{\lambda+\eta,d} - u^{\lambda,d} \ge -\eta \cdot u_{1,0}^{d} + u_{1,1}^{d}$.
\end{enumerate}
Recall that $u_{A,D}^a$ is the maximal attacker utility from attacking a neighborhood with $A$ attack drones, facing $D$ defense drones. Analogously, $u_{A,D}^d$ is the maximal defender utility when protecting a neighborhood against $A$ attack drones using $D$ defense drones.
\end{Lemma}
\begin{proof}
The attacker utility is bounded as follows:
$$ (1-\lambda) u_{1,0}^{a} \le u^{\lambda,a} \le (1-\lambda) u_{1,0}^{a} + \lambda u_{1,1}^{a}.$$
The defender utility is bounded as follows:
$$ \lambda \cdot u_{1,1}^{d} + (1-\lambda) u_{1,0}^{d} \le u^{\lambda,d} \le (1-\lambda) u_{1,0}^{d}.$$
Subtracting the lower and upper bounds appropriately yields the above lower bounds.
\end{proof}

\begin{Theorem}
\label{theorem:tight-coarsensing}
A $\delta$-coarsening is $\epsilon$-tight, for $\epsilon=2\delta AP$.
\end{Theorem}
\begin{proof}
First, by Corollary~\ref{corollarly:attacker-respects-coarsening}, we know that:
$$\max_{\xx_d\in\Delta(\calS^d_{\hat{V}})}{u^d(\xx_d,\BR_{\epsilon}^d(\xx_d))} \ge \max_{\xx_d\in\Delta(\calS^d_{\hat{V}})}{u(\xx_d,\BR_{\epsilon,\hat{V}}^d(\xx_d))}.$$

Therefore, we will assume w.l.o.g that the attacker respects the coarsening $\hat{V}$, and compare the RHS with: $$\max\limits_{\xx_d\in\Delta(\calS^d)}{u^d(\xx_d,\BR_{\epsilon,\hat{V}}^d(\xx_d))} \ge \max\limits_{\xx_d\in\Delta(\calS^d)}{u^d(\xx_d,\BR^d(\xx_d))} $$
where the last inequality is again by Corollary~\ref{corollarly:attacker-respects-coarsening}.

That is, it is enough to bound the loss of the defender from respecting the coarsening $\hat{V}$, when assuming that the attacker respects the coarsening. Let $\xx_d\in\Delta\calS^d$.

Similarly to the attacker, consider strategy $\xx_d^1$ where defense drones stop before crossing a neighborhood. Suppose that with probability $>0$, playing $\xx_d$, defense drone $1\le i_D \le D$ catches an attacker drone $1\le i_A \le A$. Let $v_D$ be the start position of defense drone $i_D$, $v_M$ be the meeting point, at time $1 < t \le B$, and consider any path $\pi_A$ of length $B-t$ starting from $v_M$ within the attacker's starting neighborhood, ending at some node $v_A$, for attacker drone $i_A$. Then the path from $v_D$ to $v_A$ is of length $B$, and therefore, by Condition~\ref{cond:hatg-separated}, all of the nodes along $\pi_A$ yield a reward $\le \delta$. Hence, playing $\xx_d^1$ results in an attacker utility gain of up to $\delta AP$. Therefore, to this end, drones stay within their neighborhoods. Again, we stress that $\xx_d^1$ is only defined in case the defense drones' movement is not coordinated after initial allocation.

Next, we claim that the defender should not allocate two defense drones to the same neighborhood. Let $p<1$ be the probability that each neighborhood is protected with a single drone, in $\xx_d^1$. We can consider a strategy $\xx_d^2$ that coincides with $\xx_d^1$ when the coarsening is respected (which happens with probability $p$), and then the utility loss will be bounded by a factor of $1-p<1$. Thus, w.l.o.g, assume $p=0$.

Denote by $\cc \in \calC_{D-1}$ the coverage of the neighborhoods, with respect to $\xx_d^1$. That is, $c_{\hat{v}}$ is the probability that neighborhood $\hat{v}$ is protected with at least a single defender. Denote by $\hat{\BR}^d$ the set of $A$ targets the attacker is attacking when facing $\xx_d^1$ (this is well defined since we may now safely assume the attacker respects the coarsening).

Next, we want to introduce a strategy $\xx_d^2$ that respects the coarsening and has comparable defender utility. Denote by $\cc'\in\calC_D$ the coverage vector for $\xx_d^2$. We are interested in increasing the coverage of $\hat{\BR}^d$, while maintaining the condition that these targets are in $\BR^d$. We have an extra defense unit to allocate, since in $\xx_d^2$, only $D-1$ neighborhoods are covered in each pure strategy.

Let $0<r<1$. By Lemma~\ref{lemma:utility-gain-lower-bound}, using $5r$ defense resources on each target outside of $\hat{\BR}^d$, we can decrease the attacker's utility by at-least $5r \cdot u_{1,0}^{\hat{v},a} - u_{1,1}^{\hat{v}}$.

We will do the same for targets in $\hat{\BR}^d$, increasing their coverage by $3r$. By Lemma~\ref{lemma:utility-gain-lower-bound}, the attacker's utility will decrease by at most $3r \cdot u_{1,0}^{\hat{v},a} + u_{1,1}^{\hat{v}}$. By Condition~\ref{cond:hatg-remove-poor-neighborhoods}, the set $\hat{\BR}^d$ remains the attacker's best response, if $r u_{1,0}^{\hat{v}} > 2\max_{\hat{v}\in\hat{V}}{u_{1,1}^{\hat{v},a}}$. Meanwhile, again by Lemma~\ref{lemma:utility-gain-lower-bound}, the defender's utility on every $\hat{v}\in \hat{\BR}^d$ increases by at least $-3r\cdot u_{1,0}^{\hat{v},d} + u_{1,1}^{\hat{v},d}$.

We therefore take $r=\frac{1}{8|\hat{V}|}$, and require the following:
\begin{enumerate}
    \item $\frac{1}{8|\hat{V}|} \frac{3}{4} \max_{\hat{v}\in\hat{V}}{u_{1,0}^{\hat{v},a}} > 2\max_{\hat{v}\in\hat{V}}{u_{1,1}^{\hat{v},a}}$.
    \item $\frac{3}{8|\hat{V}|} |u_{1,0}^{\hat{v},d}| > |u_{1,1}^{\hat{v},d}| + \delta P$.
\end{enumerate}
The first condition ensures that the attacker's best response set is maintained, and the second condition ensures the defender utility is decreased by up to $\delta A P$ in total. Both of these hold by Condition~\ref{cond:hatg-single-defense-drones}, which states that (and quantifies how) the presence of a defense drone in a neighborhood significantly affects the attacker and defender utilities.

Therefore, strategy $\xx_d^2$ results with up to an additional $\delta AP$ utility loss for the defender. This completes the proof, as:
\begin{align*}
\max_{\xx_d\in\Delta(\calS^d_{\hat{V}})}{u^d(\xx_d,\BR_{\epsilon}^d(\xx_d))} \ge \max_{\xx_d\in\Delta(\calS^d_{\hat{V}})}{u(\xx_d,\BR_{\epsilon,\hat{V}}^d(\xx_d))} \ge \\
\max_{\xx_d\in\Delta(\calS^d)}{u^d(\xx_d,\BR_{\epsilon,\hat{V}}^d(\xx_d))} - \epsilon \ge \max_{\xx_d\in\Delta(\calS^d)}{u(\xx_d,\BR^d(\xx_d))} - \epsilon
\end{align*}
\end{proof}

We now explain how to bound the error of S2D2 algorithm, assuming a $\delta$-coarsening exists. 

\begin{Theorem}
\label{theorem:approximate-sse}
Let $\delta>0$ and assume $\hat{V}$ is a $\delta$-coarsening. Then S2D2 outputs an $\epsilon$-SSE for $\epsilon=2AP\delta+2\epsilon'$, where $\epsilon'$ is an upper bound on the error of the single-attacker single-defender oracle.
\end{Theorem}

Recall that our proposed S2D2 algorithm consists of 3 steps. In the coarsening step, the algorithm outputs, along with the coarsening $\hat{V}$, a parameter $\delta$. By Theorem~\ref{theorem:tight-coarsensing}, the optimal strategy that respects the coarsening is an $\epsilon_1$-SSE, for $\epsilon_1=2AP\delta$.

In the second step, we solve the single-attacker single-defender game on each neighborhood. We can bound the utility of the attacker by $(1-\lambda)u_{1,0}^{\hat{v},a}\le u_\lambda^{\hat{v},a}\le u_{1,0}^{\hat{v},a}$. The lower bound is reached by setting the attacker strategy as $s_a\in\argmax_{s_a'}{u^{\hat{v}}(\bot,s_a')}$ a greedy strategy, regardless of the defense strategy. The upper bound is reached by setting the defense strategy to be $\bot$. Similarly, the utility of the defender can be bounded by $$ u_{1,0}^{\hat{v},d} \le u_\lambda^{\hat{v},d} \le (1-\lambda)\cdot u_{1,0}^{\hat{v},d}.$$
Indeed, the lower bound is when $\xx_d=\bot$, and the upper bound is reached if we assume that a defender in the neighborhood protects all targets completely (and $\lambda$ is small, so that a greedy strategy is approximately $\BR$).

As a consequence, the error of the utility estimation from step 2 can be bounded by $\epsilon_2=\lambda_{\max} P \max_{v\in V}{R^a(v)}$.

In the third step, the algorithm solves the multi-drone meta-game, using the approximated utility function from step 2 as an input. By Lemma~\ref{lemma:sse-continuous}, given the error for the second step, the third step outputs a $2\epsilon_2$-SSE.
The final solution is thus ensured to be an $(\epsilon_1+\epsilon_3)$-SSE. Nevertheless, to make sure the error is small, the error $\delta_2$ must be small as well.

Whenever $\lambda_{\hat{v}}>\lambda_c$ for some cutoff $\lambda_c$, in order to get a meaningful bound for the error, we must approximate the utility more accurately. As suggested in Section~\ref{subsec:single-drone-seq-ssg}, we should consider all strategies $s_a$ for the attacker in $\hat{v}$, such that $u^{\hat{v}}(\bot,s_a) \ge (1-\lambda_{\hat{v}}) u_{1,0}^{\hat{v},a}$, which may consist of more than all greedy strategies. 
In particular, let $r_1 \ge r_2 \ge \ldots \ge r_P > \delta$ be the top $P$ rewarding nodes. Then the attacker strategy space consists of all paths that pass through enough of these nodes so that the overall utility when there is no defender is more than $(1-\lambda_{\hat{v}}) u_{1,0}^{\hat{v},a}$. Doing so will result with an exact solution, and will allow us to replace $\lambda_{\max}$ with $\lambda_c$, as desired.

\section{Experiments}
\label{sec:experiments}

Our experiments were aimed at assessing the efficacy of S2D2 by comparing runtime and defender utilities with a baseline.
We first synthetically generated utilities for nodes in 80 world cities from all continents (except Africa and Antarctica), including several capitals. 
Then, we used manual annotations for different facilities through a survey of 7 security and defense experts. 
Finally, we did a detailed case study that qualitatively assessed the defenses recommended for a single city.

All the experiments were run on an Intel(R) Core(TM) i9-10980XE CPU with 256 GB RAM.

Implementation of S2D2, baseline defense strategy, and the code for the experiments presented below, are all publicly available on Github~\url{https://github.com/tonmoay/S2D2-Experiments}.

\subsection{Setting} 
\paragraph{Dataset and Parameters} We created a dataset of 80 cities, ranging from a few thousands nodes up to a few hundreds of thousands of nodes. For each city, the street networks were sourced from the \emph{OpenStreetMap} platform via the OSMnx library~\cite{boeing2017osmnx}.  

The number of neighborhoods was fixed to $|\hat{V}|=8 D$, that is, proportional to the number of defense drones. Taking a large constant will result in a graph that is mostly unprotected, and neighborhoods that are too small. Taking a small constant would mean that there are enough defense drones to cover all neighborhoods with probability 1, yet those neighborhoods will be too big to protect. 

As the dataset lacked rewards/penalties for nodes, we assigned those parameters in two ways: (i) sampling them independently from a distribution (log-normal/Zipf); (ii) using security experts to manually annotate 6 cities. Defender penalties were then assigned by randomly perturbing the rewards. This maintains some degree of correlation while circumventing a zero-sum game scenario. 

For the synthetic data, we assigned rewards to city nodes by sampling independently from a distribution over the $[0,\infty)$ interval. The reason is that (i) rewards should be non-negative and (ii) we expect the set of nodes with high rewards to be sparse. Otherwise, the game essentially becomes an evasion game where the goal is to catch the attacker as soon as possible. Specifically, we sampled from the log-normal distribution with $\mu=0,\sigma=4$, as well as a Zipf distribution with $s=2$.

For the manually annotated data, we asked 7 senior defense and security officials from the US, EU, Asia, and the Middle East to rate the importance of different facilities in city neighborhoods. The 6 cities included three major U.S. cities, an Asian megacity with a population of over 20M people, and two smaller cities in the Middle East. In all, the cities included 3 world capitals. We asked the experts to imagine a city that they knew well when filling out the survey without telling them which city to look at. We asked questions related to the following types of facilities: Local/Municipality Buildings (e.g., the office of the mayor or city administration), National Government Leadership Buildings (e.g., the White House in Washington DC or 10 Downing Street in London), National Government Operational Buildings (e.g., the office of a Ministry), Security Installation Buildings (e.g., Ministry of Defense or Europol Headquarters), Hospitals, Electricity/Natural Gas Plants, Sanitation and Water Plants, Industrial and Hazardous Materials Areas, Transportation Hubs (e.g., airports, train stations, etc.), Tourist Sites, Financial Districts, Shopping and Entertainment Areas, Sports Arenas, and High Density Areas. Each type of facility was to be ranked on a 1 to 5 scale with 1 meaning it was of very low importance and 5 meaning it was of critical importance. The median values obtained are summarized in Table~\ref{tab:experts-utilities}.
\begin{table}[h!t]
\begin{center}
\begin{small}
\begin{tabular}{cc}
\toprule
\bf \multirow{2}{*}{Facility} & \bf Median \\
\bf  & \bf Utility \\
\midrule
National Government Leadership Buildings & 5 \\
Security Installation Buildings  & 5 \\
Electricity/Natural Gas Plants  & 5 \\
National Government Operational Buildings  & 4 \\
Hospitals  & 4 \\
Sanitation and Water Plants  & 4 \\
Industrial and Hazardous Materials Areas  & 4 \\
Transportation Hubs  & 4 \\
High Density Areas  & 4 \\
Local/Municipality Buildings & 3 \\
Tourist Sites  & 3 \\
Financial Districts  & 3 \\
Shopping and Entertainment Areas  & 3 \\
Sports Arenas  & 3 \\
\bottomrule
\end{tabular}
\end{small}
\caption{Median utilities assigned to various types of facilities by security experts.}\label{tab:experts-utilities}
\end{center}
\end{table}
All the experts agreed that security installations and major national government buildings would have top priority followed by utilities (e.g., power, water). The vast majority of the cities (e.g., residential areas) would have much lower rewards.
A specifically designed annotation interface (Figure ~\ref{fig:chicago_blur}) was used by the experts to draw rectangles and/or polygons and provide a utility value for each polygon, i.e., value of the region. \revision{To avoid risks to real cities, the figure has been intentionally blurred to show the overall use of the interface without making it possible to identify the specific city.}
\begin{figure}[h!t]
\centering
\includegraphics[width=0.8\textwidth]{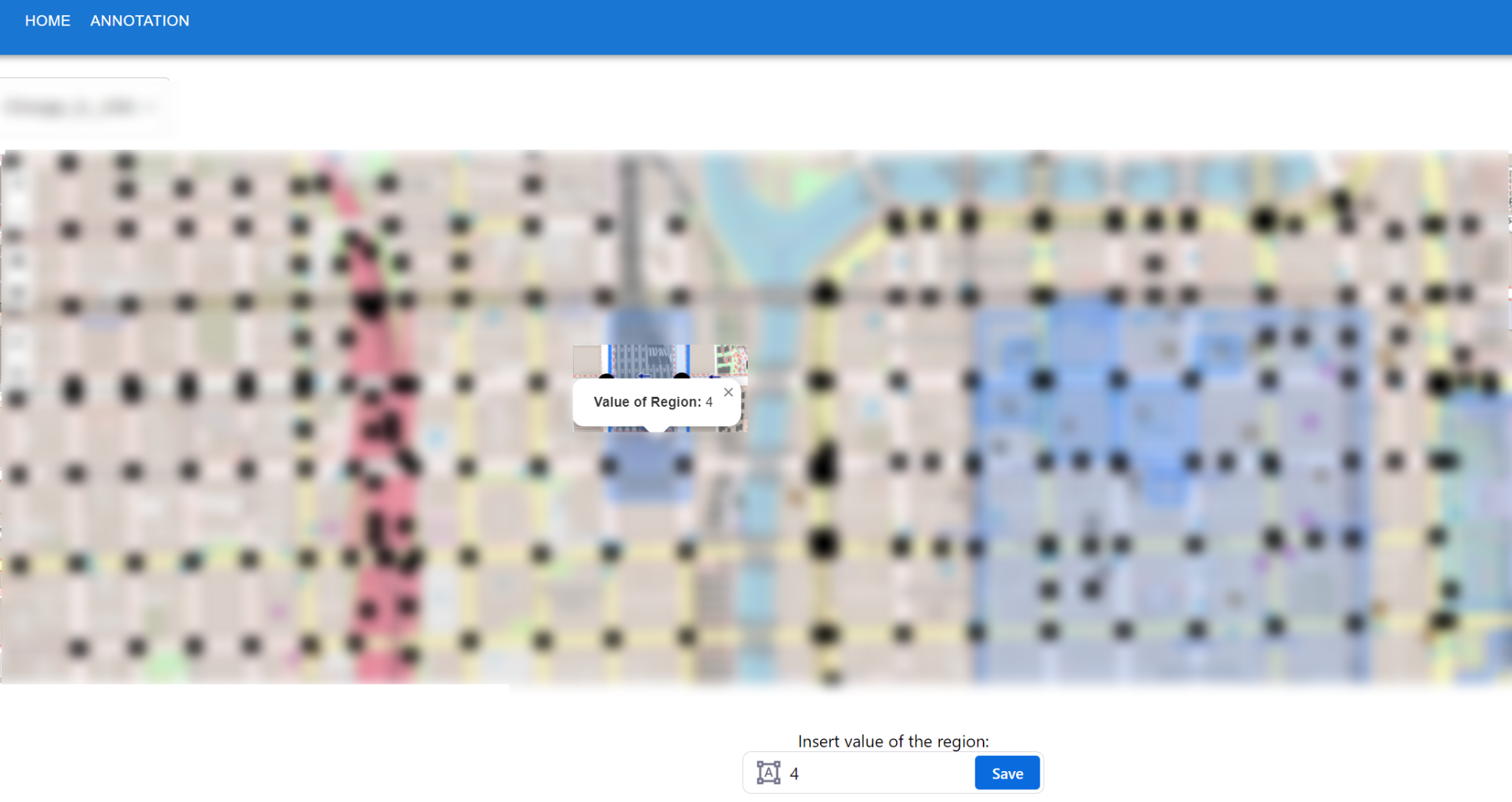}
\caption{Annotation interface \revision{(intentionally blurred)}.}
\label{fig:chicago_blur}
\end{figure}

We fixed the number of defender and attacker drones to be $D=A=4$, the battery capacity $B=6$ and the payload $P=4$. However, in some experiments, we also varied $A,D,B$, and $P$. 
The outcomes presented pertain to 100 iterations in all reported results.

\paragraph{Baseline} 
We compared S2D2's runtime and expected utility with a greedy baseline (Algorithm~\ref{alg:greedy-baseline}) that employs a drone-swarm defense mechanism.
\begin{algorithm}[h!t]
\caption{\textsf{Greedy Baseline}}\label{alg:greedy-baseline}
\begin{small}
\begin{algorithmic}[1]
\REQUIRE 
Undirected graph $G = (V, E)$; \\
\quad\ \  coarsening $\hat{V}$ of $G$, and $\delta$; \\
\quad\ \  numbers of attacker and defender drones $A,D\in \NN$; \\
\quad\ \  attacker drone's  payload $P \in \NN$; \\
\quad\ \ drone's battery capacity $B \in \NN$; \\
\quad\ \ attacker rewards $R^a \in \NN^{|V|}$; \\
\quad\ \ defender penalties $P^d \in \ZZ_{<0}^{|V|}$. \\
\ENSURE 
Defense mixed strategy $(\hat{p}_d,\hat{\xx}_d) \in \Delta(\calS_{\hat{V}}^d)$. \\
\STATE weights\_d $\leftarrow$ []; 
\FOR{{\bf each} $\hat{v}\in\hat{V}$}
    \STATE $\hat{v}_a\leftarrow\hat{v}$.\textsf{get\_top\_attack\_rewards}($P$,$R^a$,$P^d$); \COMMENT{Breaks ties in favor of defender}
    \STATE weights\_d[$\hat{v}$] $\leftarrow\hat{v}_A$.\textsf{sum\_abs\_penalties}($P^d$);
\ENDFOR
\STATE $\hat{p}_d\leftarrow$ D * weights\_d / weights\_d.sum(); \label{line:greedy-def-dist}
\FOR{{\bf each} $\hat{v} \in \hat{f}_a$}
    \STATE $s_1^d\leftarrow$ \textsf{goto\_random\_node}($\hat{v}$);
    \FOR{{\bf each} $2 \le t \le B$}
        \STATE $s_t^d\leftarrow$ \textsf{move\_towards\_closest\_node}$(\hat{v},\delta)$;
    \ENDFOR
    \STATE $\hat{x}_d[\hat{v}] \leftarrow (s_1^d,\ldots,s_t^d)$; \COMMENT{a pure strategy}
\ENDFOR
\RETURN $(\hat{p}_d\hat{x}^d)$;
\end{algorithmic}
\end{small}
\end{algorithm}
The baseline allocates protection to each neighborhood with a drone, doing so proportionally to the cumulative absolute penalties of the top $P$ attacker rewarding nodes. This is done by letting each defense drone sample its starting neighborhood independently from the distribution $\hat{p}_d$ (Line~\ref{line:greedy-def-dist}).

Within each neighborhood, the baseline drone starts at a random node. The output of $\textsf{move\_towards\_closest\_node}(\hat{v})$ is a function that assigns a random start node from $\hat{v}$ in each execution. It then follows a greedy next-step function, defined for each time step $1\le t \le B$ as follows. The defender looks for the target $v \in \hat{V}$ that is closest to the attacker drone position, is of interest to the attacker (amongst the top $P$), is valuable to the defender (penalty $<-\delta$), and the defender can reach there before the attacker, and moves one step along the shortest respective path. If there are multiple attackers in sight (in $\hat{v}$), it picks the closest one to hunt down, breaking ties randomly. Importantly, the algorithm returns a strategy, so the output of $\textsf{move\_towards\_closest\_node}$ is actually a function, that determines the next step for the defender given the defense current position, attacker last observed position, updated penalties and rewards, and the neighborhood graph structure.
This procedure is reiterated based on the attacker's updated position. The above baseline is an adaptation of~\cite{chen2016multiplayer} to our setting. As in~\cite{chen2016multiplayer}, each defender is paired with the closest observed attacker drone. In addition, chasing a drone is done by predicting its projectile. While~\cite{chen2016multiplayer} assumes a straight-line projectile, our heuristic takes into account the attacker rewards, and also rather than straight-line the defender takes the shortest path along the graph.

Notably, the defense strategy above can be implemented by using a defense drone swarm. As the experiments will demonstrate, such a defense strategy is more scalable in terms of run-time, but the expected defender utility is going to be smaller.

\paragraph{Attacker} 
Both S2D2 and baseline defenders were paired with the S2D2 attacker, as we are interested in handling a strong attacker approximating the best response.

\subsection{Results} 

Figure~\ref{fig:main_figure} reports the results regarding runtime and defender utility (expressed as ratios between the baseline and S2D2).
\begin{figure}[h!t]  
    \centering
    \begin{tabular}{cc}
         \includegraphics[width=0.48\textwidth]{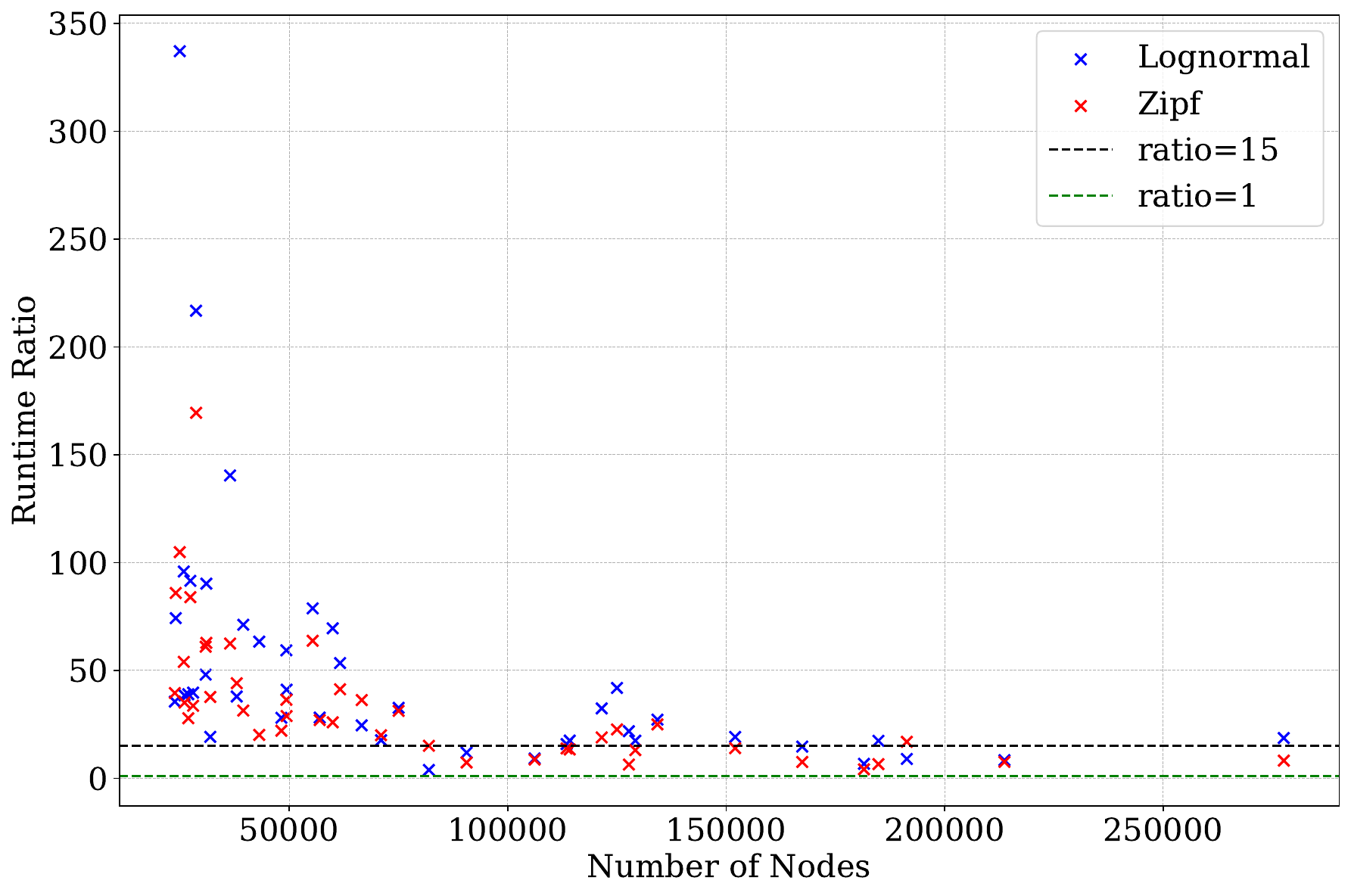} &  
         \includegraphics[width=0.48\textwidth]{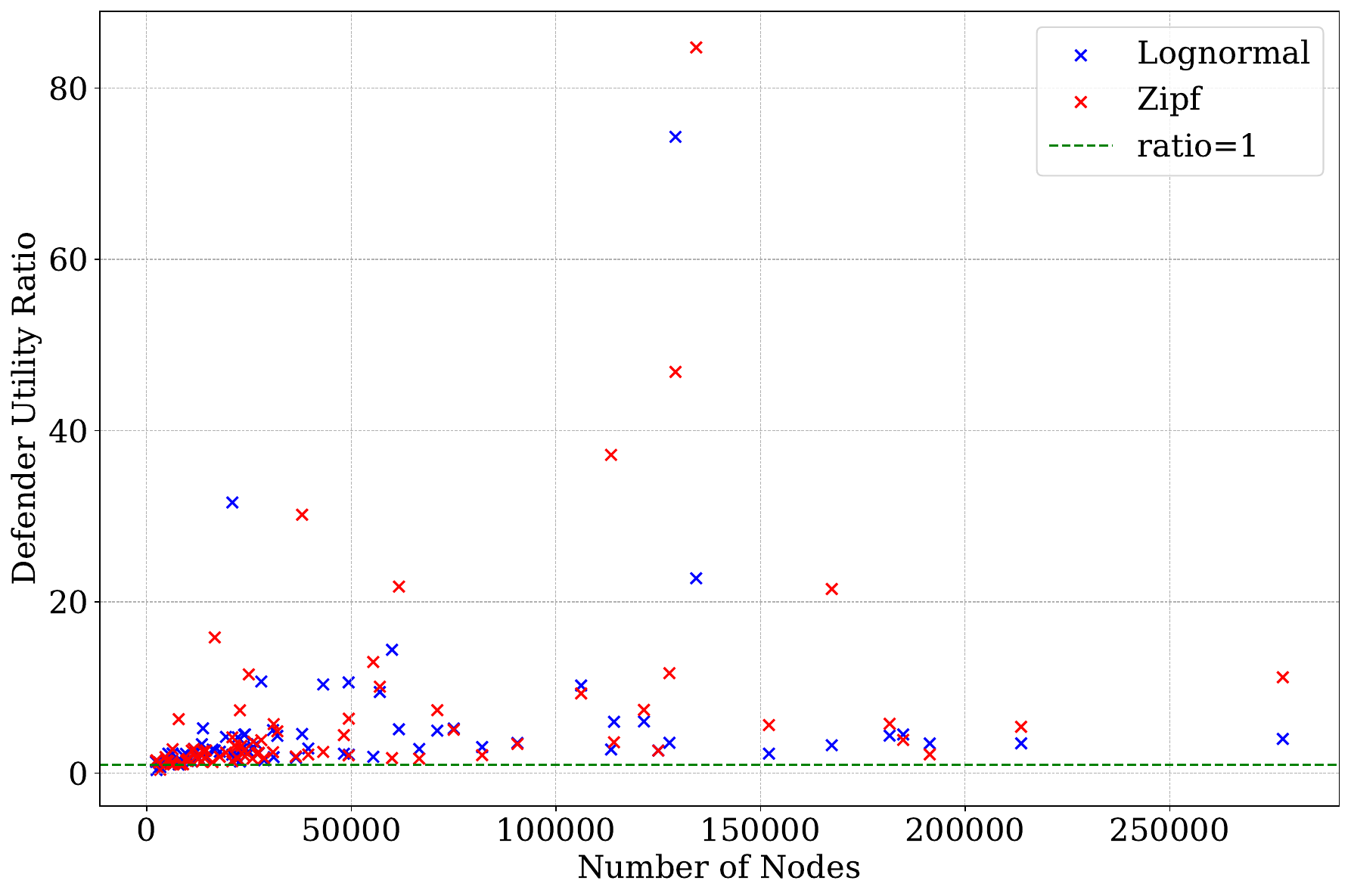} \\
         (a) Runtime ratio &(b) Defender utility ratio\\
         &(synthesized utilities)\\
         & \\
    \multicolumn{2}{c}{\includegraphics[width=0.48\textwidth]{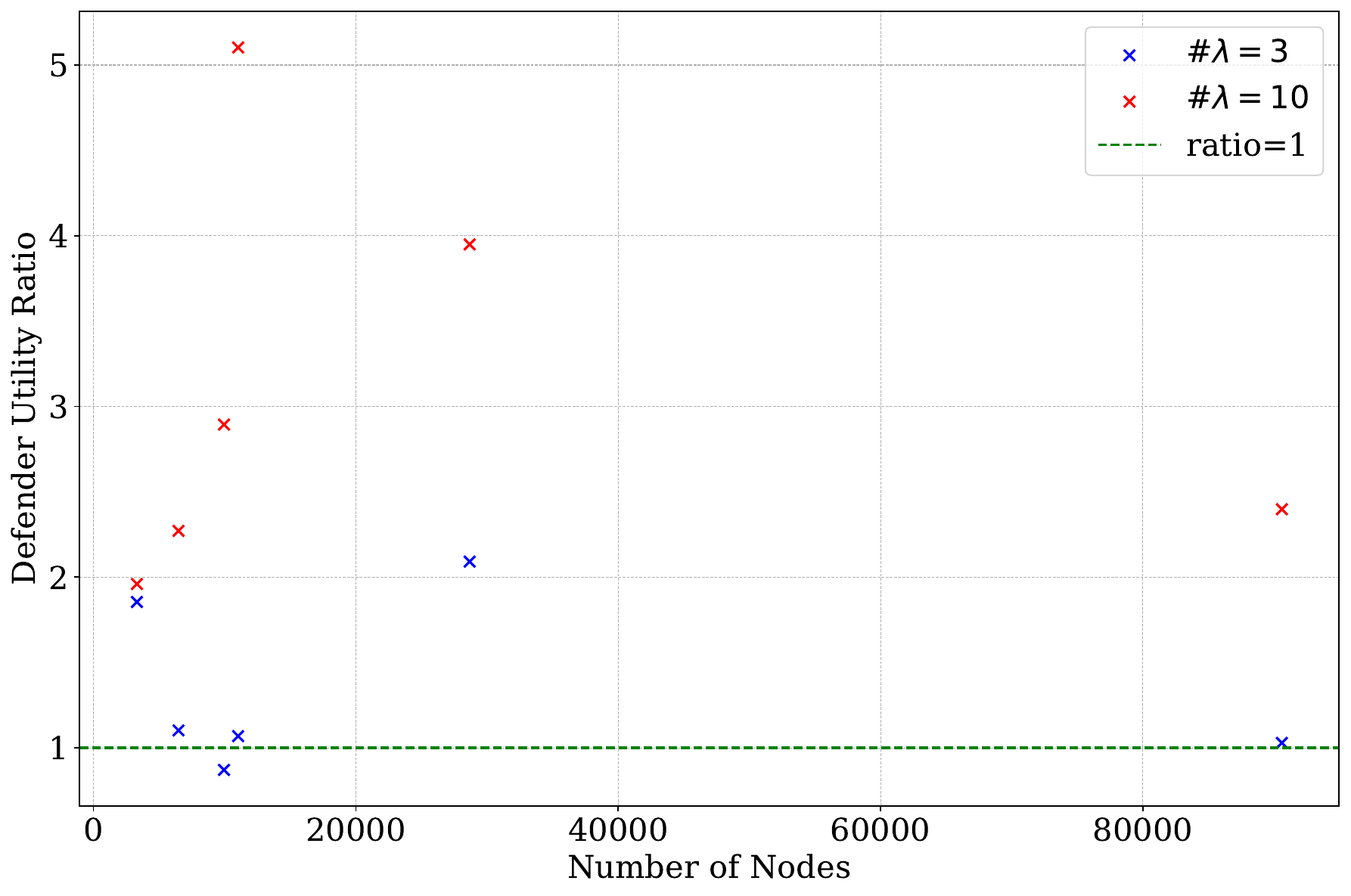}} \\
    \multicolumn{2}{c}{(c) Defender utility ratio (manual annotation)} \\
    \end{tabular}
    \caption{Comparison of S2D2 with the baseline.}
    \label{fig:main_figure}
\end{figure}

Figure~\ref{fig:main_figure}(a) shows that the runtime ratio between S2D2 and the baseline approaches a constant factor of about 15. Indeed, for small cities S2D2 runtime is dominated by the multi-drone stage, but this cost becomes negligible as neighborhood size grows. Accordingly, S2D2 can be run within times that are reasonable even for the largest cities (3.5 hours with one CPU for the largest one).
The more expensive runtime is amply rewarded by the fact that S2D2 decreases the expected defender loss and the attacker's expected utility by an average of 2.84 times that of the baseline, as shown in Figure~\ref{fig:main_figure}(b). 
Interestingly, the utility is less dependent on graph size and more on the structure of the graph and distribution of rewards.

As for manually annotated cities, Figure~\ref{fig:main_figure}(c) shows that the baseline is better for one city. This could be due to the piece-wise linear approximation of single-drone neighborhood utility function, where we only used $\#\lambda=3$ points $\lambda_i \in \{0,0.5,1\}$ (which performed good enough for the synthetic data). We therefore repeated the experiment with S2D2 using $\#\lambda=10$ points. Indeed, S2D2 convincingly outperforms the baseline in all manually annotated cities.

Figure~\ref{fig:suppl_figure} reports results regarding the dependency of runtime and defender utility on the ratio between the number of attacker to defender drones (synthetic utilities, log-normal distribution). The number of attacker drones is $A= \lfloor D \times \text{ADR} \rfloor $ where ADR is the \emph{attacker-defender ratio}.
\begin{figure}[h!t] 
    \centering
    \begin{tabular}{cc}
         \includegraphics[width=0.48\textwidth]{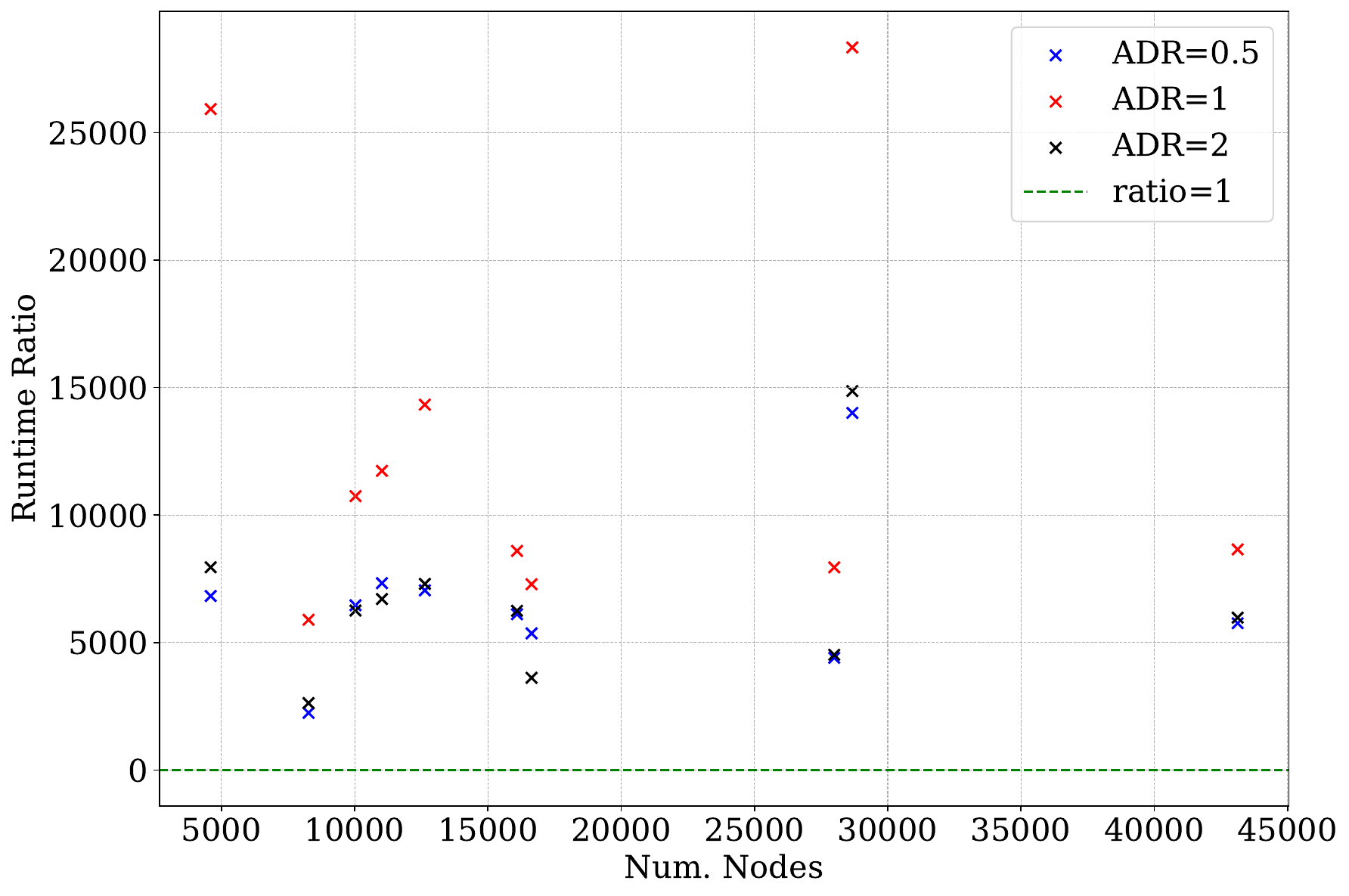} &  
         \includegraphics[width=0.48\textwidth]{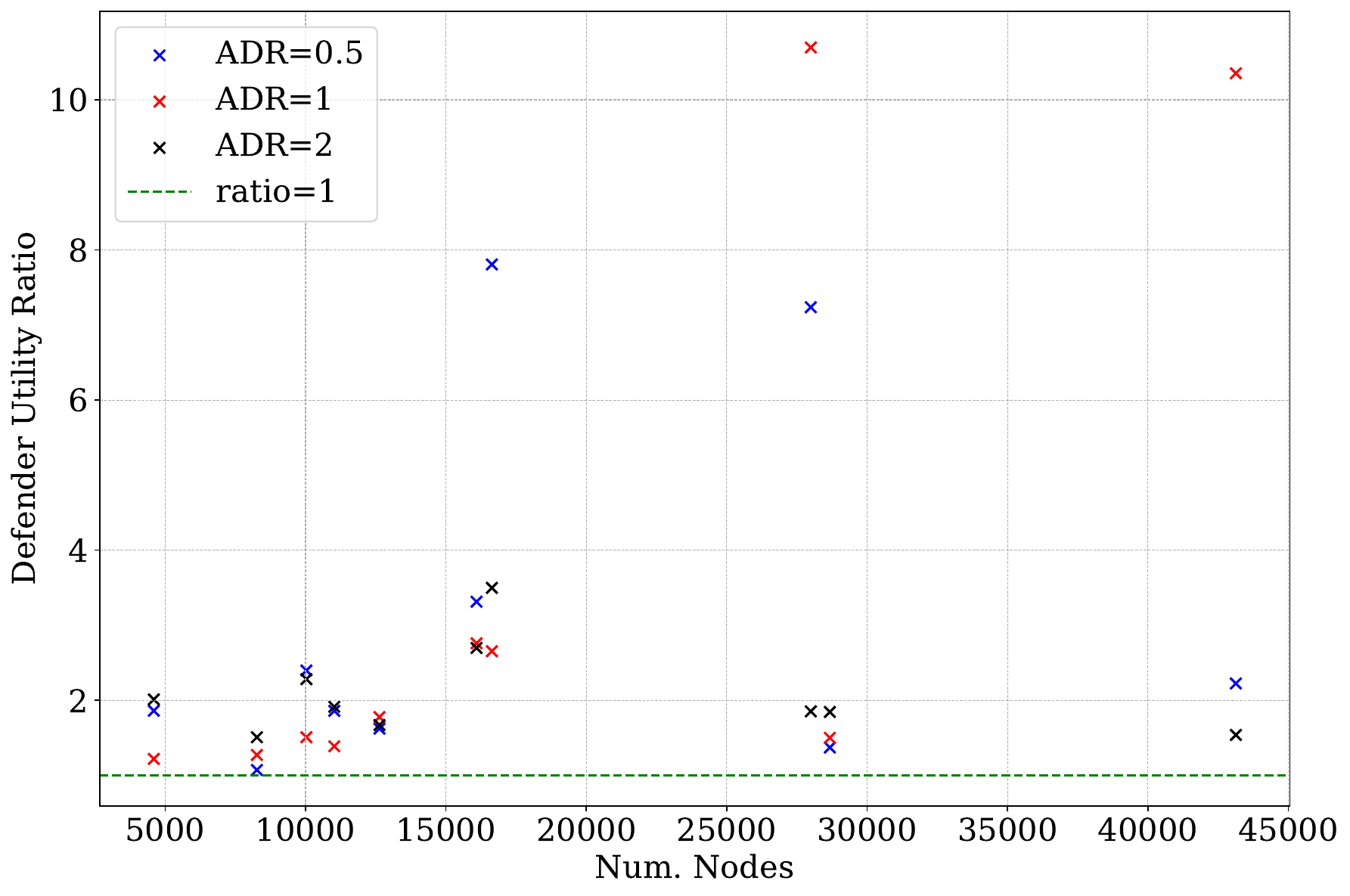} \\
         (a) Runtime ratio & (b) Defender utility ratio\\
    \end{tabular}
    \caption{Comparison of S2D2 with the baseline, with ADR $\in\{0.5,1,2\}$.}
    \label{fig:suppl_figure}
\end{figure}
S2D2 yields a higher defender utility at the expense of increased run-time compared to the baseline. The runtime ratio is relatively big since the cities are relatively small (up to 45,000 nodes), and so the multi-drone phase of the solution takes a significant portion of time. However, we do not observe a strong correlation between ADR and runtime, so handling more attackers does not incur a higher computational cost. We also do not observe a correlation between ADR and defender utility ratio. This is probably because both the baseline and the S2D2 defender utilities are similarly affected from the change of ADR.

Figure~\ref{fig:manually_annotated_def_utility_ratio} shows 
how defender utility varied when we perturbed manually annotated utilities. The latter were perturbed by $\pm$10\%, i.e.\ they were fixed to $\pm$10\% of the true value.
\begin{figure}[h!t]  
    \centering
    \begin{tabular}{cc}
         \includegraphics[width=0.48\textwidth]{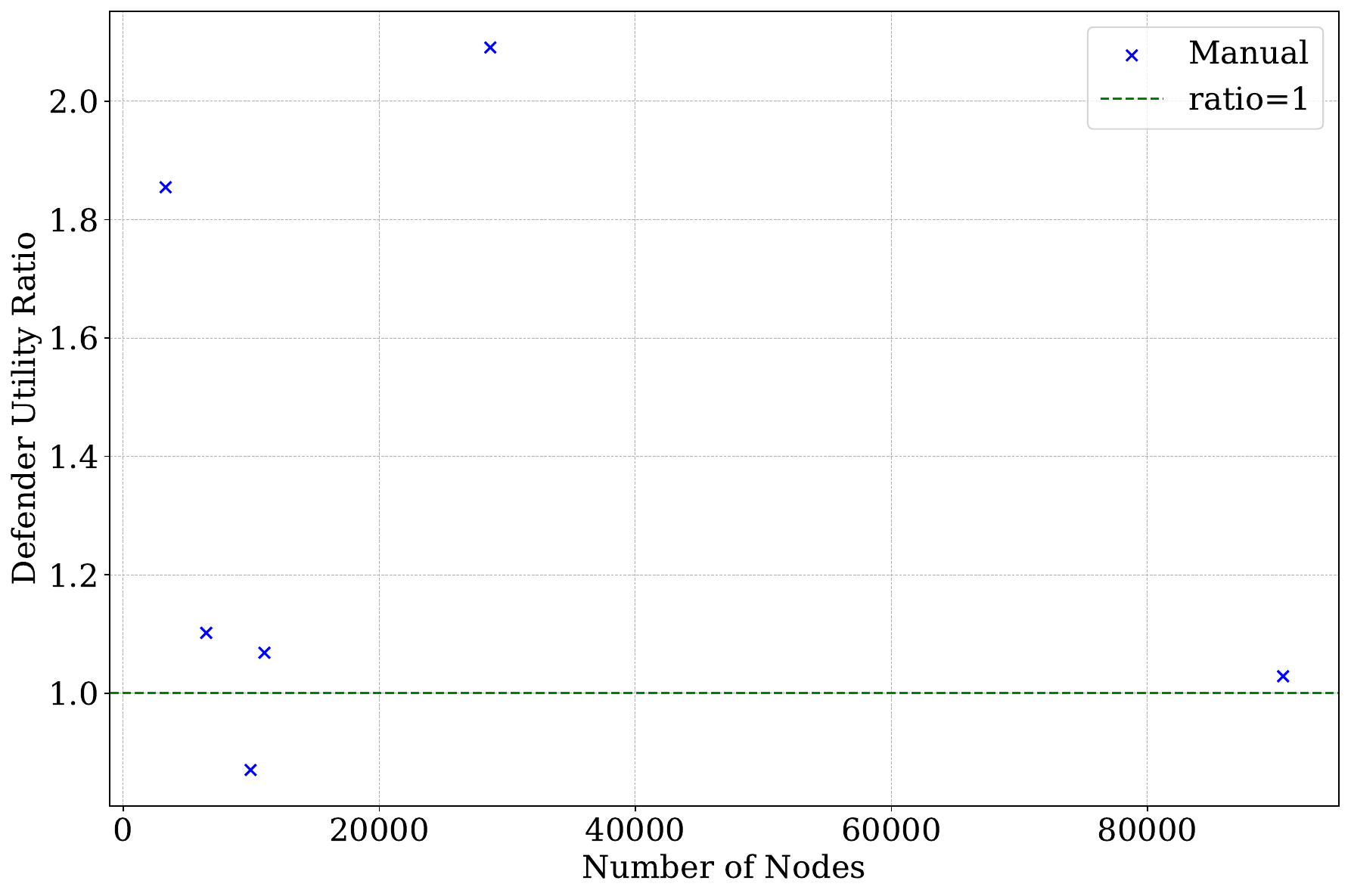} &  
         \includegraphics[width=0.48\linewidth]{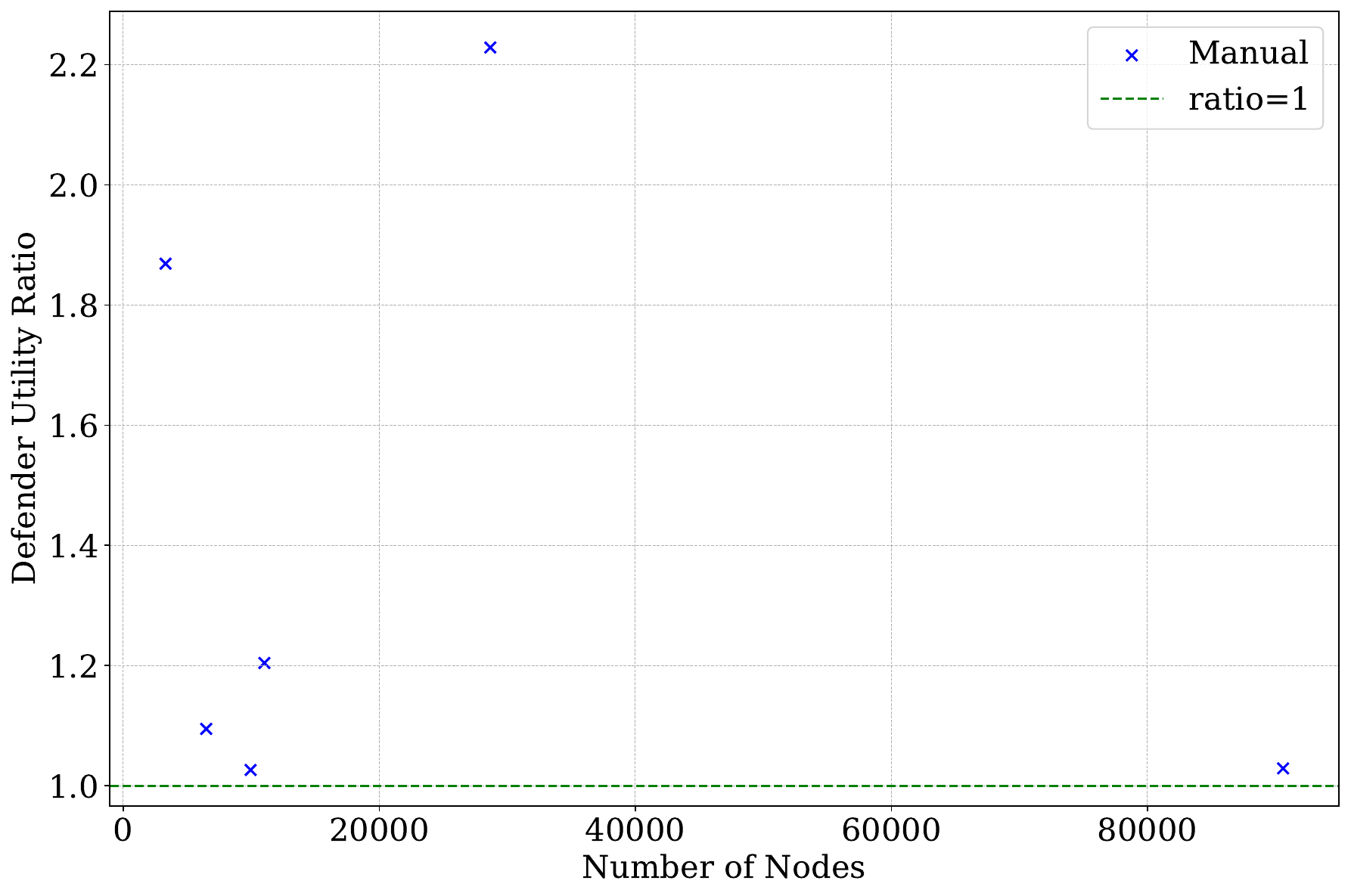} \\
         (a) Defender utility ratio & (b) Defender utility ratio\\
         (No perturbation) & ($10\%$ perturbation)
    \end{tabular}
    \caption{Defender utility ratios when the defender's estimate of the attacker's utility is off by 0 (a) or by $\pm 10$\% (b).}
    \label{fig:manually_annotated_def_utility_ratio}
\end{figure}
In the case of one city, the defender utility ratio is less than 1, so the baseline outperforms S2D2. This is explained by the coarse approximation of the single-drone utility function as a piece-wise linear function, which can be improved if necessary at the expense of runtime. In all the other 5 cases, the utility ratio ranges from around 1.05 to 2.15. In two cases, S2D2 yields almost double the utility of the baseline, and in 3 other cases, it outperforms the baseline by a smaller margin. 
It should be noted that, if we average the defender's utility ratio across the 6 cities, the average is 1.34. Thus, on average, S2D2 provides 34\% improvement over the baseline algorithm in terms of the utility to the defender. Even though that comes at the cost of an increased runtime, most cities would be happy to make this tradeoff: saving 34\% more of the utility of the city (lives and property damage).

Figure~\ref{fig:manually_annotated_def_utility_ratio}(b) shows that in every single case, the defender utility ratio is over 1, so the S2D2 algorithm outperforms the baseline. 
Even when the defender's assumption about the attacker's utility is slightly incorrect, the average improvement over the baseline is 41\%. Most cities would prefer to save an additional 41\% of the utility of the city compared to saving some runtime.

Finally, Figure~\ref{fig:fixed-att} reports the results we obtained when we fixed $A=5$ attacker drones and selected 5 cities based on a balanced size distribution, from 2,283 nodes to 125,013 nodes (synthetic utilities, Zipf distribution). Battery capacity and payloads were fixed to $B=6$ and $P=4$. 

Figure~\ref{fig:fixed-att}(a) reports the impact on defender utility of varying the number of defender drones from 2 to 10 against 5 attacker drones.

\begin{figure}[h!t]  
    \centering
    \begin{tabular}{cc}
         \includegraphics[width=0.48\textwidth]{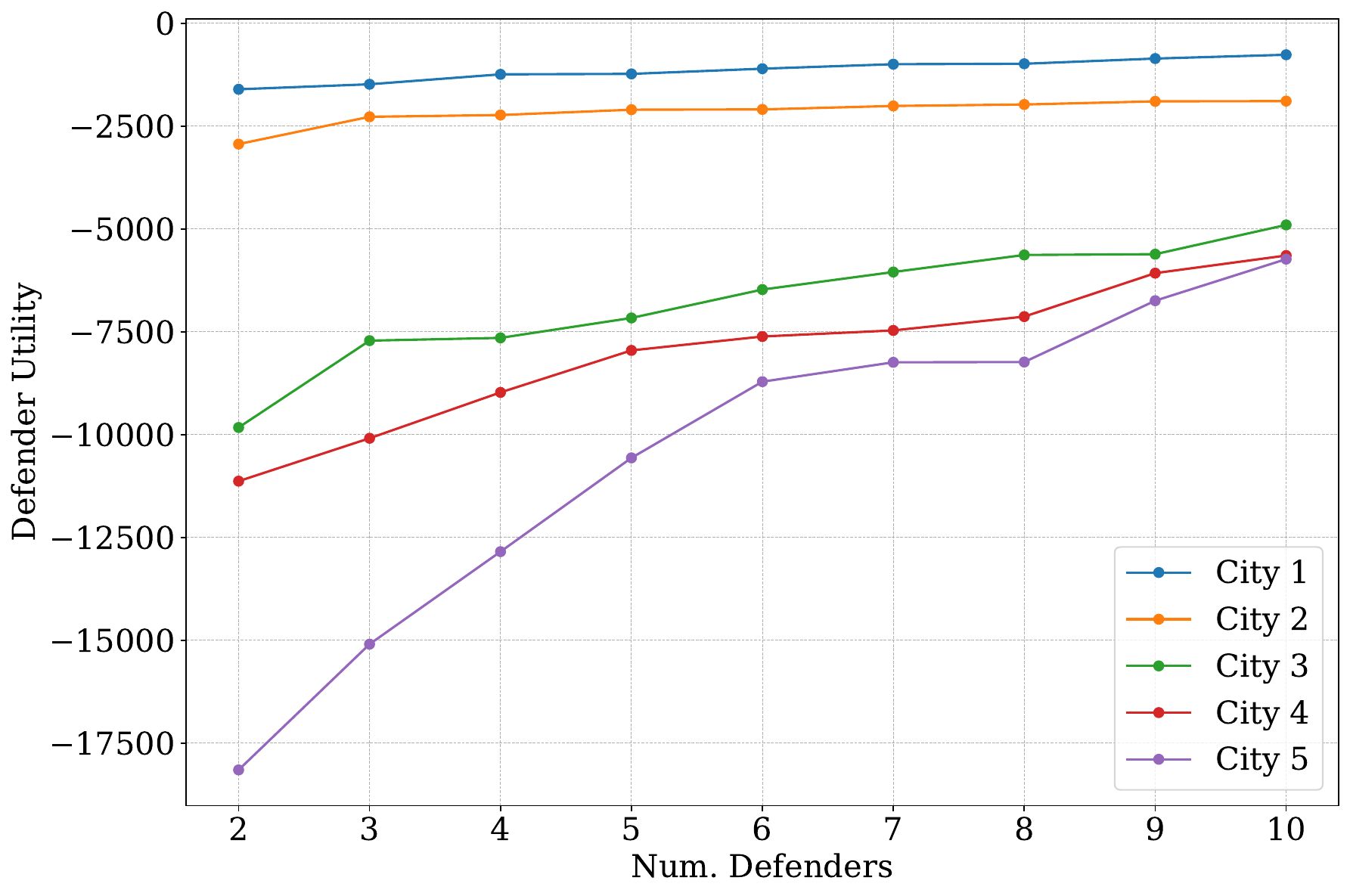} &  
         \includegraphics[width=0.48\textwidth]{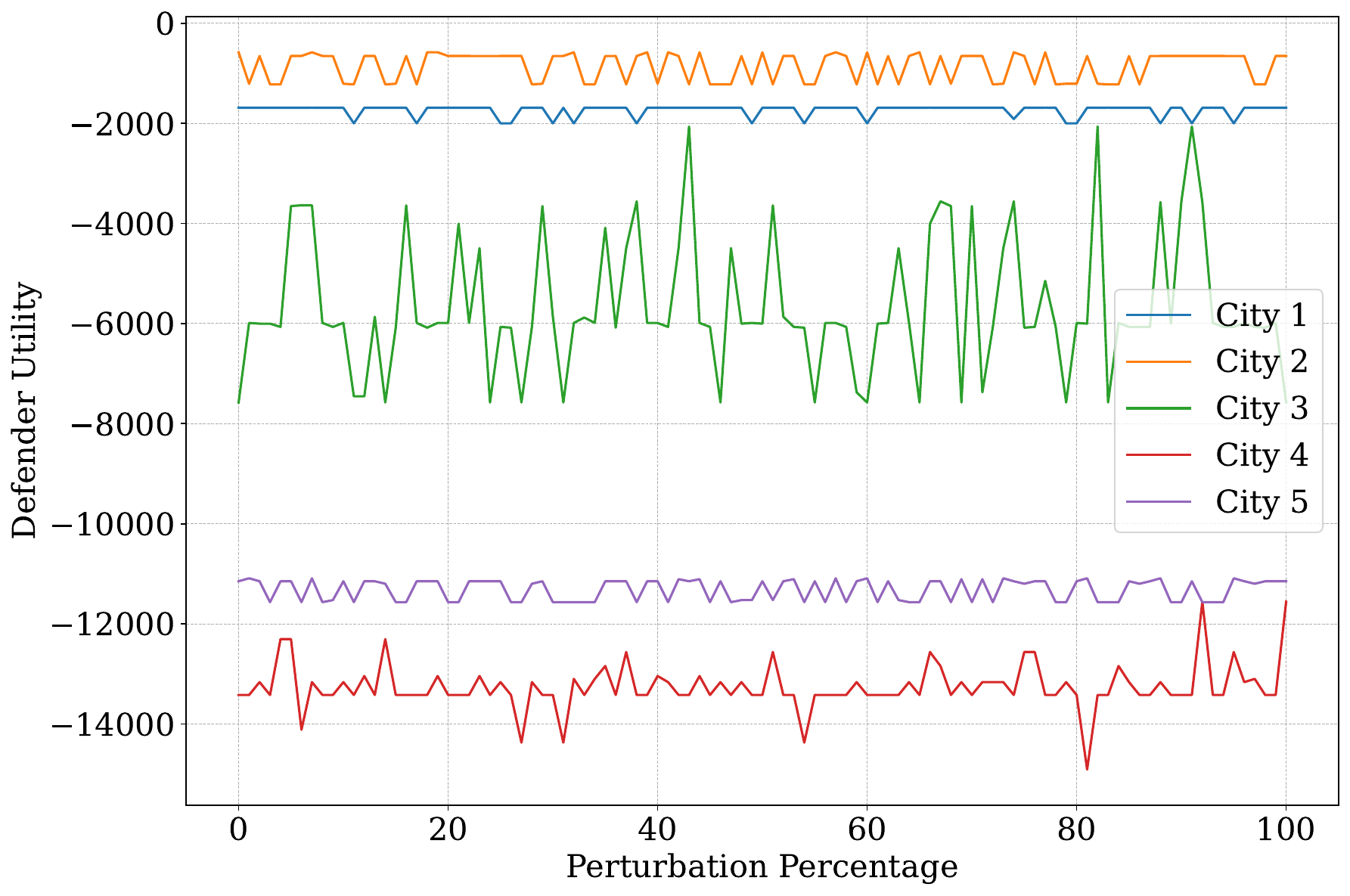}
         \\
         (a) Defender utility & (b) Defender utility \\
         vs. number of defender drones & vs. perturbation percentage\\
    \end{tabular}
    \caption{Defender utilities with 5 attacker drones.}
    \label{fig:fixed-att}
\end{figure}
As expected, increasing the number of defenders consistently enhances defender utility.
Figure~\ref{fig:fixed-att}(b) shows the impact of perturbing utilities from 0\% to 100\% in 1\% increments, with penalties unchanged, when we fixed $D=5$ defender drones.
We sampled a defender strategy from the mixed strategy set, perturbed the rewards, and then evaluated the attacker's response and the utilities of both sides across all perturbation levels.
The results show that, despite the perturbations, the defender strategy remains effective, demonstrating robustness against the attacker's adaptations. 
Although some noise was observed, the overall utility trends were stable.

\paragraph{Case Study of One Major City}
We now describe a detailed case study of one major city from the Americas with a population of over 2M. The city had 28,671 nodes and was manually annotated. Unless stated otherwise, the experimental parameters described previously were used in this case study.

Figure~\ref{fig:city_casestudy_figure} shows the runtimes and utility ratios we obtained when varying the number of defender drones, the payload, and the battery capacity. 
\begin{figure}[h!t]  
    \centering
    \begin{tabular}{cc}
         \includegraphics[width=0.48\textwidth]{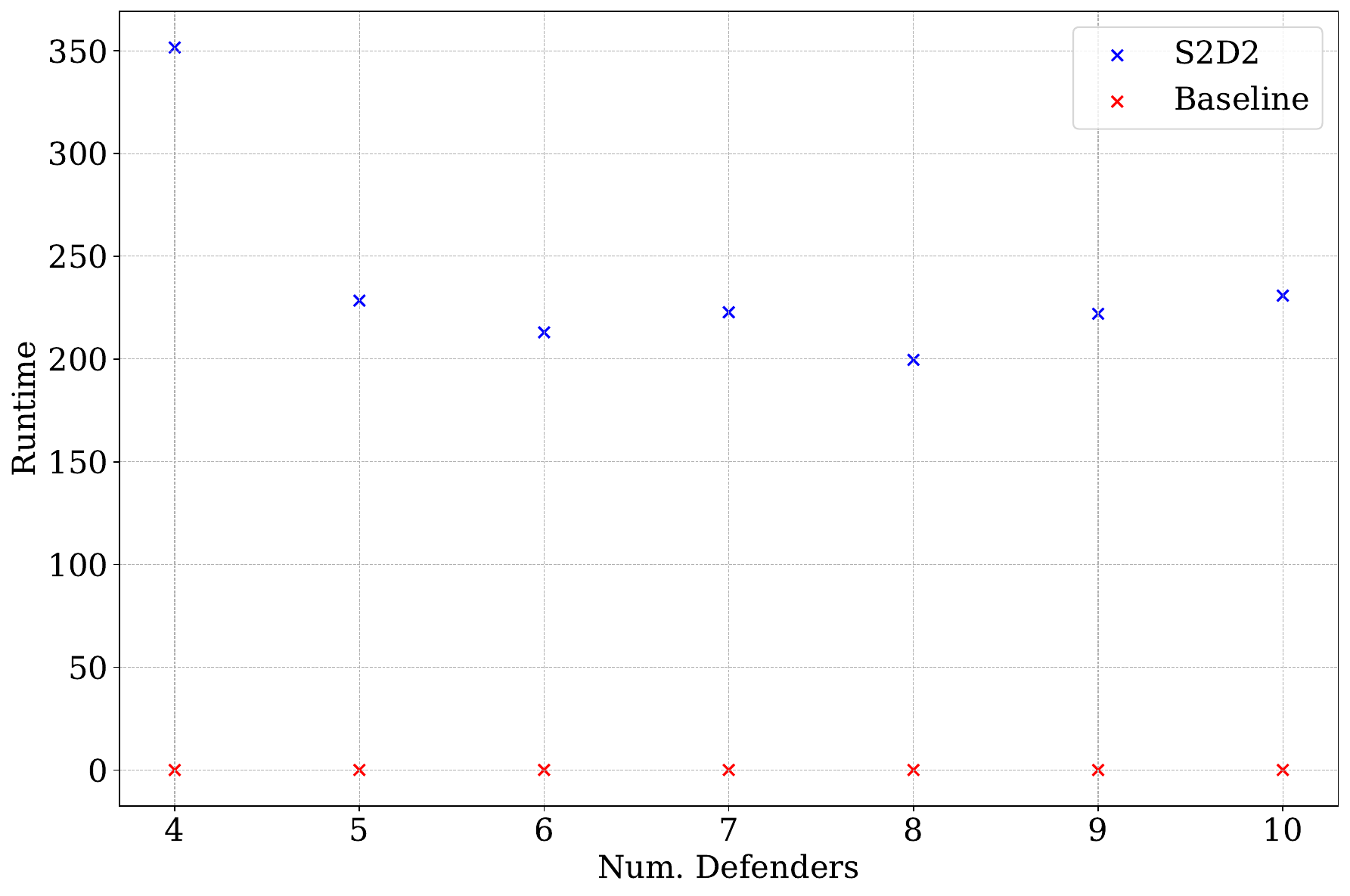} &  
         \includegraphics[width=0.48\textwidth]{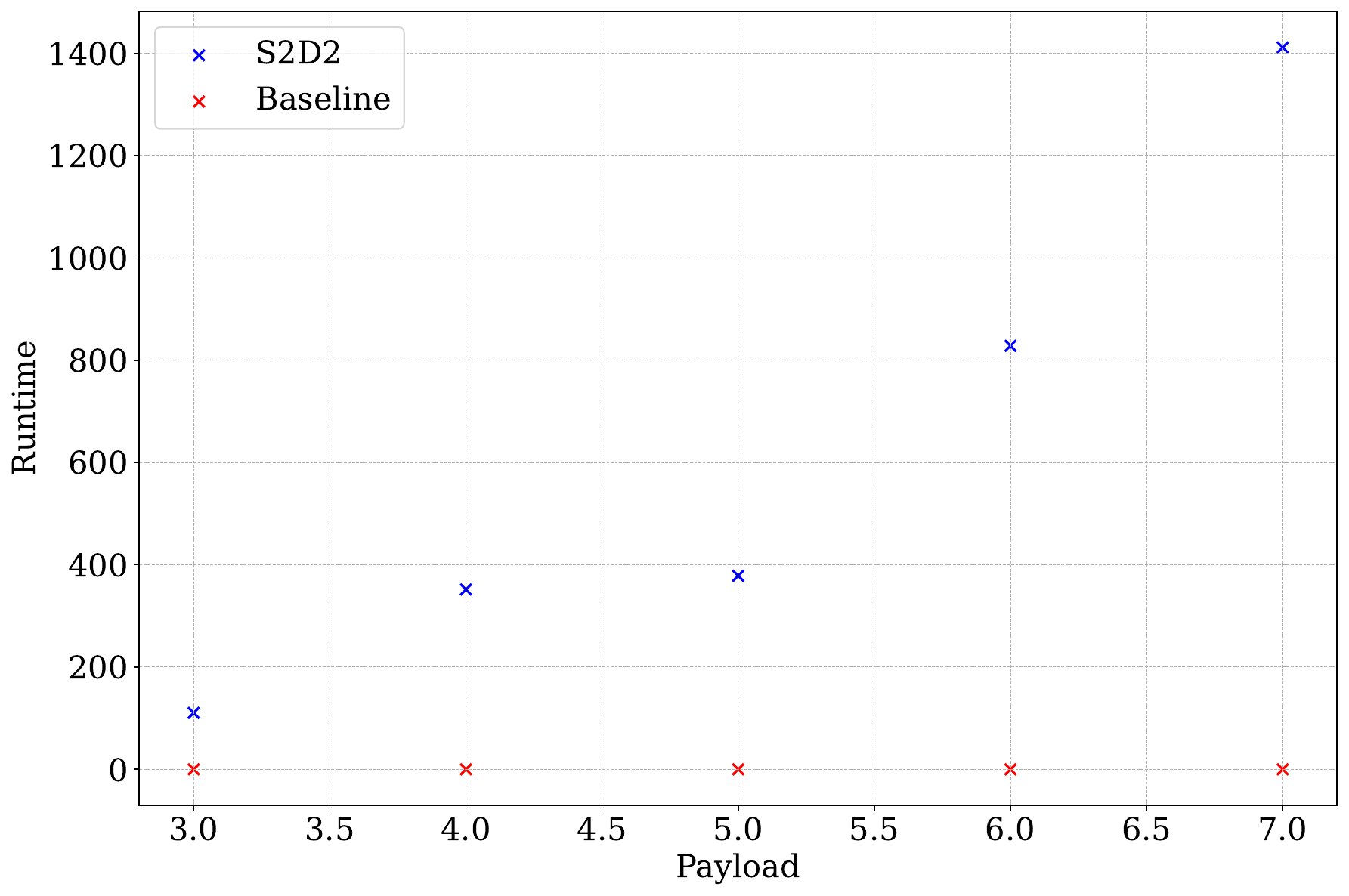} \\
         Runtime (s) & 
         Runtime (s) \\
         vs. number of defender drones & vs. payload \\
         & \\
         \includegraphics[width=0.48\textwidth]{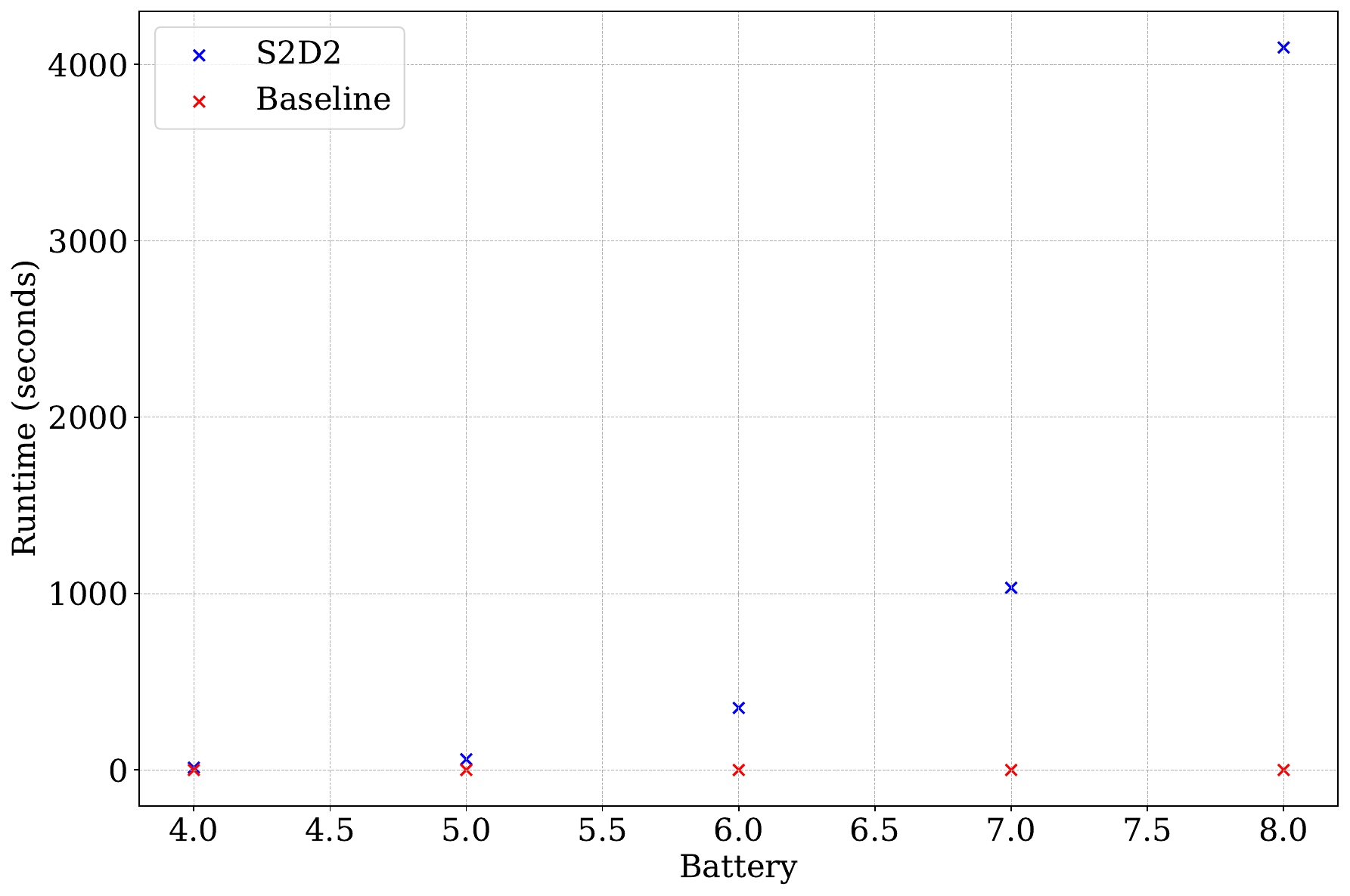} &
         \includegraphics[width=0.48\textwidth]{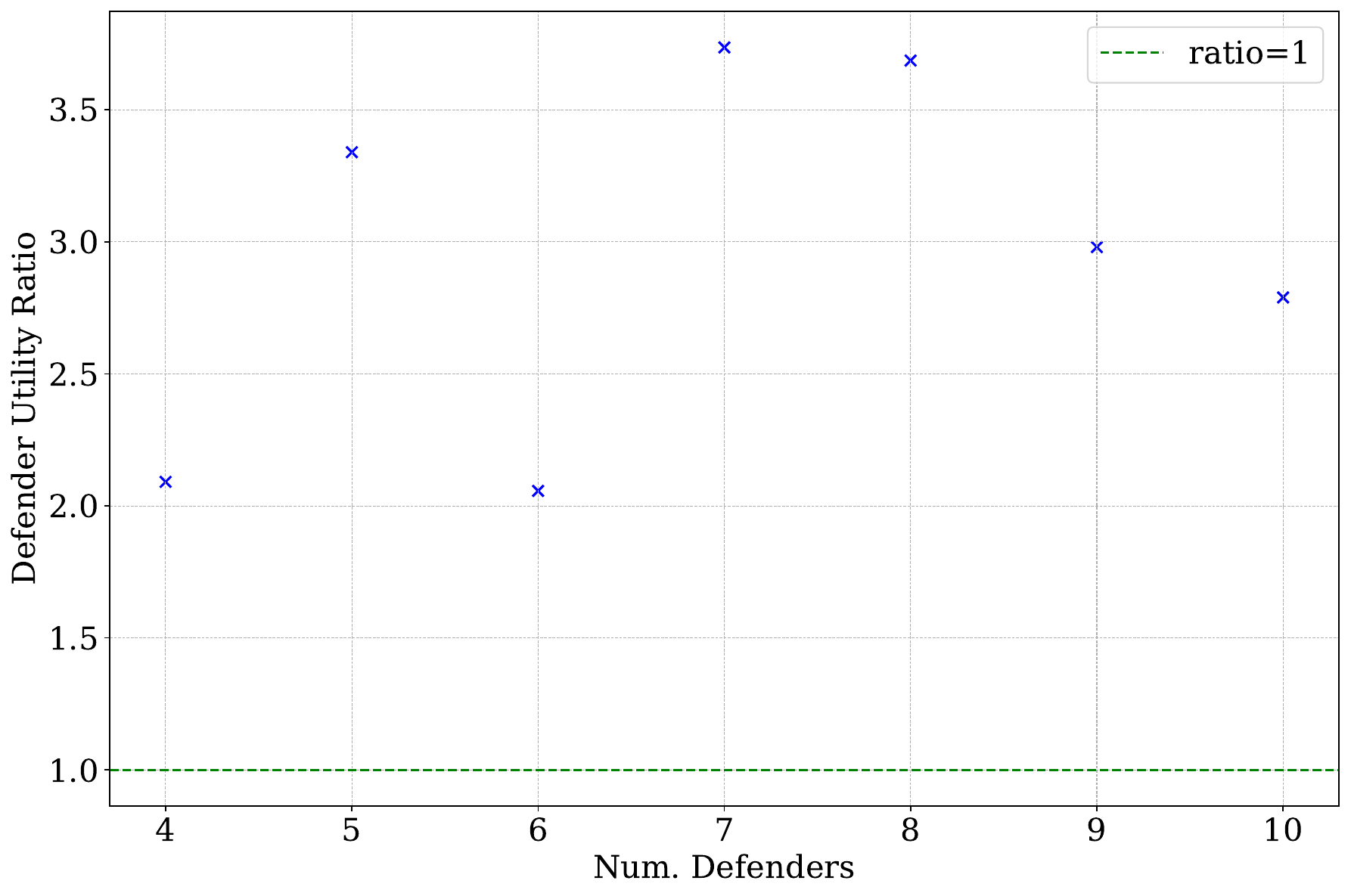} \\
         Runtime (s) &
         Defender utility ratio \\
         vs. battery capacity & vs. number of defender drones \\
         & \\
         \includegraphics[width=0.48\textwidth]{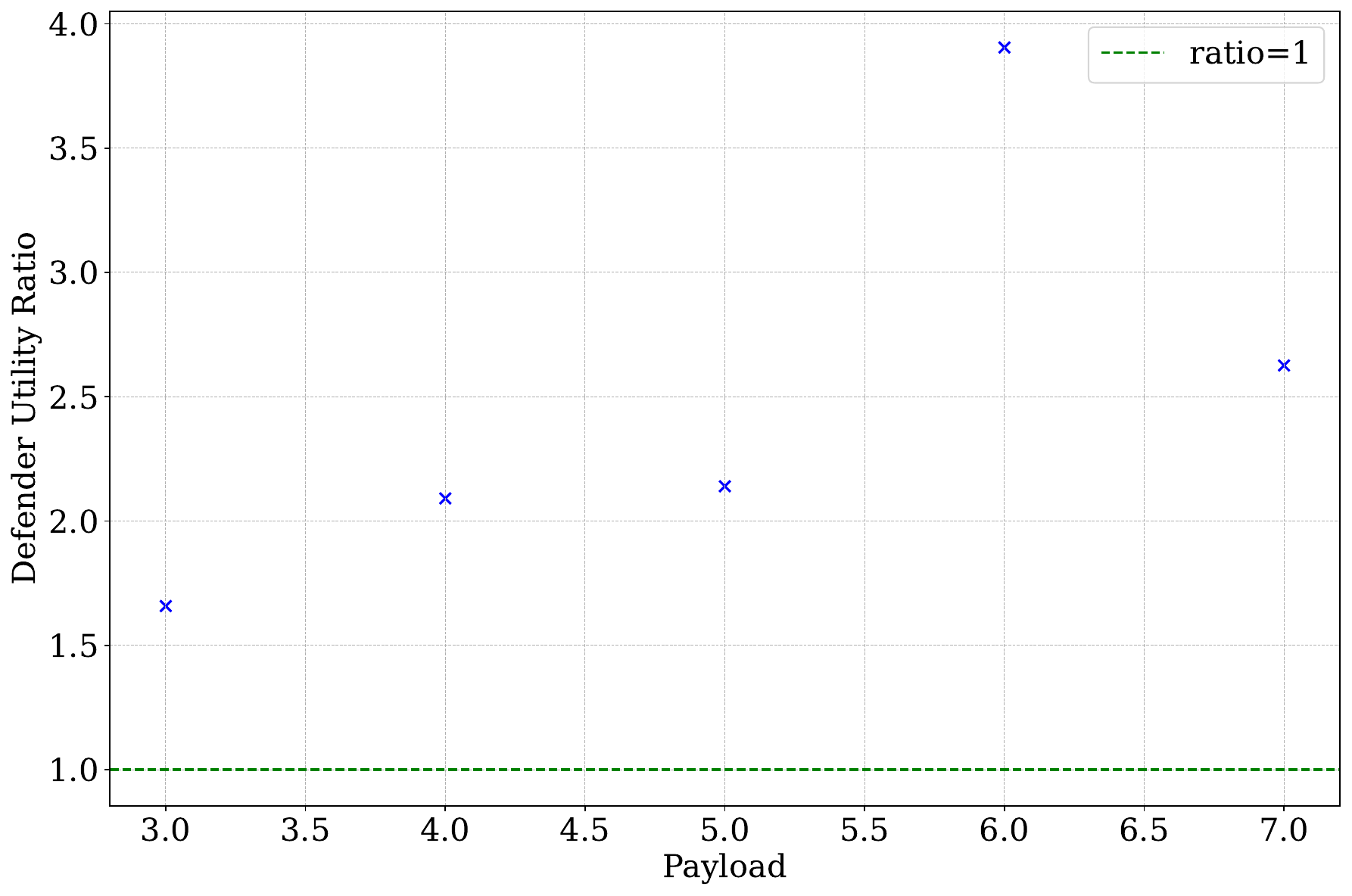} &
         \includegraphics[width=0.48\textwidth]{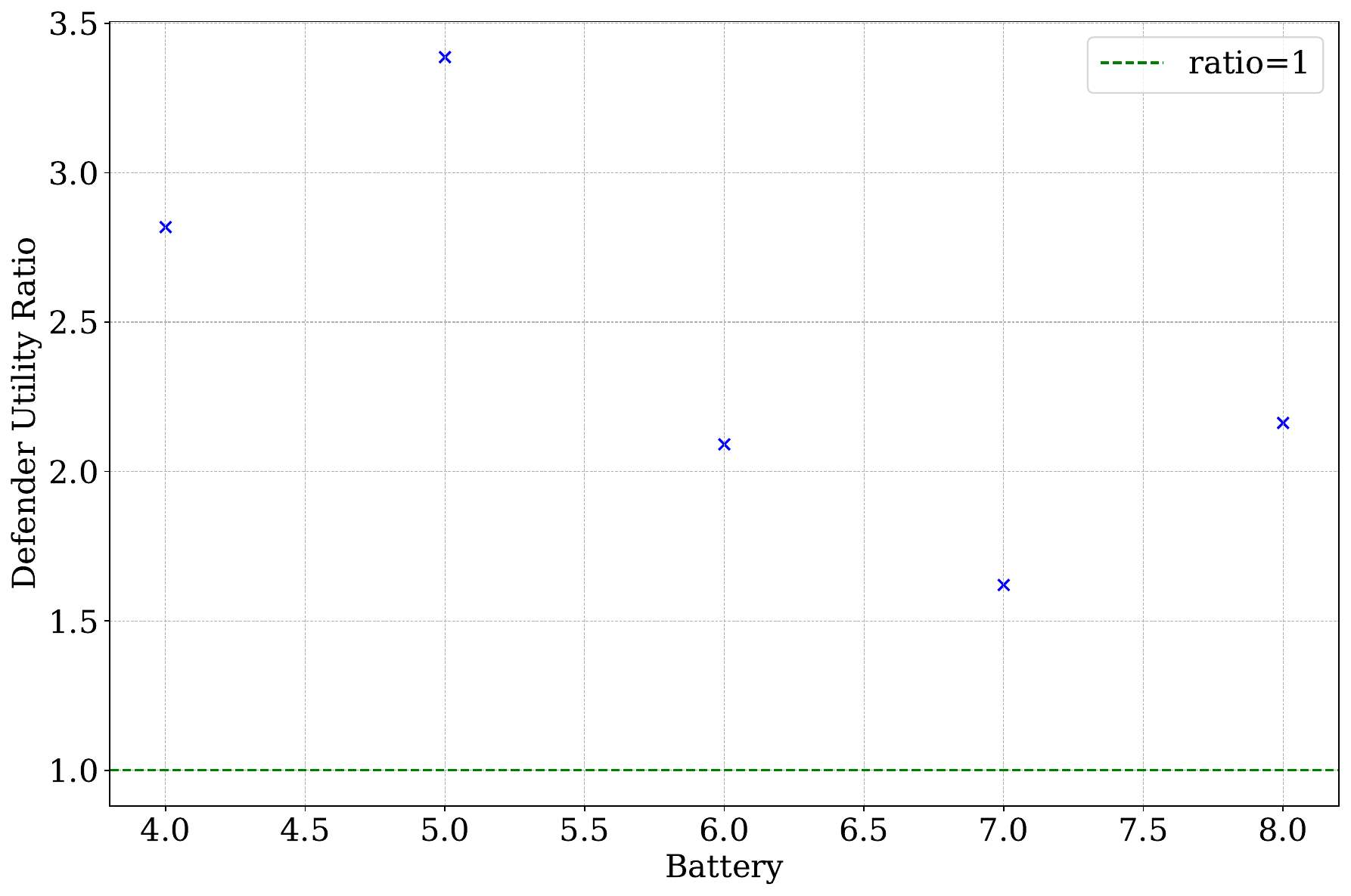} \\
         Defender utility ratio &
         Defender utility ratio \\
         vs. payload &
         vs. battery capacity \\
    \end{tabular}
    \caption{Results of a case study for one major city.}
    \label{fig:city_casestudy_figure}
\end{figure}
Not surprisingly, our very simple baseline is significantly faster than the S2D2 algorithm. Nevertheless, 
Figure~\ref{fig:city_casestudy_figure}(a) shows that S2D2 takes a reasonable amount of runtime (about 3--4 minutes) when the number of defensive drones is varied from 4 to 12 --- moreover, the time seems more or less constant. 
When the payload is varied (see Figure~\ref{fig:city_casestudy_figure}(b)), we see that S2D2's runtime increases linearly --- but still stays in a matter of minutes (13 minutes when payload is 7). 
When the battery capacity is varied (see Figure~\ref{fig:city_casestudy_figure}(c)), S2D2's runtime increases exponentially --- when battery capacity is 7 or less, it takes about 16-17 minutes, but once it goes up to 8, the runtime increases significantly to about 66 minutes. This is not surprisingly because the number of possible paths grows exponentially with the battery capacity --- as battery capacity increases, the attacker can travel further.

Figure~\ref{fig:city_casestudy_figure}(d) shows that when the number of defensive drones is varied from 4 to 12, S2D2 delivers 2 to 4 times the utility provided by the baseline. However, there is no consistent increase in this ratio. While both S2D2 and baseline defender utilities are expected to increase when the number of defenders goes up, we do not see a reason for the ratio to increase nor decrease.
Figure~\ref{fig:city_casestudy_figure}(e) shows that when the payload is varied from 3 to 7, S2D2 delivers 1.7 to 2.7 times the utility provided by the baseline. Again, there is some fluctuation in the ratio. Finally, Figure~\ref{fig:city_casestudy_figure}(f) shows that when the battery capacity of the drones is varied from 4 to 8, S2D2 delivers 1.7 to 4 times the utility provided by the baseline. 
All of these numbers suggest that S2D2 is a definitive improvement over the baseline as far as defender utility is concerned --- this comes at the cost of runtime, even though the latter is still reasonable.

Figure~\ref{fig:histogram_plots_case_study} shows what percentage of the nodes were destroyed for each of the 5 utility values, with $B=6$, $P=4$, $\textsf{ADR}=1$, and $D=4$.
\begin{figure}[h!t]  
    \centering
    \begin{tabular}{ccc}
         \includegraphics[width=.31\textwidth]{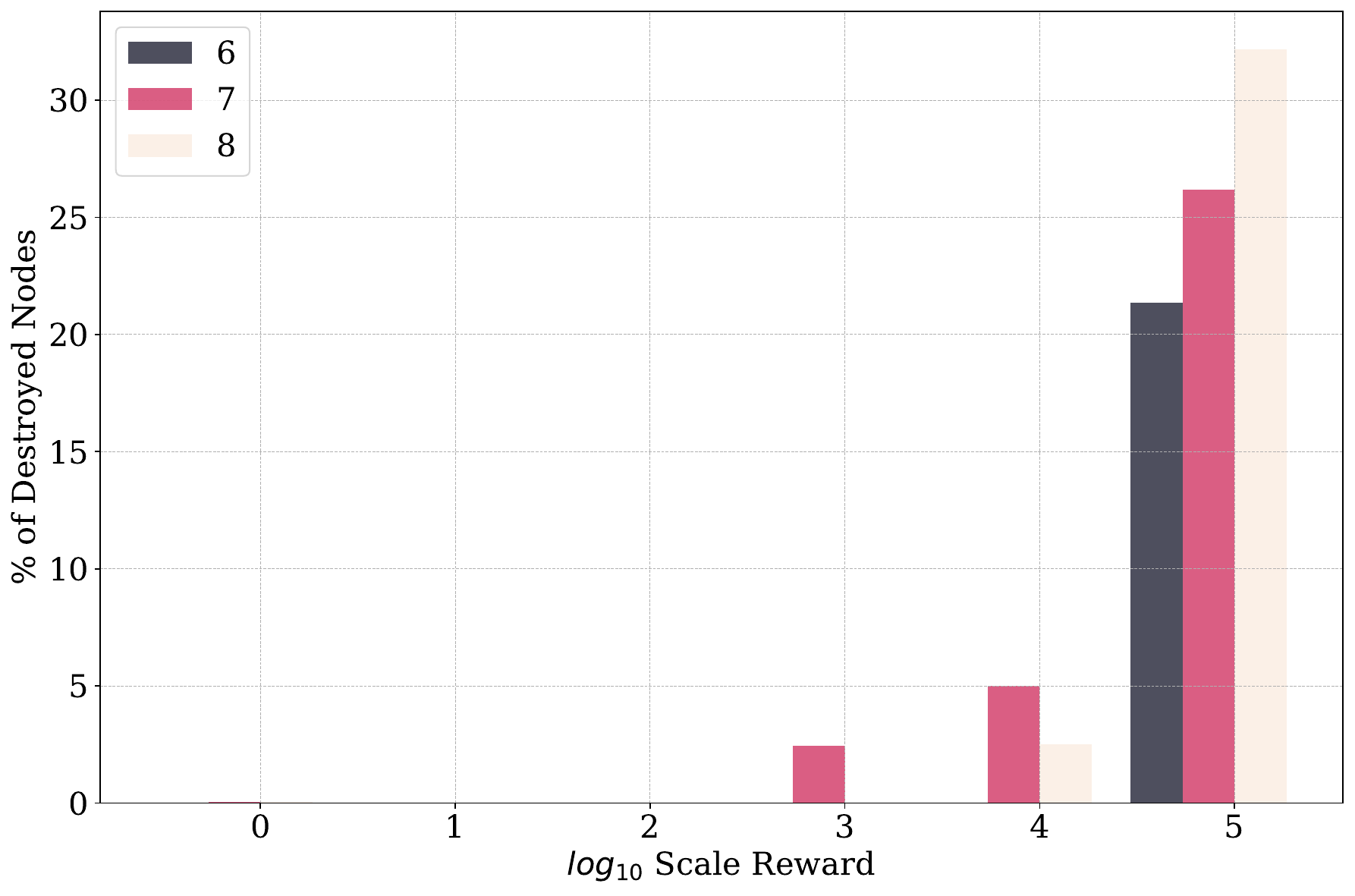} &  
         \includegraphics[width=.31\textwidth]{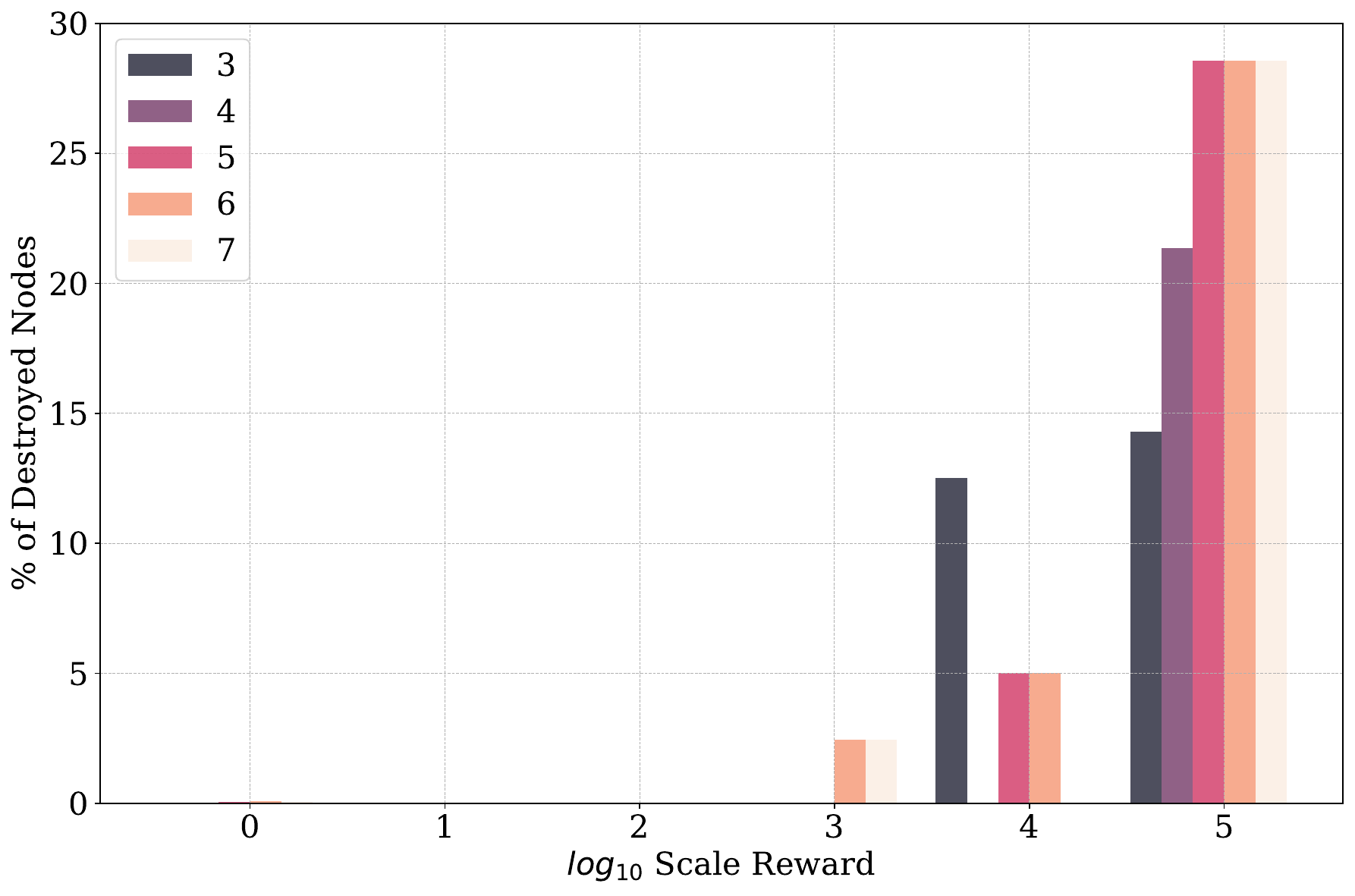} &
         \includegraphics[width=.31\textwidth]{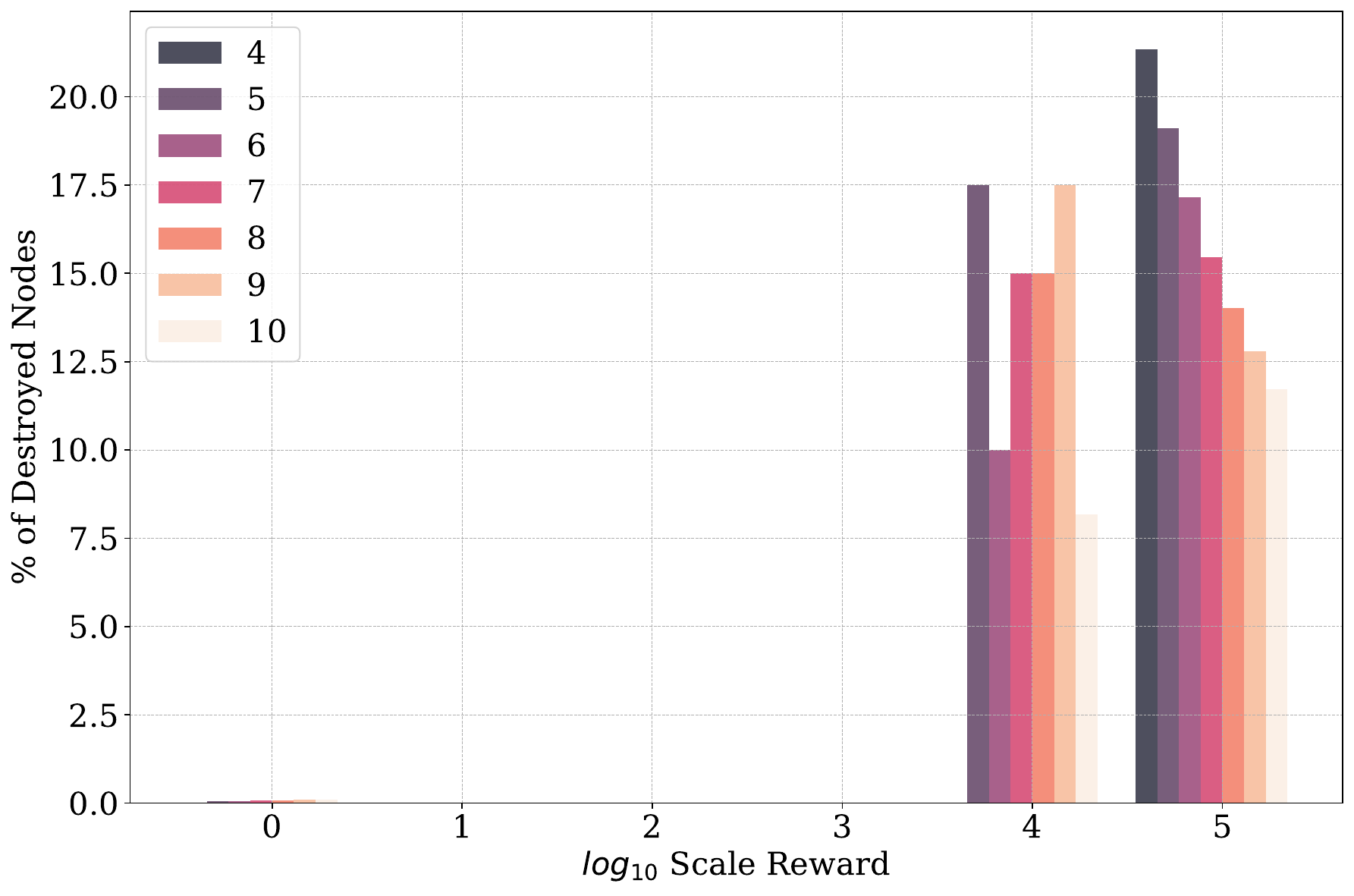}
         \\
         (a) Varying battery& 
         (b) Varying payload &
         (c) Varying number\\
         capacity& 
         &
         of defender drones \\
    \end{tabular}
    \caption{Percentage of destroyed nodes for each utility value.}
    \label{fig:histogram_plots_case_study}
\end{figure}
Specifically, Figure~\ref{fig:histogram_plots_case_study}(a) looks at what happens when varying battery capacity. We observe that as battery capacity increases, the percentage of destroyed nodes having utility 5 nodes increases --- this is probably because the attacker has enough battery to attack multiple 5-ranked nodes with a single drone. 
Figure~\ref{fig:histogram_plots_case_study}(b) shows what happens when we vary the drones' payload. Again, we see that as payload increases, the percentage of destroyed nodes with utility 5 increases --- as the attacker has more payload, it may prefer to attack more of those nodes. The saturation at payload of 5 is probably because there are no more 5-ranked nodes that a single drone can reach with its battery capacity. 
Figure~\ref{fig:histogram_plots_case_study}(c) shows the situation when we vary the number of defender drones. As expected, we see that as the number increases, the percentage of nodes with utility 5 being destroyed decreases. This demonstrates that S2D2 utilizes each added defender drone to cover more of the high ranked nodes. We see a similar correlation for utility 4 nodes, but it is weaker. This can be explained as a side effect where the attacker, taking into consideration that the utility 5 nodes are more protected, now prefers to strike against utility 4 nodes where it is less likely to get caught.

\section{\revision{Discussion}}
\label{sec:discussion}
\revision{\paragraph{Coarsening} The first step of S2D2 is to  coarsen the input graph. Specifically, the neighborhoods should be densely connected inside and relatively isolated from one another. 
However, in urban environments, target locations may not exhibit this kind of clean structure.
First, we emphasize that the output coarsening is not restricted to respect jurisdictional divisions or any type of man-made partition of the city. Second, our intuition was that important areas of interest typically come in clusters, e.g., a dense neighborhood, an industrial area, or a group of government buildings, and so, we expect a good coarsening to exist. For instance, Wall Street in New York City is densely clustered. Likewise, the major government buildings in Washington DC are also densely clustered. We then tested our intuition via two approaches. First, we analyzed real-world large-scale cities, annotated according to the knowledge of security experts. Second, we synthetically annotated nodes with respect to two distributions. In both cases, we observed that the output coarsening was not always ideal. Nevertheless, it did effectively separate the city into neighborhoods and in most of the cities, our experimental results where reasonable even when the coarsening was not ideal. From both theoretical and experimental perspectives, we identify the coarsening approach to be promising, and believe that improvements in coarsening have the potential to significantly improve defense strategies.

\paragraph{Graph Structure} Aerial drones need not be subject to any ground-based constraints of the underlying city, and can seamlessly reach any location through direct aerial traversal. However, we note that this does not suggest that all targets are fully connected. Cities can be quite large, and it takes time to go from one point to another. Also, the defender is also using drones and is therefore moving at a comparable speed as the attacker.\footnote{\revision{We believe that in cases when the drones move so fast that the targets are fully connected, the problem should be modeled as a non-sequential SSG where the attacker has $A \times P$ resources, and the defender has $D$ resources. For that matter, one may use~\cite{korzhyk2011security} to get an efficient exact solution.}} We also emphasize that our choice of modeling a city as a general graph, enables applicability beyond drone swarm defense. Land-based vehicles are more subject to the topographical structure of the city, and so our approach can be even more effective in this setting. Also, not using a perfect grid allows our model to capture various obstacles, e.g., a skyscraper cluster, or an area that is protected with GPS jamming devices.

\paragraph{Attacker adaptivity} While our model allows the defender to apply an adaptive strategy that changes as a function of the state, the attacker is modeled to be static. This simplification seems to contrast with some pursuit-evasion games of similar type, where both players are adaptive. In real-world scenarios, attackers are likely to observe the state, at least partially, and adapt their strategies in accordance. This may raise concerns about whether the model's validity is compromised by this simplification. To this end, we emphasize that while acquiring aerial drones is relatively easy, tracking drones is still imperfect, and as we cover in the related work, is an independent area of interest. For instance, \cite{deb2025drone} analyzes 8 months of drone flights over The Hague, but it is clear that some drone flights may have been missed due to imperfections in tracking. 
Therefore, we assume that both the attacker and the defender are not aware of the locations of the opposing drones at the beginning. The asymmetry stems from the assumption that after significant damage is caused by an attacker drone, it can be tracked. We believe this assumption to be more realistic compared to prior works, at least in the context of drones that can be small, silent, fast, and may not communicate.

In the case the attacker is aware of the defensive drone's location, it could utilize this additional information to evade the defender more effectively, and therefore is expected to cause more damage. This framework could be an interesting future work. However, if there is concern that the defender utility is significantly smaller, a conclusion could be that more research and effort should be put into preventing the attacker from gaining this information to begin with.

\paragraph{Knowledge of $\xx_d$} In Stackelberg games, the attacker knows the defender's committed mixed strategy and best responds to it. This is often justified by the claim that the attacker can survey the defender's strategies before launching an attack. In this paper, we follow this approach. However, this assumption may seem less feasible in the context of state-dependent strategies, as the defender's strategy space is overwhelmingly large. Without sufficient real attacks, many states may not even occur, making it harder for an attacker to accurately observe and infer the defender's strategies. While we acknowledge that surveillance in our case is probably not sufficient for the attacker to unravel the defender's mixed strategy, assuming a stronger attacker may result in a more robust system. The alternative approach, of attempting to model the limitation of an attacker, come at a high risk. If the defense relies on an attacker with certain capabilities which do not hold in reality, it can lead to severe consequences. In contrast, overestimating the attacker capabilities could result in a less effective defense overall, but the caused damage will not exceed model expectations. Indeed, we point out that at least the division of the defender into neighborhoods, and the allocation strategy into neighborhoods, can be observed. The single-drone defensive strategy within a neighborhood is then already more compact. If S2D2 is the chosen implementation, the attacker may be able to compute it by itself. Additionally, the attacker can gain information in other manners. For instance, it may conduct a cyber attack or procure a captured defense drone to get the defender's code. It may also get human intelligence from people who work at the companies that manufactures the drones, or from an employee who developed the defense mechanism for the drones. With all of these considerations in mind, we opted to assume full knowledge of the adversary with regards to the mixed defender strategy.

\paragraph{Zero-sum games} Another alternative approach would be to assume the game is zero sum. In this case, Nash equilibrium is sufficient, and one need not assume knowledge of the defender's mixed strategy. On the other hand, this reduces the model's generality compared to our proposed general-sum setup. This is crucial, as in many cases, the objectives of the attacker and defender may not be aligned. For instance, the attacker may care more about damaging infrastructure, while the defender may care more about civilian casualties, or vice versa. In such scenarios, either a zero-sum based model will not be deployed at all, or alternatively, one would approximate it to be such. Similarly to the previous point, this may result in an unrealistic model of the adversary, and in turn could cause damage that exceeds model expectations.

\paragraph{Scalability} Since the defender's strategy space appears extremely large, questions about scalability may arise. While we claim the runtime to be polynomial with $|\calS^d|$, the strategy space itself is exponential with respect to the input problem size. In this paper, we took a three step approach to address the problem at hand. (\romannumeral 1) First, we provided theoretical analysis and explored a theoretically proven algorithm which is inevitably impractical. (\romannumeral 2) After identifying the bottlenecks, we introduced heuristics to make the algorithm practical. These include narrowing down the strategy space $|\calS^d|$ of the defender, as well as relaxing the conditions for our coarsening. (\romannumeral 3) Finally, we extensively tested the performance of the heuristic algorithm. We acknowledge that S2D2 can be improved both in scale and performance in future works. However, we believe that following the blueprint outlined above is a vital cornerstone. Therefore, down the line, our theoretical results may turn out to be of greater importance than any of the three building blocks of S2D2(coarsening, single drone sequential sub-games, multi-drone meta-game) as instantiated in this work.
}

\section{Conclusions}\label{sec:conclusions}
Multi-drone strikes are increasingly likely to be used to target cities. The threat actors using such techniques will include both nation state actors and terrorist groups.

In this paper, we have developed a realistic model to defend cities against multi-drone attacks via 4 contributions: (i) We extend sequential SSGs involving multiple attack/defense
drones with payload and battery constraints. (ii) We propose the Sequential Stackelberg Drone Defense (S2D2) paradigm to solve the problem of minimizing damage to the city by the attacker and show detailed theoretical results that show that under some conditions related to a novel concept called $\delta$-coarsening, S2D2 provides a strong approximation algorithm for a computationally difficult problem. We prove that
S2D2 outputs an approximate SSE and an upper bound on the error, under such conditions.
(iii) Experiments on a dataset of 80 famous cities compare S2D2 with a heuristic swarm-defense algorithm, demonstrating a trade-off between runtime and defender utility. Importantly, the experiments show that even when the $\delta$-coarsening conditions do not hold, S2D2 still works effectively.
(iv) Our experiments were run with real data on 80 cities using randomly assigned utilities. But in addition, we assigned utilities to locations in 6 cities using general guidelines provided by security experts from the US, EU, Asia, and Middle East. Our experiments also included these 6 cities. Finally, we did a detailed case study of one large North American city with expert input.
To the best of our knowledge, past work on game-theoretic defenses of cities against multi-drone attacks have not done that.

\bibliographystyle{elsarticle-num} 
\bibliography{references}

\end{document}